\documentclass[a4paper, 12pt]{article}

\usepackage[left=1.25in, right=1.25in, top=1.25in, bottom=1.25in]{geometry} 

\usepackage[dvipsnames]{xcolor}
\usepackage[doublespacing]{setspace}
\usepackage{ragged2e} 
\usepackage[hang,flushmargin]{footmisc} 

\usepackage{amsmath, amssymb, amsthm, amsfonts} 
\usepackage{mathrsfs}
\usepackage{mathtools}
\usepackage[shortlabels]{enumitem} 
\usepackage{bm, bbm, upgreek}
\usepackage{nicefrac}
\usepackage{stackengine}
\usepackage{thmtools}

\DeclareFontFamily{U}{mathx}{\hyphenchar\font45}
\DeclareFontShape{U}{mathx}{m}{n}{
      <5> <6> <7> <8> <9> <10>
      <10.95> <12> <14.4> <17.28> <20.74> <24.88>
      mathx10
      }{}
\DeclareSymbolFont{mathx}{U}{mathx}{m}{n}
\DeclareFontSubstitution{U}{mathx}{m}{n}
\DeclareMathAccent{\widecheck}{0}{mathx}{"71}

\usepackage{graphicx} 
\usepackage{array}
\usepackage{threeparttable}
\usepackage{booktabs} 
\usepackage{placeins}
\usepackage{multirow} 
\usepackage{subcaption}
\usepackage{pdfpages}
\usepackage{rotating}
\usepackage{subcaption}

\usepackage[longnamesfirst]{natbib} 
\usepackage{xr-hyper}
\usepackage{hyperref} 
\usepackage{cleveref}

\hypersetup{
  colorlinks = true, 
  linkcolor={red!50!black},
	citecolor={blue!50!black},
	urlcolor={red!50!black},
  filecolor = {red!50!black}
}

\usepackage{titling}

\setlist{noitemsep}
\setlist[1]{labelindent=\parindent}
\setlist[itemize]{leftmargin=*}
\setlist[itemize,1]{label=$\triangleright$}
\setlist[enumerate,1]{label = (\roman*),ref   = (\roman*)}

\declaretheoremstyle[
    spaceabove=2ex,
    spacebelow=2ex,
    headfont=\normalfont\scshape, 
    bodyfont=\normalfont\itshape, 
    postheadspace=.5ex,
    notebraces={(}{)},
    headpunct={:}%
]{thmstyle}

\declaretheorem[name={Assumption}, style=thmstyle]{assumption}
\declaretheorem[name={Theorem}, style=thmstyle]{theorem}
\declaretheorem[name={Proposition}, style=thmstyle]{proposition}
\declaretheorem[name={Corollary}, style=thmstyle]{corollary}
\declaretheorem[name={Lemma}, style=thmstyle]{lemma}

\declaretheorem[name={Example}, style=thmstyle]{example}
\declaretheorem[name={Remark}, style=thmstyle]{remark}
\declaretheorem[name={Definition}, style=thmstyle]{definition}

\declaretheoremstyle[
    spaceabove=1ex,
    spacebelow=3ex,
    headfont=\normalfont\scshape, 
    bodyfont=\normalfont, 
    postheadspace=.5ex,
    qed=$\blacksquare$,
    headpunct={:}%
]{proofstyle}

\declaretheoremstyle[
    spaceabove=2ex,
    spacebelow=2ex,
    headfont=\normalfont\scshape, 
    bodyfont=\normalfont\itshape, 
    postheadspace=.5ex,
    headpunct={:}%
]{algorithmstyle}

\newcommand{\notes}[1]{\justify{\footnotesize{\textit{Notes:} #1}}}

\renewenvironment{abstract}
 {\small
  \begin{center}
  \bfseries \abstractname\vspace{-.5em}\vspace{0pt}
  \end{center}
  \list{}{
    \setlength{\leftmargin}{0cm}%
    \setlength{\rightmargin}{\leftmargin}%
  }%
  \item\relax}
 {\endlist}

\newcommand{\deriv}[2]{\frac{\mathrm{d}#1}{\mathrm{d}#2}}
\newcommand{\pderiv}[2]{\frac{\partial #1}{\partial #2}}

\DeclareMathOperator*{\argmin}{arg\,min}

\newcommand{\dmu}{\,\mathrm{d}\mu}
\newcommand{\dnu}{\,\mathrm{d}\nu}
\newcommand{\dlambda}{\,\mathrm{d}\lambda}
\newcommand{\dP}{\,\mathrm{d}P}
\newcommand{\darg}[1]{\,\mathrm{d}#1}

\newcommand{\E}{\mathbb{E}}
\newcommand{\Var}{\mathrm{Var}}
\newcommand{\Cov}{\mathrm{Cov}}
\newcommand{\weakconv}{\rightsquigarrow}

\newcommand{\R}{\mathbb{R}}
\newcommand{\N}{\mathbb{N}}
\newcommand{\Z}{\mathbb{Z}}

\newcommand{\F}{\mathcal{F}}   
\newcommand{\ms}[1]{\mathscr{#1}}

\newcommand{\Span}{\operatorname{Span}}
\newcommand{\rank}{\operatorname{rank}}
\newcommand{\IP}[2]{\left\langle {#1}\, ,\, {#2} \right\rangle}

\newcommand{\ran}{\operatorname{ran}}

\newcommand{\cl}{\operatorname{cl}}

\newcommand{\define}{\coloneqq}

\newcommand{\lin}{\operatorname{lin}}
\newcommand{\cllin}{\overline{\operatorname{lin}}\ }

\renewcommand{\P}{\mathrm{P}}
\newcommand{\bbP}{\mathbb{P}}
\newcommand{\EP}{\mathbb{P}_n}
\newcommand{\mc}[1]{\mathcal{#1}}
\newcommand{\mb}[1]{\mathbb{#1}}

\newcommand{\limn}{\lim_{n\to\infty}}
\newcommand{\limm}{\lim_{m\to\infty}}

\newcommand{\limsupn}{\limsup_{n\to\infty}}
\newcommand{\liminfn}{\liminf_{n\to\infty}}

\newcommand{\transp}{^\prime}

\newcommand{\fracn}{\frac{1}{n}}
\newcommand{\G}{\mathbb{G}}

\newcommand{\fracrootn}{\frac{1}{\sqrt{n}}}
\newcommand{\sumin}{\sum_{i=1}^n}

\newcommand{\meanin}{\frac{1}{n}\sumin}

\newcommand{\dotscr}{\dot{\ell}_{\gamma}}
\newcommand{\dotscrarg}[1]{\dot{\ell}_{#1}}

\newcommand{\effscr}{\tilde{\ell}_{\gamma}}

\newcommand{\effscrarg}[1]{\tilde{\ell}_{#1}}

\newcommand{\effinfo}{\tilde{\mc{I}}_{\gamma}}

\newcommand{\effinfoarg}[1]{\tilde{\mc{I}}_{#1}}

\newcommand{\lcontig}{\triangleleft\,}

\newcommand{\mutcontig}{\triangleleft \, \triangleright\; }
\newcommand{\mcontig}{\mutcontig}
\newcommand{\quotient}[2]{{#1} \,/\, {#2}}

\newcommand\independent{\protect\mathpalette{\protect\independenT}{\perp}}
\def\independenT#1#2{\mathrel{\rlap{$#1#2$}\mkern2mu{#1#2}}}

\begin{document}


\title{\sc \Large 
Locally regular and efficient tests in non-regular semiparametric models
}

\author{{Adam Lee\thanks{BI Norwegian Business School, \href{mailto:adam.lee@bi.no}{adam.lee@bi.no}.
Previous versions of this paper were titled ``Robust and Efficient Inference For Non-Regular Semiparametric Models''.
I have benefitted from discussions with and comments / questions from Majid Al-Sadoon, 
 Isiah Andrews,
 Christian Brownlees, Bjarni G. Einarsson, Juan Carlos Escanciano, Kirill Evdokimov, Lukas Hoesch,
Geert Mesters,
Vladislav Morozov, Jonas Moss, 
Whitney Newey, 
Katerina Petrova, Francesco Ravazzolo, Barbara Rossi, Andr\'{e} B. M. Souza, Emil Aas Stoltenberg, Philipp Tiozzo and participants at various conferences and seminars.
All errors are my own.}}
	}
\date{\today 
}

\maketitle
\thispagestyle{empty}

\begin{abstract}
\onehalfspacing
\noindent

This paper considers hypothesis testing in semiparametric models which may be non-regular. 
I show that C($\alpha$) style tests are locally regular under mild conditions, including in cases where locally regular estimators do not exist, such as models which are (semiparametrically) weakly identified.
I characterise the appropriate limit experiment in which to study local (asymptotic) optimality of tests in the non-regular case and generalise classical power bounds to this case. I give conditions under which these power bounds are attained by the proposed C($\alpha$) style tests. The application of the theory to a single index model and an instrumental variables model is worked out in detail.

\noindent 



\end{abstract}

\clearpage

\onehalfspacing
\setcounter{page}{1}
\section{Introduction}\label{sec:intro}

It is often considered desirable that estimators are ``locally regular'' in that they exhibit the same limiting behaviour under the true parameter as they do 
 under sequences of ``local alternatives'' which cannot be consistently distinguished from the true parameter.\footnote{Precise definitions will be given below. See \cite{BKRW98, vdV98}, for example, for textbook treatments.} 
Unfortunately, there are many 
models in which locally regular estimators do not exist.\footnote{See e.g. \cite{C86, C92,  N90, RB90} for some examples.} One necessary condition is given by \cite{C86}: if the efficient information for a scalar parameter is 0, then no locally regular estimator of that parameter exists. Similarly, singularity of the efficient information matrix implies the non-existence of locally regular estimators of Euclidean parameters. Models in which this may occur are called ``non-regular''. Many widely used models are non-regular (at least at certain parameter values): examples include single index models, instrumental variables, errors-in-variables, mixed proportional hazards, discrete choice models and sample selection models.

In this paper, I demonstrate that locally regular \emph{tests} exist in a broad class of non-regular models, despite the non-existence of locally regular estimators. In particular, I show that a class of tests based on the C$(\alpha)$ idea of \cite{N59, N79} are locally regular.

These tests are based on a quadratic form of moment conditions evaluated under the null hypothesis. The key [C($\alpha$)] idea which ensures the local regularity is that the moment conditions must be (asymptotically) orthogonal to the collection of score functions for all nuisance parameters. Such moment conditions can always be constructed from any initial moment conditions by an orthogonal projection.

A key advantage of these C($\alpha$) tests is that they do not (asymptotically) overreject under (semiparametric) weak identification asymptotics, i.e. under local alternatives to a point of identification failure.\footnote{The semiparametric weak identification asymptotics used are those of \cite{Kaji21} (see also \citealp{AM22}), suitably generalised to permit non-i.i.d models. 
} The local regularity of these tests ensures that if the test is asymptotically of level $\alpha$ under any fixed parameter consistent with the null, it is also asymptotically of level $\alpha$ under any sequence of local alternatives consistent with the null, i.e. under (semiparametric) weak identification asymptotics. In addition to the well-studied case where weak identification stems from potential identification failure due to a finite dimensional nuisance parameter, the results in this paper also cover the case where identification failure is due to an infinite dimensional nuisance parameter and thus provide a generally applicable approach to weak identification robust inference in semiparametric models.\footnote{
    These C($\alpha$) tests also behave well in other non-standard settings, such as when nuisance functions are estimated under shape constraints; see Section \ref{sm:shape-constraints} for a discussion.
} Even in the case where the identification failure due to a finite dimensional nuisance parameter, the resulting weak identification robust tests appear to be new in the literature.\footnote{For instance, in the case of homoskedastic linear IV, the test that results from the construction in this paper does not coincide with any of the ``usual'' weak instrument robust tests (e.g. AR, LM, K, CLR). Demonstration of this is available from the author.}
 The tests proposed here are derived directly from an asymptotic orthogonality condition. As such they are close in spirit to the identification robust test of \cite{K05} which also requires an orthogonalisation, albeit with respect to different objects and in a different Hilbert space.

Achieving local regularity does \emph{not} come at the expense of (local asymptotic) power. I characterise  power bounds for tests in non-regular models and show 
that the C($\alpha$) tests proposed in this paper acheive these power bounds provided the moment conditions are chosen optimally. These power bounds contain those for regular models as a special case. Moreover, the conditions required for attainment of the power bounds are weaker than those in the literature.\footnote{In particular, in regular models the attainment result is well known if either (a) the observations are i.i.d. \cite[cf.][Chapter 25]{vdV98} or (b) the information operator (as defined in \citealp[][p. 846]{CHS96}) is boundedly invertible \citep{CHS96}. The result in this paper does not require either of these conditions.} 

Following the theoretical development, I give details of its application to two examples: (i) a single index model which may be weakly identified when the link function is too flat and (ii) an instrumental variables (IV) model which may be weakly identified when the (nonparametric) first stage is too close to a constant function. Simulation experiments based on these examples demonstrate that the proposed tests enjoy good finite sample performance. 

The application to IV may also be of interest for empirical researchers concerned about weak instruments. If the instruments are mean independent of the errors, then the test proposed here is robust to weak identification and can be substantially more powerful than tests assuming a linear first stage. This imposes no cost if the true first stage is (approximately) linear:  the power of the proposed test is comparable to optimal tests based on a linear first stage. The practical use of these tests is demonstrated in two IV applications with possibly weak instruments.

This paper is connected to three main strands of the literature: the first is that concerned with general results on estimation and testing in semiparametric models. Much of this is now textbook material: see e.g. \cite{N90, CHS96, BKRW98, vdV98}. The second is the literature on C($\alpha$) tests. These were introduced by \cite{N59, N79} and have seen many useful applications, most recently as a way to handle machine learning or otherwise high dimensional first steps \citep[see e.g.][]{CHS15,BEvK20, CEINR22}. In this paper, the same structure which ensures good performance in such settings is used for a different purpose -- to construct tests which remain robust in non-regular settings. Lastly, the literature on robust testing in non -- regular or otherwise non -- standard settings is closely related to this paper \cite[e.g.][]{AG09,RS12,EMW15, McC17}. In particular, the  locally regular tests derived in this paper are especially useful in cases of weak identification and therefore this paper is closely related to the literature on weak identification robust inference \citep[e.g.][]{SS97, D97, SW00, K05, AC12, AM15, AM16b}. More specifically, this paper is most closely related to the recent work on semiparametric weak identification \citep{Kaji21, AM22} and extends the notion of semiparametric weak identification considered therein to non -- i.i.d. models.\footnote{Failure of local identification and singularity of the information matrix are closely linked in parametric models, see \cite{R71}. In the semiparametric case, parameters may be identified but nevertheless have a singular efficient information matrix. The relationship between the efficient information matrix and identification is considered by \cite{E22}.}

\section{Locally regular testing}\label{sec:heuristic}

\subsection{The local setup}

The goal considered throughout this paper is to construct hypothesis tests of $\mathrm{H}_0: \theta = \theta_0$ against $\mathrm{H}_1: \theta \neq \theta_0$ in the sequence of models $\mc{P}_n = \{P_{n, \gamma}: \gamma\in \Gamma\}$ where $\gamma = (\theta, \eta)\in \Gamma = \Theta\times\mc{H}$ for some open $\Theta\subset \R^{d_\theta}$ and $\mc{H}$ an arbitrary set. Each $\mc{P}_n$ consists of probability measures on a measurable space $(\mc{W}_n, \mc{B}(\mc{W}_n))$ and is dominated by a $\sigma$-finite measure $\nu_n$.\footnote{Typically the index $n$ is sample size and $\mc{W}_n$ is the space in which a sample of size $n$ takes its values. This is the situation considered in Section \ref{ssec:smooth-iid} as well as in the examples in Section \ref{sec:examples}. 
}

Let $H_{\gamma}=\R^{d_\theta}\times B_{\gamma}$ be a subset of a linear space containing 0, and suppose that $\{P_{n, \gamma, h}: h\in H_{\gamma}\}\subset \mc{P}_n$ are such that $P_{n, \gamma} = P_{n, \gamma, 0}$. Elements of $H_\gamma$ will be written as $h = (\tau, b) \in \R^{d_\theta}\times B_{\gamma}$.\footnote{In most examples, $H_\gamma$ will be a linear space. The more general situation as considered here is nevertheless important to allow for, for example, Euclidean nuisance parameters subject to boundary constraints. In such a setting, if the constraint is binding at $\gamma$, then $\gamma$ can only be perturbed in certain directions if $P_{n, \gamma, h}$ is to remain within the model.} The measures $P_{n, \gamma, h}$ should be viewed as local perturbations of the measure $P_{n, \gamma}$ in a ``direction'' $h\in H_\gamma$. These local perturbations can be split in two groups: the perturbations $P_{n, \gamma, h}$ with  $h\in H_{\gamma, 0}\define \{(0, b): b\in B_{\gamma}\}$ correspond to the null hypothesis $\mathrm{H}_0: \theta = \theta_0$ and those with $h\in H_{\gamma, 1}\define \{h = (\tau, b): 0 \neq \tau\in \R^{d_\theta}, b\in B_{\gamma}\}$ to the alternative $\mathrm{H}_1: \theta \neq \theta_0$. 
As such, $P_{n, \gamma, h}$ for $h\in H_{\gamma, 0}$ will be referred to as \emph{local perturbations consistent with the null hypothesis}, whilst $P_{n, \gamma, h}$ for $h\in H_{\gamma, 1}$ are \emph{local alternatives}. 
The subsequent analysis is local with the parameter $\gamma$ being considered fixed at a $\gamma$ consistent with $\mathrm{H}_0$. As such, to lighten the notation, dependence on $\gamma$ will be mostly left implicit: I write  $P_{n, h}$ for $P_{n, \gamma, h}$, $H$ for $H_\gamma$, $H_{i}$ for $H_{\gamma, i}$ ($i=0, 1$) and similarly for other objects. I also use the abbreviation $P_n\define P_{n, 0}$.

I use the single-index model as a running example throughout the paper.\footnote{Technical details for this example are deferred to Sections \ref{ssec:sim} and \ref{ssec:example-details-sim}.}

\begin{example}[Single-index model]\label{ex:SIM-running-example}
    Suppose that the researcher observes $n$ i.i.d. copies of $W = (Y, X_1, X_2) \in \R^{2 + K}$ where
    \begin{equation}\label{eq:SIM-running-example-mdl}
        Y = f(X_1 + X_2\transp\theta) + \epsilon, \qquad \E[\epsilon  | X] = 0,
    \end{equation}
    and where $f$ belongs to some set of continuously differentiable functions $\ms{F}$. The description of the model is completed by $\zeta \in \ms{Z}$, the density function of $(\epsilon, X)$ with respect to some $\sigma$-finite measure.
    The model $\mc{P}_n$ consists of the product measures $P_n = P^n$ where $P$ is the probability measure corresponding to the density
    \begin{equation}\label{eq:SIM-running-example-dens}
        p(W) =  p_{\gamma}(W) \define \zeta(\epsilon_{f, \theta},\, X), \qquad \epsilon_{f, \theta}\define Y - f(V_\theta),\quad  V_\theta \define X_1 + X_2\transp \theta,
    \end{equation}
    for a $\gamma = (\theta, f, \zeta)\in \Theta \times \ms{F} \times \ms{Z} = \Gamma$. 
    A class of local perturbations to this model are the probability measures $P_{n, h} = P_{h}^n$ where $P_h$ has density $p_{\gamma + \varphi_n(h)}$ with
    \begin{equation}\label{eq:SIM-local-alt}
             \varphi_n(h) = \left(\tau,\, b_1,\, b_2\zeta \right) / \sqrt{n},\qquad h = (\tau, (b_1, b_2)) \in H\define \R^{d_\theta} \times (B_1 \times B_2),
    \end{equation}
    where $B_1$ is a subset of the  bounded, continuously differentiable functions with bounded derivative and $B_2$ is a subset of the bounded functions $b_2:\R^{1+K}\to \R$, continuously differentiable in the first argument with bounded derivative.
\end{example}

\subsection{Local asymptotic normality}

The key technical condition under which the theory is developed is local asymptotic normality (LAN; see e.g. \citealp[Chapter 7]{vdV98} or \citealp[Chapter 6]{LCY00}).
Define the log-likelihood ratios
\begin{equation}\label{eq:LLR}
    L_{n}(h) \define \log \frac{p_{n, h}}{p_{n, 0}}, \qquad \text{ where } \ p_{n,  h}\define \deriv{P_{n, h}}{\nu_n}, \text{ for } h\in H.
\end{equation}

\begin{assumption}[LAN]\label{ass:LAN}
    For bounded linear maps $\Delta_{n}:\cllin H_{}\to L_2^0(P_{n})$, 
    \begin{equation}\label{eq:LAN}
        L_{n}(h) = \Delta_{n}h - \frac{1}{2} \|\Delta_{n}h\|^2 + R_{n}(h), \qquad h\in H
    \end{equation}
    with  $R_{n}(h) \xrightarrow{P_{n}} 0$ for all $h\in H$. Additionally, for each $h\in H$, the law of 
    $\Delta_{n}h$ converges to $\mc{N}(0,  \sigma(h))$ in the Mallows-2 metric, $d_2$.
\end{assumption}

The requirement that $\Delta_{n}h$ converges in $d_2$ is equivalent to requiring that it converges weakly and $(\Delta_{n}h)_{n\in \N}$ is uniformly square $P_{n}$-integrable \cite[e.g.][Appendix A.6]{BKRW98}. This implies that $\sigma(h)= \limn \|\Delta_{n}h\|^2$.

\begin{remark}\label{rem:mutual-contiguity}
    Assumption \ref{ass:LAN} ensures that the sequences $(P_{n})_{n\in \N}$ and $(P_{n, h})_{n\in \N}$ are mutually contiguous for any $h\in H$ \cite[see e.g.][Example 6.5]{vdV98}.
\end{remark}

\begin{remark}\label{rem:LAN-vs-ULAN}
    If $H$ is (pseudo-)metrised one may consider a uniform version of Assumption \ref{ass:LAN}, i.e. uniform local asymptotic normality (ULAN). Such a version is given in Assumption \ref{ass:ULAN} and is equivalent to Assumption \ref{ass:LAN} plus asymptotic equicontinuity on compact sets of $h\mapsto \Delta_{n} h$ (in $L_2(P_{n, 0})$) and $h\mapsto P_{n, h}$ (in total variation) (Proposition \ref{prop:ULAN-LAN-equicontinuity}). The latter equicontinuity condition is of interest regarding
      local uniformity of size control; cf. Corollary \ref{cor:psi-locally-uniformly-regular} and Lemma \ref{lem:level-alpha-unif-equicontinuity-TV} below.
\end{remark}

\begin{example}[
    continues=ex:SIM-running-example]
    Under regularity conditions, the single-index model satisfies Assumption \ref{ass:LAN} with 
    \begin{equation}\label{eq:SIM-running-example-score-op}
        \Delta_nh \define \frac{1}{\sqrt{n}}\sumin \tau\transp \dotscr(W_i) + [Db](W_i),
    \end{equation}
    where for $\phi(e, x)\define \pderiv{\log \zeta(e, x)}{e}$,
    \begin{equation}\label{eq:SIM-running-example-scores}
        \dotscr(W) \define -\phi(\epsilon_{f, \theta}, X) f'(V_\theta)X_2, \quad 
        [D b](W) \define -\phi(\epsilon_{f, \theta}, X)b_{1}(V_{\theta}) + b_2(\epsilon_{f, \theta}, X).
\end{equation}
\end{example}

\subsection{Local regularity for tests}

\begin{definition}\label{defn:locally-regular-test}
    A sequence of tests $\phi_n:\mc{W}_n\to [0, 1]$ of the hypothesis $\mathrm{H}_0:\theta = \theta_0$ 
     against $\mathrm{H}_1: \theta\neq \theta_0$
      is asymptotically of level $\alpha$ and locally regular if
    \begin{equation}\label{eq:locally-regular-test}
        \uppi_{n}(\tau, b) \define P_{n,  h}\phi_n \to \uppi(\tau), \quad h = (\tau, b)\in H\quad \text{ and }\quad \uppi(0) \le \alpha.
    \end{equation} 
\end{definition}

That is, the finite sample (local) power function of the test, $\uppi_n$ converges under each $P_{n, h}$ to a function $\uppi$ which may depend on $\tau$ (and, implicitly, $\gamma$) but not on $b$, the parameter which describes local deviations from the nuisance parameter $\eta$.\footnote{Cf. the definition of a (locally) regular estimator in \citealp[e.g.][p. 365]{vdV98}.}
If a sequence of tests does \emph{not} satisfy \eqref{eq:locally-regular-test} it is \emph{(locally) non-regular}.

Local regularity of test sequences as in \eqref{eq:locally-regular-test} is a pointwise concept. It is also of interest to consider a uniform version.

\begin{definition}\label{defn:locally-uniformly-regular-test}
    A sequence of tests $\phi_n:\mc{W}_n\to [0, 1]$ of the hypothesis $\mathrm{H}_0:\theta = \theta_0$ 
     against $\mathrm{H}_1: \theta\neq \theta_0$
    is asymptotically of level $\alpha$ and locally uniformly regular on $K\subset H$ if \eqref{eq:locally-regular-test} holds uniformly on $K$.
\end{definition}

If $H$ is a (pseudo-)metric space and $K$ is a compact set, for the convergence in \eqref{eq:locally-regular-test} to hold uniformly on $K$ it is necessary and sufficient to show that the sequence of functions $\uppi_n$ is asymptotically equicontinuous on $K$.\footnote{
The same is true if $K$ is totally bounded. See e.g. \cite{D21}, p. 123, for the definition of asymptotic equicontinuity.\label{ftnt:K-compact-totally-bounded}}

Directly working with the power functions $\uppi_n$ to show their asymptotic equicontinuity is complicated in many cases. It is, however, often possible to show results which imply this property. For instance, the functions $h\mapsto P_{n, h}$ being asymptotically equicontinuous in $d_{TV}$ implies the required asymptotic equicontinuity of the power functions. Despite being (much) stronger, this often holds.\footnote{
    See the discussion following Remark \ref{rem:ULAN-compact-equicontinuity-in-TV} below.
}

\paragraph{Weak identification asymptotics and local regularity}

In many models there are parameter values, $\gamma$, at which locally regular estimators do not exist. Points where the parameter of interest, $\theta$, is un- or under-identified provide an important class of examples.
Moreover, as is well known from the literature on weak identification, even if $\theta$ is identified at $\gamma$, finite sample inference may be poor if $\gamma$ is too close to a point of identification failure relative to the amount of information contained in the sample. Such behaviour has been widely studied in models where the part of $\gamma$ causing the identification failure is finite dimensional \citep[e.g.][]{AC12, AM15}.

There are also many examples where weak identification may occur due to the value of \emph{infinite-dimensional} nuisance parameters. \cite{Kaji21} and \cite{AM22} use a differentiability in quadratic mean (DQM) condition to define semiparametric weak identification asymptotics in i.i.d. models. In particular, they consider sequences $P_{n, h}^n$ which satisfy 
\begin{equation}\label{eq:dqm-weakid}
    \limn \int \left[\sqrt{n}\left(\sqrt{p_{n, h}} - \sqrt{p_{0}} \right) - \frac{1}{2}f\sqrt{p_{0}} \right]^2\dnu_n  = 0
\end{equation}
for a point $P_0$ where the parameter of interest is unidentified.
In the i.i.d. case, \eqref{eq:dqm-weakid} implies the LAN expansion in Assumption \ref{ass:LAN} with $\Delta_nh = \frac{1}{\sqrt{n}}\sumin f(W_i)$ \cite[e.g.][Lemma 25.14]{vdV98}.\footnote{
If Assumption \ref{ass:LAN} holds with $\Delta_n$ having this form, the converse is also true.
} Working with Assumption \ref{ass:LAN} in place of \eqref{eq:dqm-weakid} broadens the applicability of this class of semiparametric weak identification asymptotics to non-i.i.d. models. It is clear from Definition \ref{defn:locally-regular-test} that a locally regular test sequence will have asymptotic null rejection probability (NRP) which does not exceed the nominal level under weak identification asymptotics $P_{n, h}$.\footnote{Of course, a (non-regular) test sequence may have asymptotic NRP which depends on $b$ and yet is bounded by the nominal level under $P_{n, h}$ for all $h = (0, b)\in H_0$ and / or have an asymptotic power function which depends on $b$ for $h = (\tau, b)\in H_1$. Restricting attention to locally regular test sequences may be justified by the power optimality results of Section \ref{sec:theory}.}

I now give two examples of semiparametric models where the parameter of interest $\theta$ may be un- or under-identified depending on the value of an infinite dimensional nuisance parameter.\footnote{A further example is the linear simultaneous equations model in \cite{LM21}.} 
The first is the running example.

\begin{example}[
    continues=ex:SIM-running-example]
    As is clear from the model equation $ Y = f(X_1 + X_2\transp\theta) + \epsilon$, if $f$ is flat, i.e. $f' = 0$, then the parameter $\theta$ is unidentified. The sequences given in \eqref{eq:SIM-local-alt} are weak identification asymptotic sequences if $f'=0$.
\end{example}

\begin{example}[IV]\label{ex:IV}
    Suppose the researcher observes $n$ i.i.d. copies of $W=(Y, X, Z)$, 
    \begin{equation*}
        Y = X\transp \theta  +Z_1\transp\beta + \epsilon, \qquad \E[\epsilon |Z] = 0, \qquad Z = (Z_1\transp, Z_2\transp)\transp.
    \end{equation*}
    If $\pi(Z) \define \E[X|Z]$ is constant, $\theta$ is unidentified; if some components of $\pi(Z)$ are constant, $\theta$ is underidentified.
\end{example}

In Examples \ref{ex:SIM-running-example} and \ref{ex:IV}, at the points of identification failure, no locally regular estimator exists, however locally regular C($\alpha$) tests are developed in Section \ref{sec:examples}.\footnote{These examples consider i.i.d. data for simplicity. See \cite{HLM22} for an example of a locally regular C($\alpha$) test of the form proposed in this paper for the potentially un- / under-identified parameter in a structural vector autoregressive model.}

\subsection{A class of locally regular tests}

To construct locally regular tests of $\mathrm{H}_0: \theta = \theta_0$ against $\mathrm{H}_1: \theta\neq \theta_0$,
I use a generalisation of the class of C($\alpha$) tests introduced by \cite{N59, N79} to characterise optimal tests in regular parametric models. These tests are a based on a quadratic form of (estimators of) a vector of $d_\theta$ moment conditions $g_n\in L_2(P_n)$ which satisfy the following requirements.

\begin{assumption}[Joint convergence]\label{ass:joint-conv}
    For $g_{n}\in L_2(P_{n})^{d_\theta}$ and each $h= (\tau, b)\in H$,
    \begin{equation*}
        \left(\Delta_{n}h,\; g_{n}\transp\right)\transp \overset{P_{n}}{\weakconv }\mc{N}\left(0, \Sigma(h)\right),
    \end{equation*}
    \begin{equation*}
        \Sigma(h) \define \begin{bmatrix}
            \sigma(h) & \tau\transp\Sigma_{ 21}\transp\\
            \Sigma_{ 21}\tau & V
        \end{bmatrix} = \limn \begin{bmatrix}
            \|\Delta_{n}h\|^2 & \IP{\Delta_{n}(\tau, 0)}{g_{n}\transp}\\
            \IP{g_{n}}{\Delta_{n}(\tau, 0)} & \IP{g_{n}}{g_{n}\transp}
        \end{bmatrix}.
    \end{equation*}
\end{assumption}

Built-in to Assumption \ref{ass:joint-conv} is a requirement of asymptotic orthogonality of $g_n$ and the scores for the nuisance parameters $\eta$. This generalises the analogous condition in \cite{N59, N79} and is key to the local regularity of C$(\alpha)$ tests. 

\begin{remark}\label{rem:orth}
    For Assumption \ref{ass:joint-conv} to hold it is necessary that the $g_{n}$ are approximately zero mean: since $(g_{n})_{n\in \N}$ is uniformly $P_n$-integrable, $P_{n} g_{n} = o(1)$. 
    It is also necessary that the $g_{n}$ satisfy an approximate orthogonality property with the scores for nuisance parameters: as $([\Delta_{n}h]g_{n})_{n\in \N}$ is uniformly $P_n$-integrable for each $h = (\tau, b)\in H$, 
    $   \limn \IP{\Delta_{n}h}{g_{n}\transp} =  \tau\transp\Sigma_{21}\transp = \limn \IP{\Delta_{n}(\tau, 0)}{g_{n}\transp}$,  and so
    \begin{equation}\label{rem:orth:eq:orth}
        \IP{\Delta_{n}(0, b)}{g_{n}\transp} = \IP{\Delta_{n}h}{g_{n}\transp} - \IP{\Delta_{n}(\tau, 0)}{g_{n}\transp} = o(1).
    \end{equation}
\end{remark}

Given any $d_\theta$ moment conditions $f_{n}\in L_2^0(P_{n})$, moment conditions which satisfy an exact version of the orthogonality condition \eqref{rem:orth:eq:orth} may be obtained as
\begin{equation}\label{eq:orth-proj-g}
    g_{n} \define  \Pi\left[f_{n} \middle| \left\{\Delta_{n}(0, b) : b\in B\right\}^{\perp}\right].
\end{equation}

An important special case of this construction is with $f_{n}$ the score function for $\theta$, i.e. $f_{n} = \dotscrarg{n}$ such that $\tau\transp \dotscrarg{n}= \Delta_{n}(\tau, 0)$ for each $\tau\in \R^{d_\theta}$. The function
\begin{equation}\label{eq:effscr}
    g_{n} = \effscrarg{n} \define \Pi\left[\dotscrarg{n}\middle| \left\{\Delta_{n}(0, b) : b\in B\right\}^{\perp}\right],   
\end{equation}
is called the \emph{efficient score function}.
 This yields a power optimal choice of moment conditions satisfying \eqref{rem:orth:eq:orth} as shown in Section \ref{sec:theory} below.

\begin{example}[continues=ex:SIM-running-example, 
    ]
    Let $\upomega: \R^{K}\to [\underline{\upomega},
    \overline{\upomega}]\subset (0, \infty)$. Then $g_n\define \G_n g$,
    \begin{equation}\label{ex:SIM-running-example:eq:g}
        g(W)\define\upomega(X)(Y - f(V_{\theta}))f'(V_{\theta})\left(X_2 - \frac{\E[\upomega(X) X_2| V_{\theta}]}{\E[\upomega(X) |V_{\theta}]}\right),
    \end{equation}
    has components which belong to $\{\Delta_n(0, b): b\in B\}^\perp$ (where $\Delta_n$ is as in \eqref{eq:SIM-running-example-score-op}).\footnote{
   $g$ coincides with the efficient score function $\tilde{\ell}$ (derived by \citealp{NS93}),
\begin{equation*}
   \tilde{\ell}(W)= \tilde{\upomega}(X)(Y - f(V_{\theta}))f'(V_{\theta})\left(X_2 - \frac{\E[\tilde{\upomega}(X)X_2| V_{\theta}]}{\E[\tilde{\upomega}(X) | V_{\theta}]}\right), \quad \tilde{\upomega}(X)\define \E[\epsilon^2|X]^{-1},
\end{equation*}
in the (typically infeasible) case with $\upomega = \tilde{\upomega}$.
    } Under regularity conditions, $g_n$ satisfies Assumption \ref{ass:joint-conv} (see Section \ref{ssec:sim} below).
    
\end{example}

To construct the test statistic, I assume that consistent estimators of $g_{n}$, $V^\dagger$ (the Moore-Penrose pseudo-inverse of $V$) and $r\define \rank(V)$ are available, given $\theta$. 

\begin{assumption}[Consistent estimation]\label{ass:consistent}
    $\hat{g}_{n, \theta}$, $\hat\Lambda_{n, \theta}$, $\hat{r}_{n, \theta}\in \{0, 1, \ldots, d_\theta\}$ satisfy
    \begin{enumerate}
        \item $\hat{g}_{n, \theta} - g_{n} \xrightarrow{P_{n}}0$;\label{ass:consistent:itm:ghat}
        \item $\hat\Lambda_{n, \theta} \xrightarrow{P_{n}}  V^\dagger$;\label{ass:consistent:itm:lambdahat}
        \item If $r \ge1$, then $\hat{r}_{n, \theta} \xrightarrow{P_{n}}r$; if $r = 0$, then $\rank(\hat\Lambda_{n, \theta}) \xrightarrow{P_{n}} 0$.
        \label{ass:consistent:itm:rank}
    \end{enumerate}
\end{assumption}

Verification of Assumption \ref{ass:consistent}\ref{ass:consistent:itm:ghat} typically proceeds by model specific arguments. That $g_n$ is (asymptotically) orthogonal to $\{\Delta_n(0, b):b\in B\}$ often helps in establishing this consistency (cf. \citealp{CEINR22}). 
One generally applicable approach to obtain an estimator which satisfies Assumption \ref{ass:consistent}\ref{ass:consistent:itm:lambdahat} is to take an initial estimator which is consistent for $V$, threshold its eigenvalues at an appropriate rate and then take the pseudo-inverse.\footnote{See Section S5 of \cite{LM21-S} for full details of this approach. Other regularisation schemes are also possible (see e.g. \citealp{DV15}).} If one uses the estimator $\hat{\Lambda}_{n, \theta} \define \hat{V}_{n, \theta}^\dagger$ where $\hat{V}_{n, \theta}\xrightarrow{P_{n}}V$ and $\hat{r}_{n, \theta}\define \rank(\hat{V}_{n, \theta})$ then condition \ref{ass:consistent:itm:lambdahat} holds if and only if condition \ref{ass:consistent:itm:rank} holds \cite[Theorem 2]{A87}. Nevertheless, as emphasised by the notation, it is not necessary that the estimate $\hat{\Lambda}_{n, \theta}$ be the pseudo-inverse of an initial estimate. 

\begin{example}[continues = ex:SIM-running-example]
    Given estimators $\hat{f}_{n, i}$, $\widehat{f'}_{n, i}$ of $f, f'$ and $\hat{Z}_{1, n,i}, \hat{Z}_{2, n, i}$ of $Z_1\define \E[\upomega(X)X_2|V_{\theta}]$, $Z_2\define \E[\upomega(X)|V_{\theta}]$, define $ \hat{g}_{n, \theta} \define \fracrootn\sumin \hat{g}_{n, \theta, i}$,
    \begin{equation}\label{ex:SIM-running-example:eq:ghati}
        \hat{g}_{n, \theta, i}\define  \upomega(X_i)(Y_i - \hat{f}_{n, i}(V_{\theta, i}))\widehat{f'}_{n, i}(V_{\theta, i})\left(X_{2, i} - \hat{Z}_{1, n, i}(V_{\theta, i}) / \hat{Z}_{2, n, i}(V_{\theta, i})\right).
    \end{equation}
    Under regularity conditions (see Section \ref{ssec:sim}), $\hat{g}_{n, \theta}$ satisfies part \ref{ass:consistent:itm:ghat} of Assumption \ref{ass:consistent}, thresholding the eigenvalues of $\fracn\sumin \hat{g}_{n, \theta, i}\hat{g}_{n, \theta, i}\transp$ at an appropriate rate yields an estimator $\hat{\Lambda}_{n, \theta}$ which satisfies part \ref{ass:consistent:itm:lambdahat} and $\hat{r}_{n, \theta}\define \rank(\hat\Lambda_{n, \theta})$ satisfies part \ref{ass:consistent:itm:rank}.
\end{example}

Given the estimators of Assumption \ref{ass:consistent}, the C($\alpha$)-style test statistic is 
\begin{equation}\label{eq:Shat}
    \hat{S}_{n, \theta} \define \hat{g}_{n, \theta}\transp \hat{\Lambda}_{n, \theta} \hat{g}_{n, \theta} .
\end{equation}

The C($\alpha$) -- style test $\psi_{n, \theta_0}$ of $\mathrm{H}_0$ against $\mathrm{H}_1$ 
 at level $\alpha$ is:
\begin{equation}\label{eq:psi-test}
    \psi_{n, \theta_0} \define \bm{1}\left\{
            \hat{S}_{n, \theta_0}  > c_n
        \right\},
\end{equation}
where $c_n$ is the $1-\alpha$ quantile of a $\chi^2_{\hat{r}_n}$ random variable.

\paragraph{Local regularity} Assumptions \ref{ass:LAN} -- \ref{ass:consistent} suffice for local regularity of $\psi_{n, \theta_0}$.

\begin{proposition}\label{prop:asymp-dist}
Under Assumptions \ref{ass:LAN} and \ref{ass:joint-conv}, for $h = (\tau, b)\in H$
    \begin{equation*}
        g_{n} \overset{P_{n,  h}}{\weakconv} \mc{N}\left(\Sigma_{ 21}\tau, V\right).
    \end{equation*}
    If Assumption \ref{ass:consistent} also holds, then additionally
     \begin{equation*}
        \hat{g}_{n, \theta_0} \overset{P_{n,  h}}{\weakconv} \mc{N}\left(\Sigma_{ 21}\tau, V\right) \qquad \text{ and } \qquad \hat{S}_{n, \theta_0} \overset{P_{n,  h}}{\weakconv} \chi^2_r\left(
            \tau\transp \Sigma_{ 21}\transp V \Sigma_{ 21}\tau
        \right).
    \end{equation*}
\end{proposition}

\begin{theorem}\label{thm:psi-pwr-local-alt}
    Suppose that Assumptions \ref{ass:LAN}, \ref{ass:joint-conv} and \ref{ass:consistent} hold and $h = (\tau, b) \in H$. Then,
    \begin{equation*}
        \limn P_{n,  h} \psi_{n, \theta_0} = \uppi(\tau)\define  \begin{cases}
            1- \P\left(\chi_r^2\left(\tau\transp\Sigma_{ 21}\transp  V ^\dagger \Sigma_{ 21}\tau\right) \le c_r\right) & \text{ if } r\ge 1\\
            0 &\text{ if }r=0
        \end{cases},
    \end{equation*}
    where $c_r$ is the $1-\alpha$ quantile of the $\chi^2_r$ distribution.    
\end{theorem} 

Theorem \ref{thm:psi-pwr-local-alt} immediately shows that $\psi_{n, \theta_0}$ is locally regular (cf. \eqref{eq:locally-regular-test}). The asymptotic orthogonality in \eqref{rem:orth:eq:orth} is key to this result.
If, instead, $\limn \IP{\Delta_n h}{g_{n}\transp} = \tau\transp \Sigma_{21}\transp + c(b)$ with $c(b)\neq 0$, then (by Le Cam's third Lemma) the limiting distribution of $g_n$ under $P_{n, h}$ would be $\mc{N}(\Sigma_{21}\tau + c(b), V)$  and hence the limiting power function of the test sequence would not be free of $b$.

\paragraph{Uniform local regularity}

The local regularity given by \ref{thm:psi-pwr-local-alt} may be ``upgraded'' to local uniform regularity (Definition \ref{defn:locally-uniformly-regular-test}) under various conditions. Here I consider the case where $H$ posseses a (pseudo-)metric structure (e.g. if $H$ is a subset of a (semi-)normed linear space).\footnote{If $H$ posseses a (finite) measure structure and the functions $h =(\tau, b)\mapsto \uppi_{n}(\tau, b)$ are measurable then $\psi_{n, \theta_0}$ is locally uniformly regular except on a ``small'' subset of $H$ by Egorov's Theorem. See Section \ref{ssec:unif-measure-structure} for details.}
In this case, for $\psi_{n, \theta_0}$ to be locally uniformly regular on a compact (or totally bounded) $K\subset H$ it is necessary and sufficient that the functions $h =(\tau, b)\mapsto \uppi_{n}(\tau, b)$ are asymptotically equicontinuous. 

\begin{corollary}\label{cor:psi-locally-uniformly-regular}
    Suppose that the conditions of Theorem \ref{thm:psi-pwr-local-alt} hold and that $(H, d)$ is a pseudometric space. If the functions $h = (\tau, b)\mapsto \uppi_n(\tau, b) \define  P_{n,  h} \psi_{n, \theta_0}$ are asymptotically equicontinuous on a compact (or totally bounded) subset $K \subset H$, 
    \begin{equation*}
        \limn \sup_{(\tau, b)\in K}\left|\uppi_n(\tau, b) - \uppi(\tau)
        \right| = 0.
    \end{equation*}
\end{corollary}

I now give a sufficient condition for the asymptotic equicontinuity required by Corollary \ref{cor:psi-locally-uniformly-regular}.

\begin{lemma}\label{lem:level-alpha-unif-equicontinuity-TV}
    If $(H, d)$ is a pseudometric space and $(h\mapsto P_{n, h})_{n\in \N}$ is asymptotically equicontinuous in $d_{TV}$ on $K\subset H$, then $(h\mapsto P_{n, h} \psi_{n, \theta})_{n\in \N}$ is asymptotically equicontinuous on $K$.
\end{lemma}

\begin{remark}\label{rem:ULAN-compact-equicontinuity-in-TV}
    Lemma \ref{lem:level-alpha-unif-equicontinuity-TV} requires asymptotic equicontinuity in total variation of the functions $(h\mapsto P_{n, h})_{n\in \N}$ on subsets $K\subset H$. This holds for any compact $K$ under ULAN (Assumption \ref{ass:ULAN}), as shown in Proposition \ref{prop:ULAN-LAN-equicontinuity}.
\end{remark}

In the parametric i.i.d. case LAN is often verified by establishing a DQM condition, e.g. equation (7.1) in \cite{vdV98}. This is sufficient for the ULAN expansion in Assumption \ref{ass:ULAN} to hold \cite[e.g.][Theorem 7.2]{vdV98}. Semiparametric generalisations of this result are available (e.g. combine Proposition \ref{prop:ULAN-LAN-equicontinuity} and Lemma \ref{lem:iid-ULAN-DQM}).\footnote{\cite{LM21} and \cite{HLM22} verify this asymptotic equicontinuity property in i.i.d. and time series semiparametric examples respectively.} 

The condition in Lemma \ref{lem:level-alpha-unif-equicontinuity-TV} is natural given its link with the ULAN condition. Neverthelesss, it is (much) stronger than necessary for the condition required by Corollary \ref{cor:psi-locally-uniformly-regular}; see Lemma \ref{lem:level-alpha-unif-equicontinuity-weak} for  a weaker sufficient condition.

\section{Power optimality}\label{sec:theory}

The preceding section established the local regularity of the tests $\psi_{n, \theta_0}$ based on (estimates of) moment functions $g_{n}$ satisfying certain asymptotic orthogonality conditions. Thus far, nothing has been said about the choice of $g_{n}$ beyond these orthogonality requirements. 
The choice of the functions $g_{n}$ determines the power of the corresponding test. As such, they ought to be chosen such that the resulting test has good power against alternatives of interest. 

One natural choice is the efficient score function \eqref{eq:effscr}.
It is well known that tests based on the efficient score function have certain optimality properties in regular models when (a) the observations are i.i.d. \cite[cf. Section 25.6][]{vdV98} or (b) when the information operator for $\eta$ is boundedly invertible \citep{CHS96}. I show that this optimality persists in non-regular models and does not require (a) or (b).

The results in this section are derived using the limits of experiments framework of Le Cam \cite[e.g.][]{LC86, vdV98}. In particular, I show that the local experiments consisting of the measures $P_{n, h}$ for $h\in H$ converge weakly to a limit experiment which has a close relationship to a Gaussian shift experiment on the Hilbert space formed by taking the quotient of $H$ under the seminorm induced by the variance function $\sigma(h)$. The connection between these experiments is sufficiently tight that power bounds derived in the latter transfer to the former.\footnote{That the local experiments do not converge to the mentioned Gaussian shift experiment is essentially a purely technical point: the Gaussian shift experiment is defined on a different parameter space to the local experiments, whilst (weak) convergence of experiments (in the sense of \citealp{LC86}) is defined for experiments with the same parameter space.} 

\paragraph{The limit experiment}
For this development $H$ is required to be linear and I will therefore assume that $B$ (hence $H$) is a linear space.
Under LAN, there exists a positive semi-definite symmetric bilinear form $\IP{\cdot}{\cdot}_{K}$ on $H=  \R^{d_\theta} \times B$ such that $\sigma(h) = \IP{h}{h}_{K}$.
This can be seen as a by-product of the following Lemma.
\begin{lemma}\label{lem:IP-GP}
    Suppose Assumption \ref{ass:LAN} holds and $B$ is a linear space. Let $\Delta$ be a square integrable stochastic process defined on $H$ such that $\Delta_{n}h \overset{P_{n}}{\weakconv} \Delta h$.
    Then $\Delta$ is a mean-zero Gaussian linear process with covariance kernel $K$, where
    \begin{equation*}
        K(h, g) \define \limn P_{n}\left[\Delta_{n}h\Delta_{n}g\right].
    \end{equation*}
\end{lemma}

For $h, g\in H$, setting $\IP{h}{g}_{K} \define K(h, g)$
gives a positive semi-definite symmetric bilinear form. Let $\|\cdot\|_{K}$ denote the seminorm induced by $\IP{\cdot}{\cdot}_{K}$ on $H$.

\begin{remark}\label{rem:IP-linear-operator}
    Suppose that $\IP{\cdot}{\cdot}_{H}$ is an inner product on $H = \R^{d_\theta} \times B$.
    The existence of the positive semi-definite symmetric bilinear form $\IP{\cdot}{\cdot}_{K}$ is equivalent to the existence of a bounded, self-adjoint, positive semi-definite linear operator $\mathsf{B}$ such that $\IP{h}{h}_{K} = \IP{h}{\mathsf{B} h}_{H}$ for $h\in H$ 
    \cite[cf.][p. 845]{CHS96}.
\end{remark}

Define $\mb{H}$ as the quotient of $H$ by the subspace on which $\|\cdot\|_{K}$ vanishes:  
\begin{equation}\label{eq:mbHgam}
    \mb{H} \define \quotient{H}{\{h\in H : \|h\|_{K} = 0\}}.
\end{equation}
which is an inner product space when equipped with the natural inner product induced by $\IP{\cdot}{\cdot}_{K}$, which I also denote by $\IP{\cdot}{\cdot}_{K}$.
An element of $\mb{H}$ corresponding to representative element $h\in H$ will be denoted by $[h]$.\footnote{
    Analogous comments apply to the related space $\mb{H}_1$, defined below. In both cases, to avoid an excess of parentheses / brackets, if $h = (\tau, b)$ I will write either $[h]$ or $[\tau, b]$, rather than $[(\tau, b)]$. 
\label{ftnt:h-vs-[h]}
}

The (weak) limit of the sequence of experiments consisting of the measures $P_{n, h}$ can be obtained by standard results on weak convergence of experiments. 

\begin{proposition}\label{prop:conv-exp}
    Suppose that Assumption \ref{ass:LAN} holds, that $B$ is a linear space
    and define the sequence of experiments $ \ms{E}_{n}\define \left(\mc{W}_n, \mc{B}(\mc{W}_n), (P_{n,  h}: h\in H)\right)$.
    Let $\Delta$ be the Gaussian process of Lemma \ref{lem:IP-GP} and let $(\Omega, \F, \P)$ be the probability space on which it is defined. Define the experiment $\ms{E}\define (\Omega, \F, (P_{ h}: h\in H))$ according to
    \begin{equation*}
        P_{ 0}\define \P; \qquad \deriv{P_{ h}}{P_{ 0}} = \exp\left(\Delta h - \frac{1}{2}\|h\|^2\right), \quad h\in H.
    \end{equation*}
    Then $\ms{E}_{n}$ converges weakly to $\ms{E}$.   
\end{proposition}

Under the assumption that $\mb{H}$ is separable, the experiment $\mc{E}$ is equivalent to a Gaussian shift on $(\mb{H}, \IP{\cdot}{\cdot}_{K})$, in the sense given by  Proposition \ref{prop:exp-equiv-shift} below.

\begin{assumption}\label{ass:IP}
    $B$ is a linear space and $\mb{H}$ as defined in \eqref{eq:mbHgam} is separable. 
\end{assumption}

\begin{proposition}\label{prop:exp-equiv-shift}
    Suppose Assumptions \ref{ass:LAN} and \ref{ass:IP} hold. If $\ms{E}$ is as in Proposition \ref{prop:conv-exp}, there is a Gaussian shift experiment $\ms{G} \define (\Omega, \F, (G_{[h]} : [h]\in \mb{H}))$ 
    such that  $d_{TV}(P_{h}, G_{[h]}) = 0$ for each $h\in H$.
\end{proposition}

\paragraph{The efficient information matrix}

Power bounds for tests of $\mathrm{K}_0: h\in H_0$ against $\mathrm{K}_1: h\notin H_1$ can be expressed in terms of the \emph{efficient information matrix}, $\effinfo$, so named because in the i.i.d. setting it is the covariance matrix of the efficient score function for a single observation. Here I provide an alternative definition of this matrix which applies more generally, and reduces to the classical definition in the i.i.d. case (as shown in Lemma \ref{lem:effscr-coincides-with-usual-defn}).

Let $\|\tau\| \define  \inf_{b\in B}\|(\tau, b)\|_{K}$, which defines a semi-norm on $\R^{d_\theta}$.
Equipping the quotient $\mb{H}_1 \define  \quotient{\R^{d_\theta}}{\{\tau \in \R^{d_\theta} :\|\tau\| = 0\}}$ with the natural norm induced by $\|\cdot\|$ (which I also denote by $\|\cdot\|$) turns it into a normed space.
Define the linear map $\pi_1:\mb{H} \to\mb{H}_1$ as $\pi_1([\tau, b]) \define [\tau]$.
As $\pi_1$ is continuous it may be uniquely extended to a continuous function defined on $\overline{\mb{H}}$, the completion of $\mb{H}$; this extension will also be called $\pi_1$.
Since $\pi_1$ is continuous, $\ker \pi_1\subset \overline{\mb{H}} $
is closed. Let $\Pi$ be the orthogonal projection onto $\ker \pi_1$ and define $\Pi^\perp\define I - \Pi$, the orthogonal projection onto $[\ker \pi_1]^\perp$. 
Let $e_i$ be the $i$-th canonical basis vector in $\R^{d_\theta}$ and define the \emph{efficient information matrix} $\effinfo$ as the $d_\theta \times d_\theta$ matrix with $i,j$-th entry $\effinfoarg{ij}$ given by\footnote{Lemma \ref{lem:calculation-of-effinfo} gives an alternative expression for $\effinfo$ based on the Gaussian process $\Delta$ of Lemma \ref{lem:IP-GP}.}
\begin{equation}\label{eq:effinfo}
    \effinfoarg{ij} = \IP{\Pi^\perp [e_i, 0]}{\Pi^\perp [e_j, 0]}_{K}.
\end{equation}

\begin{lemma}\label{lem:ker-effinfo-is-fzero}
    Under Assumption \ref{ass:IP}, $\|\tau\|^2 = \tau\transp \effinfo\tau$ and $\ker \effinfo = \{\tau\in \R^{d_\theta}: \|\tau\| = 0\}$.
\end{lemma}

\subsection{Tests of a scalar parameter}

The following Theorem records the power bound for (locally asymptotically) unbiased two-sided tests of a scalar $\theta$. In this case 
the matrix $\effinfo$ has rank either 0 or 1.
Theorem \ref{thm:two-sided-pwr-bound} handles both cases simultaneously.

\begin{theorem}\label{thm:two-sided-pwr-bound}
    Suppose that Assumptions \ref{ass:LAN} and \ref{ass:IP} hold  and $d_{\theta}=1$.
    Let $\phi_n: \mc{W}_n\to [0, 1]$ be a sequence of locally asymptotically unbiased level $\alpha$ tests of $\mathrm{K}_0: \tau = 0$ against $\mathrm{K}_1:\tau\neq 0$. That is, 
    \begin{equation*}
        \limsupn P_{n, h} \phi_n \le \alpha, \quad h\in H_{0},\qquad  \text{ and }\qquad \liminfn  P_{n, h} \phi_n \ge \alpha, \quad h\in H_{1}.
    \end{equation*}
    Then, for any
     $h\in H$, 
    \begin{equation}\label{thm:two-sided-pwr-bound:eq:power-bound}
        \limsupn  P_{n, h} \phi_n \le 1 - \Phi\left(z_{\alpha/2} - \effinfo^{1/2}\tau\right)  + 1- \Phi\left(z_{\alpha/2} + \effinfo^{1/2}\tau
\right),
    \end{equation}
    where $z_{\alpha}$ is the $1-\alpha$ quantile and $\Phi$ the CDF of the standard normal distribution.
\end{theorem}

Theorem \ref{thm:psi-pwr-local-alt} implies that the  power bound of Theorem \ref{thm:two-sided-pwr-bound} is achieved by the test $\psi_{n, \theta_0}$ provided $\Sigma_{21}V^\dagger \Sigma_{21}\transp = \effinfo$ and $r=1$.

\begin{corollary}\label{cor:psi-two-sided}
    Suppose that Assumptions \ref{ass:LAN}, \ref{ass:joint-conv} and \ref{ass:consistent}
     hold with $\Sigma_{21}V^\dagger\Sigma_{21}\transp = \effinfo$ and $r=1$.
    Then, for $h\in H$,
    \begin{equation}\label{eq:two-sided-power-bound}
        \limn P_{n, h}\psi_{n, \theta_0}   = 
            1 - \Phi\left(z_{\alpha/2} -\effinfo^{1/2}\tau\right)  + 1- \Phi\left(z_{\alpha/2} +\effinfo^{1/2}\tau\right).
    \end{equation}
\end{corollary}

\subsection{Tests of a multivariate parameter}

When $d_\theta >1$ there is an intermediate case where $0<\rank(\effinfo)<d_\theta$. Here I permit $0<\rank(\effinfo)\le d_\theta$ and establish  a maximin power bound for (potentially) non -- regular models, which contains the (regular) full rank case as a special case.\footnote{Section \ref{sm:ssec:stringent} shows that the most stringent test (in the sense of \citealp{W43}) in the limit experiment has the same power function as the maximin test, and no sequence of asymptotically level $\alpha$ tests can correspond to a test in the limit experiment with smaller regret.}

\begin{theorem}\label{thm:maximin-pwr-bound}
    Suppose that Assumptions \ref{ass:LAN} and \ref{ass:IP} hold and $r \define \rank(\effinfo) \ge 1$. Let $\phi_n:\mc{W}_n \to [0, 1]$ be a sequence of tests
     such that for each $h = (0, b)\in H_{0}$
        \begin{equation*}
            \limsupn P_{n, h} \phi_n \le \alpha
        \end{equation*}
    Let $c_r$ the $1-\alpha$ quantile of a $\chi^2_r$ random variable.
    Then, if $a\ge0$, 
    \begin{equation}\label{thm:maximin-pwr-bound:eq:bound}
        \limsupn \inf \left\{P_{n, h}\phi_n :h = (\tau, b)\in H,\ \tau\transp\effinfo\tau \ge a\right\} \le 1 - \P(\chi_r^2(a) \le c_r).
    \end{equation}    
\end{theorem}

By Theorem \ref{thm:psi-pwr-local-alt}, the power bound on the right hand side of \eqref{thm:maximin-pwr-bound:eq:bound} is achieved by $\psi_{n, \theta_0}$ provided $\Sigma_{21}V^\dagger \Sigma_{21}\transp = \effinfo$ and $\rank(V) = \rank(\effinfo) = r\ge1$.
In order that the test be asymptotically maximin over a compact subset $K_a$ of $\{h = (\tau, b)\in H: \tau\transp \effinfo \tau \ge a\}$, with $a = \inf\{\tau\transp \effinfo \tau = a: h\in K_a\}$, some uniformity is required.\footnote{  The pseudometric $d$ in Corollary \ref{cor:psi-maximin} 
need not be related to the seminorm $\|\cdot\|_K$.}

\begin{corollary}\label{cor:psi-maximin}
    Suppose that Assumptions \ref{ass:LAN}, \ref{ass:joint-conv} and \ref{ass:consistent}
    hold with $ \Sigma_{ 21} V^{\dagger} \Sigma_{ 21}\transp = \effinfo$ and $r = \rank(\effinfo) = \rank(V)\ge 1$.
    Then for $h = (\tau, b)\in H$
     \begin{equation}\label{cor:psi-maximin:eq-1}
        \limn P_{n,  h}\psi_{n, \theta_0} = 1 - \P\left(\chi^2_r\left(a\right) \le c_r\right), \quad a = \tau\transp \effinfo \tau.
     \end{equation}
     Additionally, suppose that $(H, d)$ is a (pseudo-)metric space and let $K_a$ be a compact subset of $\{h = (\tau, b)\in H: \tau\transp \effinfo \tau \ge a\}$ such that $a = \inf\{\tau\transp \effinfo \tau: h = (\tau, b)\in K_a\}$. If the functions $h\mapsto P_{n,  h}\psi_{n, \theta_0}$ are asymptotically equicontinuous on $K_a$,
     \begin{equation}\label{cor:psi-maximin:eq-2}
        \limn \inf_{h\in K_a} P_{n,  h}\psi_{n, \theta_0} = 1 - \P\left(\chi^2_r\left(a\right) \le c_r\right).
     \end{equation}
\end{corollary}

A sufficient condition for the asymptotic equicontinuity required for the second part of Corollary \ref{cor:psi-maximin} based on an asymptotic equicontinuity in total variation requirement was given as Lemma \ref{lem:level-alpha-unif-equicontinuity-TV} in the previous section.\footnote{As noted there Lemma \ref{lem:level-alpha-unif-equicontinuity-weak} provides some weaker sufficient conditions; in the present context, condition \ref{lem:level-alpha-unif-equicontinuity-weak.itm.Lambda} of Lemma \ref{lem:level-alpha-unif-equicontinuity-weak} is not required (cf. Remark \ref{rem:level-alpha-unif-equicontinuity-weak-condition-split}).}

\subsection{The degenerate case}

If the efficient information matrix $\effinfo$ is zero, no test with correct asymptotic size has non -- trivial asymptotic power against any sequence of local alternatives. 

\begin{proposition}\label{prop:asy-size-alpha-and-r-0-imply-no-power}
    Suppose Assumptions \ref{ass:LAN} and \ref{ass:IP} hold and $r \define \rank(\effinfoarg{}) = 0$. Let $\phi_n:\mc{W}_n \to [0, 1]$ be a sequence of tests such that $\limsupn P_{n, h} \phi_n \le \alpha$ for each $h = (0, b)\in H_{0}$
    Then, for $h\in H$, $\limsupn P _{n, h} \phi_n \le \alpha$.
\end{proposition}

\subsection{Discussion of the power bounds}

There are a number of important aspects to highlight regarding the interpretation of the power bounds obtained in the preceding subsections.

\paragraph{Optimality in multivariate testing problems}

Just as in the classical finite -- dimensional case, the multivariate optimality results just presented 
should not be taken in an absolute sense. Nevertheless they seem reasonable if the researcher does not have directions against which they wish to direct power a priori. If there are alternatives of particular interest, one could construct a locally regular test by utilising the same moment conditions $g_{n}$ but weighting them differently  \cite[cf.][]{BRS06}.

\paragraph{The intermediate case with $1\le\rank(\effinfo)<d_\theta$} 

A key benefit of the multivariate power results is that they apply equally to non-regular models, i.e. cases where $\effinfo$ is rank deficient. This scenario can occur for various reasons. Firstly the model may not identify all parameters of interest $\theta$ (i.e. underidentification). Secondly some of the elements of $\theta$ may be weakly identified (i.e. weak underidentification). The power results above apply in either of these cases. 

There are a number of other papers which provide inference results in similarly rank deficient settings \citep[e.g.][]{RCBR00, HM19, AG19, ABS23}; none of these papers consider optimal testing.

\paragraph{Alternative approximations}

In the case where $\rank(\effinfo) = 0$, Proposition \ref{prop:asy-size-alpha-and-r-0-imply-no-power} reveals that the LAN approximation in Assumption \ref{ass:LAN} is, in a certain sense, the wrong approximation: it does not provide any useful way of (asymptotically) comparing tests. Other approximations might provide valuable comparisons. Alternative approximations have been explored in, for example, the IV model \cite[e.g.][]{M09} and semiparametric GMM models by \cite{AM22, AM23}. For example, in the IV case \cite{M09} considers alternatives which are at a fixed distance from the true parameter, rather than in a shrinking $\sqrt{n}$-neighbourhood. Whether such an approach can be developed for the class of models considered here is an interesting question for future work. 

\subsection{Attaining the power bounds}\label{ssec:attainment-pwr-bounds}

Provided that the $L_2$ distance between $g_{n}$ and $\effscrarg{n}$ (as defined in \eqref{eq:effscr}) vanishes, $\psi_{n, \theta_0}$ attains the power bounds established in the preceding subsections. For regular models this result is well known in two special cases: (a) the i.i.d. case (cf. Section 25.6 in \cite{vdV98}; Lemma \ref{lem:effscr-coincides-with-usual-defn}) and (b) when the information operator $\mathsf{B}$ in Remark \ref{rem:IP-linear-operator} is positive-definite with $\mathsf{B}_{22}$, the information operator for $\eta$, boundedly invertible \citep{CHS96}. Here I provide a general version of this result which applies to both regular and non-regular models and does not require (a) or (b). In particular, I show that $\Sigma_{21}V^\dagger \Sigma_{21}\transp = \effinfo$, which suffices given Theorem \ref{thm:psi-pwr-local-alt} and the power bounds in Theorems \ref{thm:two-sided-pwr-bound} and \ref{thm:maximin-pwr-bound}.

\begin{theorem}\label{thm:pwr-bounds-effscr}
    Suppose that Assumptions \ref{ass:LAN}, \ref{ass:joint-conv}, \ref{ass:consistent} and \ref{ass:IP} hold and $g_n$ is such that $\limn \int \|g_{n}  - \effscrarg{n}\|^2\dP_n = 0$. Then $\Sigma_{ 21} = V = \effinfo$, hence $\Sigma_{ 21}V^\dagger \Sigma_{21}\transp = \effinfo$.
\end{theorem}

\section{The smooth i.i.d. case}\label{ssec:smooth-iid}

In this section I give conditions which  are sufficient for some of the foregoing Assumptions in the in the benchmark case for semiparametric theory: where the observations are i.i.d. and the model is ``smooth''.

\begin{assumption}[Product measures]\label{ass:iid}
    Suppose $W^{(n)} = (W_1, \ldots, W_n)\in \prod_{i=1}^n\mc{W} = \mc{W}_n$
     and that each $P_{n, h}$ is a product measure: $P_{n,h} = P_{h}^n$. Each probability measure in $\mc{P}_n$ is dominated by the $n$-fold product of a $\sigma$-finite measure $\nu$. 
\end{assumption}

In the i.i.d. setting, it is well known that quadratic mean differentiability of the square root of the density $p = \deriv{P}{\nu}$ of $P \define P_0$ is sufficient for LAN. In particular, if
\begin{equation}\label{eq:dqm}
    \limn \int \left[\sqrt{n}\left(\sqrt{p_{h_n}} - \sqrt{p}\right) - \frac{1}{2}A h\sqrt{p}\right]^2 = 0,
\end{equation}
for a measurable $Ah: \mc{W}\to \R$, then with $\Delta_{n}h \define \fracrootn \sumin [Ah](W_i)$ the remainder term $R_{n}$ in the LAN expansion satisfies $R_{n}(h_n)\xrightarrow{P}0$ \cite[e.g.][Lemma 3.10.11]{vdVW96}. This can be used to establish either the LAN expansion required by Assumption \ref{ass:LAN} by taking $h_n = h$ for each $n\in\N$ or the ULAN expansion as in Assumption \ref{ass:ULAN} by considering sequences $h_n\to h$. Sufficient conditions for \eqref{eq:dqm} are well known \cite[e.g.][Lemma 7.6]{vdV98}.

In this case, the scores $Ah$ typically take the form 
\begin{equation}\label{eq:Agamh}
    [Ah](W_i) = \tau\transp \dotscr (W_i) + [Db](W_i), \quad  h = (\tau, b)\in H,
\end{equation}
where $\dotscr$ is a $d_\theta$-vector of functions in $L_2^0(P)$ (typically the partial derivatives of $\theta \mapsto \log p_{\gamma}$ at $\gamma$) and $D:\cllin B \to L_2^0(P)$ a bounded linear map. Showing that \eqref{eq:dqm} holds (with $h_n =h$) is typically the most straightforward way to verify the LAN expansion required by Assumption \ref{ass:LAN}. If $A: \cllin H \to L_2(P)$ is a bounded linear map, then the remainder of Assumption \ref{ass:LAN} also follows directly.\footnote{A version of Lemma \ref{lem:iid-LAN-DQM} for ULAN (Assumption \ref{ass:ULAN}) is Lemma \ref{lem:iid-ULAN-DQM} in the supplementary material.}

\begin{lemma}\label{lem:iid-LAN-DQM}
    Suppose that Assumption \ref{ass:iid} holds and for each $h\in H$ equation \eqref{eq:dqm} holds (with $h_n=  h$) with $A: \cllin H \to L_2(P)$ a bounded linear map. Then Assumption \ref{ass:LAN} holds with $P_{n, h} = P_{h/\sqrt{n}}^n$ and 
    $[\Delta_{n} h](W^{(n)}) = \G_n A h$.
\end{lemma}

When the data are i.i.d., the the joint convergence of $(\Delta_{n}h, g_{n}\transp)$ as in Assumption \ref{ass:joint-conv} is particularly straightforward to verify. As noted in the discussion around \eqref{eq:orth-proj-g}, the required orthogonality condition can be ensured by performing an orthogonal projection. Assumption \ref{ass:joint-conv} then follows straightforwardly. 
In the i.i.d. setting typically $g_{n}$ will have the form $g_{n}(W^{(n)}) 
 = \G_n g$.

\begin{lemma}\label{lem:iid-orthocomp-joint-conv}
    Suppose that Assumptions \ref{ass:LAN}  and \ref{ass:iid} hold, with $\Delta_{n} h  = \fracrootn\sumin Ah$, where $A h$ is as in equation \eqref{eq:Agamh}. Additionally suppose that 
    $g \in \left\{Db : b\in B\right\}^\perp\subset L_2^0(P)$. Then Assumption \ref{ass:joint-conv} holds with $g_{n}(W^{(n)})\define \G_n g$. 
\end{lemma}
\begin{corollary}\label{cor:iid-orth-projection-joint-conv}
    In the setting of Lemma \ref{lem:iid-orthocomp-joint-conv}, if $f\in L_2^0(P)$ and $g$ is the orthogonal projection $g = \Pi[f | \{Db : b\in B\}^\perp ]$
    then Assumption \ref{ass:joint-conv} holds with $g_{n}(W^{(n)})\define \G_n g$.
\end{corollary}

\section{Examples}\label{sec:examples}

I now illustrate the application of the theoretical results to the single index and IV models and conduct simulation studies to investigate finite sample performance of the proposed approach. 
In this section I work under high level conditions to avoid repeating standard regularity conditions; lower level sufficient conditions are given in section \ref{sm:sec:examples-detail} of the supplementary material.

\subsection{Single index model}\label{ssec:sim}

Consider the single index model of Example \ref{ex:SIM-running-example}. I now formalise the development given in Section \ref{sec:heuristic}.
 The model parameters are $\gamma = (\theta, \eta)$ where $\eta= (f, \zeta)$ and the density of one observation with respect to a $\sigma$-finite measure $\tilde{\nu}$ is $p_{\gamma}$ as in \eqref{eq:SIM-running-example-dens}; $P_{\gamma}$ denotes the corresponding probability measure.
The parameters $\gamma$ are restricted by the following Asssumption. Let $\ms{X}$ be the support of $X$, $\ms{D}$ a convex open set containing $\{x_1 + x_2\transp \theta: \theta\in \Theta,x\in \ms{X}\}$ and $C_b^1(\ms{D})$ the class of real functions which are bounded and continuously differentiable with bounded derivative on $\ms{D}$.
\begin{assumption}\label{ass:sim-mdl}
    The parameters $\gamma = (\theta, f, \zeta)\in \Gamma = \Theta\times \ms{F} \times \ms{Z}$ where $\Theta$ is an open subset of $\R^{d_\theta}$, $\ms{F} = C_b^1(\ms{D})$ and  $\zeta\in \ms{Z}$, for 
    \begin{equation*}
                \ms{Z}\define \left\{\zeta\in L_1(\R^{1+K}, \nu): \zeta\ge 0,\, \int_{\R\times \ms{X}}\zeta\dnu = 1,\,
                  \text{if } (\epsilon, X)\sim \zeta \text { then } \eqref{ass:sim-parameters:eq:sim-moments} \text{ holds}\right\},
        \end{equation*}
        with $L_1(A, \nu)$ is the space of $\nu$ -- integrable functions on $A$ and 
        \begin{equation}\label{ass:sim-parameters:eq:sim-moments}
            \E[\epsilon|X] = 0, 
            \  \E[\epsilon^2] < \infty,\,
            \  \E[(|\epsilon|^{2+\rho} + |\phi(\epsilon, X)|^{2+\rho} + 1)\|X\|^{2+\rho}] < \infty,\,
            \  \E[XX\transp] \succ 0,
        \end{equation}
        for $\phi(\epsilon, X)$ the derivative of $e\mapsto \log \zeta(e, X)$. 
    Additionally, for each $\gamma\in \Gamma$, $p_{\gamma}$ is a probability density with respect to some $\sigma$-finite measure $\tilde{\nu}$. 
\end{assumption}
That $p_{\gamma}$ is a valid probability density holds automatically (with $\tilde{\nu} = \nu$) when $\epsilon|X$ is continuously distributed.

\paragraph*{Local Asymptotic Normality} 
Consider local perturbations $P_{\gamma + \varphi_n(h)}$ for 
\begin{equation}\label{eq:sim-varphin}
    \varphi_n(h) = \left(\frac{\tau}{\sqrt{n}},\ \varphi_{n, 2}(b_1, b_2)\right), \qquad h = (\tau, b_1, b_2)\in H = \R^{d_\theta}\times B_{1}\times B_{2}.
\end{equation}
$B_{1}$ is the set which indexes the perturbations to $f$ and consists of a subset of the continuously differentiable functions $b_1:\ms{D}\to\R$.
$B_{2}$ indexes the perturbations to $\zeta$ and consists of a subset of the functions $b_2:\R^{1+K}\to \R$ which are continuously differentiable in their first argument and satisfy
\begin{equation}
    \label{eq:sim-b2-conditions}
    \E[b_2(\epsilon, X)] = 0,\  \E[\epsilon b_2(\epsilon, X)|X] = 0,\  \E[b_2(\epsilon, X)^2] < \infty \quad \text{ for } (\epsilon, X)\sim \zeta.
\end{equation} 
The precise form of $\varphi_{n, 2}$ is left unspecified. It is required only that the local perturbations satisfy the LAN property below.\footnote{Examples of $\varphi_{n, 2}$ and $B$ for which Assumption \ref{ass:sim-LAN} holds are given in Section \ref{app:ssec:sim-LAN}. }

\begin{assumption}\label{ass:sim-LAN}
    Suppose that $\mc{W}_n = \prod_{i=1}^n \R^{1+K}$ and 
    $P_{n,h}\define P_{\gamma + \varphi_n(h)}^n \ll \nu_n$ for all $\gamma\in \Gamma$ and $h\in H$ and are such that Assumption \ref{ass:LAN} holds with 
    \begin{equation}
        \log \frac{p_{n, h}}{p_{n,  0}} = \fracrootn\sumin [Ah](W_i) - \frac{1}{2}\sigma(h) + o_{P_{n, 0}}(1), \quad h\in H_,
    \end{equation}
    where $\sigma(h) = \int [Ah]^2\darg{P}$ and $A$ is as in equation \eqref{eq:Agamh} with 
    \begin{align*}
        \dot{\ell}(W) &\define -\phi(Y - f(V_\theta), X) f'(V_\theta)X_2\\
        [D b](W) &\define -\phi(Y - f(V_\theta), X)b_{1}(V_{\theta}) + b_2(Y - f(V_\theta), X).
    \end{align*}
\end{assumption}

\paragraph*{The moment conditions}\label{para:sim:g}
The test statistic is based on $g_{n}\define \G_n g$ for $g$ given in \eqref{ex:SIM-running-example:eq:g}. This satisfies Assumption \ref{ass:joint-conv} under  Assumptions \ref{ass:sim-mdl} \& \ref{ass:sim-LAN} and \eqref{eq:sim:eps-phi--1} below.

\begin{proposition}\label{prop:sim-joint-conv}
    Suppose Assumptions \ref{ass:sim-mdl} \& \ref{ass:sim-LAN} hold and under $P$,
    \begin{equation}\label{eq:sim:eps-phi--1}
        \E[\epsilon^2|X] \le C < \infty,\quad \E\left[\epsilon\phi(\epsilon, X)| X\right] = -1, \quad \text{a.s.}~.
    \end{equation}    
    Then Assumption \ref{ass:joint-conv} holds with $g_{n}\define \G_n g$ for $g$ given in \eqref{ex:SIM-running-example:eq:g}. 
\end{proposition}

\paragraph*{A feasible test}\label{para:sim:ghat}

To form a feasible test $g_n$ must be replaced by an estimator $\hat{g}_{n, \theta}$. Let this have the form $\hat{g}_{n, \theta}(W^{(n)}) \define \fracrootn\sumin \hat{g}_{n, \theta, i}$, for $\hat{g}_{n, \theta, i}$ defined as in \eqref{ex:SIM-running-example:eq:ghati}. 
To keep the notation concise let $Z_{3}\define f$, $Z_4\define f'$,
$Z_0 \define Z_{1} / Z_2$ and correspondingly $\hat{Z}_{0,n, i}\define \hat{Z}_{1, n, i} / \hat{Z}_{2, n, i}$. 
Let $\check{V}_{n, \theta}\define \meanin \hat{g}_{n, \theta, i}\hat{g}_{n, \theta, i}\transp$ and 
If $V$ is known to have full rank then let $\hat{V}_{n, \theta} \define \check{V}_{n, \theta}$, $\hat{\Lambda}_{n, \theta} \define \hat{V}_{n, \theta}^{-1}$ and $\hat{r}_{n, \theta} = \rank(V)$. 
Form the estimator $\hat{V}_{n, \theta}$ according to the construction in Section S5 of \cite{LM21-S} using a truncation rate $\upnu_n$. $\hat{\Lambda}_{n, \theta}$ is then taken to be $\hat{V}_{n, \theta}^\dagger$ and $\hat{r}_{n, \theta}\define \rank(\hat{V}_{n, \theta})$.
Under the following condition, these estimators satisfy the conditions of Assumption \ref{ass:consistent}.
\begin{assumption}\label{ass:sim-estimation-HL}
    Suppose that equation 
    \eqref{eq:sim:eps-phi--1} holds (under $P$), $X$ has compact support, 
    $\E[\epsilon^4]<\infty$, and with $P$ probability approaching one $\mathsf{R}_{l, n, i} \le r_n = o(n^{-1/4})$, 
    \begin{equation*}
        \mathsf{R}_{l, n, i} \define \left[\int \left\|\hat{Z}_{l, n, i}(v) - Z_l(v)\right\|^2 \darg{\mc{V}(v)}\right]^{1/2}, \quad l = 1,\ldots 4,
    \end{equation*}
    where $\mc{V}$ is the law of $V_{\theta}$ under $P$ and where $\hat{Z}_{l, n, i}(V_{\theta, i})$ is $\sigma(\{V_{\theta, i}\} \,\cup\, \mc{C}_{n, j})$ measurable with $j=1$ if $i >  \lfloor n/2\rfloor$ and 2 otherwise, with $\mc{C}_{n, 1}\define \{W_j : j\in \{1, \ldots, \lfloor n/2\rfloor\}\}$ and $\mc{C}_{n, 2}\define \{W_j : j\in \{\lfloor n/2\rfloor + 1, \ldots, n\}\}$.
\end{assumption}

The rate conditions in Assumption \ref{ass:sim-estimation-HL} can be satisfied by  e.g. (sample -- split) series estimators under standard conditions; see e.g. \cite{BCCK15}.

\begin{proposition}\label{prop:sim-estimation-HL}
    Suppose Assumptions \ref{ass:sim-mdl}, \ref{ass:sim-LAN} and \ref{ass:sim-estimation-HL} hold and $\upnu_n$ is such that $r_n = o(\upnu_n)$.
    Then Assumption \ref{ass:consistent} holds with $V\define \int gg\transp\darg{P}$.
\end{proposition}

A consequence of Assumption \ref{ass:sim-LAN} and Propositions \ref{prop:sim-joint-conv} and \ref{prop:sim-estimation-HL} is that the test $\psi_{n, \theta_0}$ formed as in \eqref{eq:psi-test} is locally regular by Theorem \ref{thm:psi-pwr-local-alt}.

\paragraph*{Simulation study}

I take $K=1$ and test $H_0: \theta=\theta_0 = 1$ at a nominal level of 5\%. Each study reports the results of 5000 monte carlo replications with a sample size of $n\in \{400, 600, 800\}$. 
I report empirical rejection frequencies for the $\psi_{n, \theta_0}$ test along with a Wald test based on an estimator in the style of \cite{I93}. 

I consider two different classes of link function. The first sets $f(v) =f_j(v) = 5\exp (-v^2 / 2c_j^2)$ (``exponential''); the second $f(v) =f_j(v) = 25\left(1 + \exp(-v  /c_j)\right)^{-1}$ (``logistic'').  The values of $c_j$ considered are recorded in Table \ref{tbl:SIM-index-fcns}.

\begin{table*}[!htbp]
	\footnotesize
	\begin{center}
		\caption{\label{tbl:SIM-index-fcns} Functions used in the simulation experiments}
		\begin{threeparttable}
			
\begin{tabular}{llrrr}
\toprule
name & expression & $c_1$ & $c_2$ & $c_3$\\
\midrule
Exponential & $f_j(v) = 5 \exp (-v^2 / 2c_j^2)$ & 1 & 2 & 4\\
Logistic & $f_j(v) = 25(1 + \exp (- v / c_j))^{-1}$ & 1 & 4 & 32\\
\bottomrule
\end{tabular}

		\end{threeparttable}
	\end{center}
\end{table*}

In each case, as $c_j$ increases, the derivative of $f$ flattens out, moving towards a point with $f'=0$, at which $\theta$ is unidentified.\footnote{The functions $f$ and $f'$ are plotted in  Figures \ref{fig:SIM-gaussian-fs} and \ref{fig:SIM-logistic-fs}.} 
I draw covariates as $X= (Z_1, 0.2 Z_1 + 0.4 Z_2 + 0.8)$, where each $Z_k\sim U(-1, 1)$ is independent. The error term is drawn either as $\epsilon = \upsilon / \sqrt{3/2}$ with $\upsilon \sim t(6)$ (``homoskedastic'') or $\epsilon \sim \mc{N}(0,  1 + \sin(X_1)^2)$ (``heteroskedastic'').

I compute the test $\psi_{n, \theta_0}$ as described on p. \pageref{para:sim:ghat}, with $\upomega(X) = 1$. The functions $f,\, f'$ and $Z_{1}$ are estimated via sample split smoothing splines.\footnote{I use the base \textsf{R} function 
\texttt{smooth.spline} with 20 knots. In this setting $Z_2(V_\theta) = 1$ is known.} The truncation parameter $\upnu$ is set to $10^{-3}$. I additionally compute a Wald test in the style of \cite{I93}, using the same non-parametric estimates as for $\hat{g}_{n, \theta}$.\footnote{Given $\hat{f}$, $\hat\theta = \argmin_{\theta\in \Theta_\star} \meanin (Y_i - \hat{f}(V_{\theta, i}))^2$, for $\Theta_\star = [-10, 10]$. The asymptotic variance is estimated by $\hat{\sigma}^2 / \meanin \left(\widehat{f'}(V_{\hat\theta, i}) \left[X_2 - \hat{Z}(V_{\hat\theta, i})\right]\right)^2$, for $\hat{\sigma}^2 = \meanin (Y_i - \hat{f}(V_{\hat\theta, i}))^2$.}

The empirical rejection frequencies of these procedures are recorded in Table \ref{tbl:SIM-size}: $\psi_{n, \theta_0}$ rejects at close to the nominal 5\%  for all simulation designs considered whilst the Wald test over -- rejects in most of the simulation designs considered.
Figure \ref{fig:SIM-power} contains power plots of the $\psi_{n, \theta_0}$ test ($n=800$). For almost flat link functions there is very identifying information and hence very little power available. As the link function moves away from the point of identification failure ($f'=0$), the available power increases and is captured by the $\psi_{n, \theta_0}$ test.

\begin{table*}[!htbp]
	\setlength{\tabcolsep}{5pt}
	\footnotesize
	\begin{center}
		\caption{\label{tbl:SIM-size} \small ERF (\%), Single-index model}
		\begin{threeparttable}
			
\begin{tabular}{rrrrrrrrrrrrr}
\toprule
\multicolumn{1}{c}{} & \multicolumn{6}{c}{Exponential} & \multicolumn{6}{c}{Logistic} \\
\cmidrule(l{3pt}r{3pt}){2-7} \cmidrule(l{3pt}r{3pt}){8-13}
\multicolumn{1}{c}{} & \multicolumn{3}{c}{Homoskedastic} & \multicolumn{3}{c}{Heteroskedastic} & \multicolumn{3}{c}{Homoskedastic} & \multicolumn{3}{c}{Heteroskedastic} \\
\cmidrule(l{3pt}r{3pt}){2-4} \cmidrule(l{3pt}r{3pt}){5-7} \cmidrule(l{3pt}r{3pt}){8-10} \cmidrule(l{3pt}r{3pt}){11-13}
$n$ & $f_1$ & $f_2$ & $f_3$ & $f_1$ & $f_2$ & $f_3$ & $f_1$ & $f_2$ & $f_3$ & $f_1$ & $f_2$ & $f_3$\\
\midrule\multicolumn{7}{l}{$\psi_{n, \theta_0}$}\\
\midrule
400 & 6.08 & 5.90 & 5.30 & 5.66 & 5.68 & 5.12 & 6.36 & 5.94 & 4.80 & 5.66 & 6.00 & 4.68\\
600 & 6.12 & 5.68 & 5.26 & 5.44 & 4.76 & 4.32 & 6.04 & 5.74 & 4.10 & 5.26 & 5.18 & 4.62\\
800 & 6.20 & 6.04 & 5.28 & 5.46 & 5.72 & 5.22 & 6.00 & 5.90 & 4.26 & 5.62 & 5.16 & 4.46\\
\midrule\multicolumn{7}{l}{Wald}\\\midrule
\addlinespace
400 & 13.06 & 18.94 & 13.54 & 15.20 & 20.32 & 14.38 & 8.12 & 13.60 & 14.74 & 8.24 & 15.52 & 14.92\\
600 & 10.28 & 16.32 & 12.58 & 10.60 & 18.92 & 14.18 & 7.18 & 10.74 & 13.88 & 6.82 & 12.60 & 14.84\\
800 & 10.30 & 16.64 & 12.84 & 9.60 & 19.00 & 13.62 & 6.86 & 10.68 & 12.94 & 6.74 & 11.26 & 14.18\\
\bottomrule
\end{tabular}

		\end{threeparttable}
	\end{center}
\end{table*}

\begin{figure}[!htbp]
		\centering
		\caption{\label{fig:SIM-power} Single-index model, $\psi$ power}
		\begin{subfigure}[b]{0.495\textwidth}
			\centering
			\includegraphics[width=0.99\textwidth]{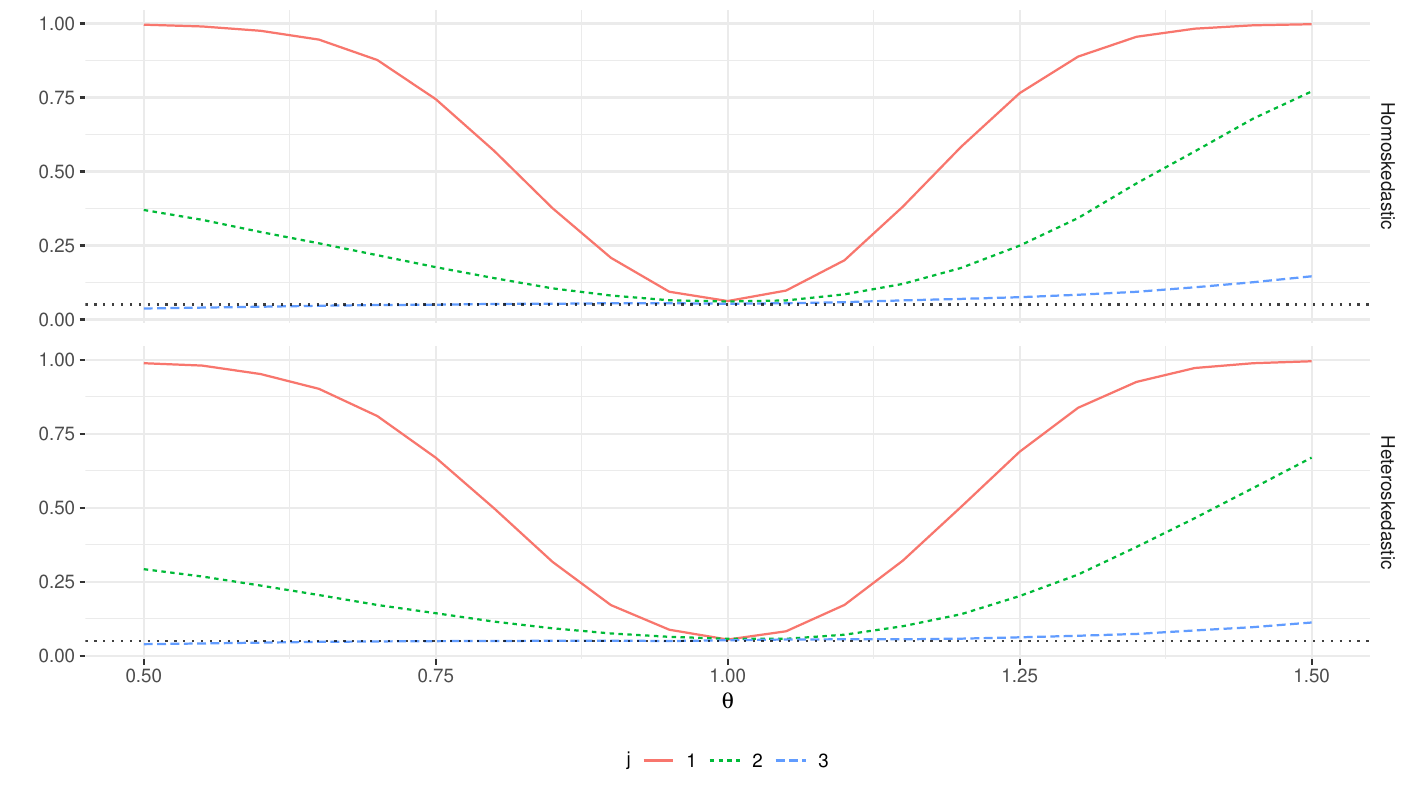}
			\caption[]%
			{{\small Exponential}}
		\end{subfigure}
		\hfill
		\begin{subfigure}[b]{0.495\textwidth}
			\centering
			\includegraphics[width=0.99\textwidth]{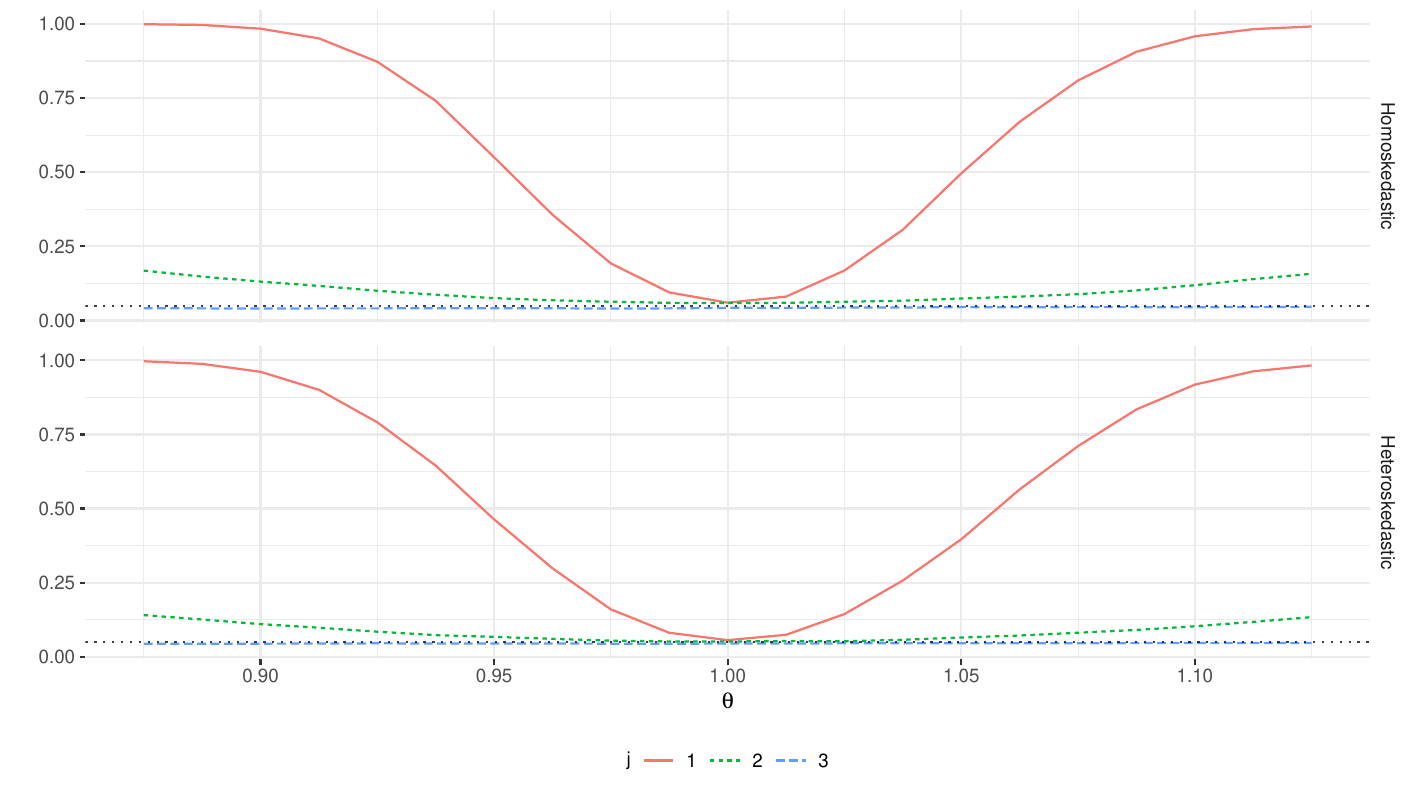}
			\caption[]%
			{{\small Logistic}}
		\end{subfigure}
\end{figure}

\subsection{IV model}\label{ssec:iv} 

In the IV model of Example \ref{ex:IV},  $n$ i.i.d. copies of $W = (Y, X, Z)$ are observed where 
\begin{equation}\label{eq:iv-mdl-0}
    Y = X\transp \theta  +Z_1\transp\beta + \epsilon, \qquad \E[\epsilon |Z] = 0, \qquad Z = (Z_1\transp, Z_2\transp)\transp.
\end{equation}
Let $d_W\define d_{\theta} + d_Z + 1$.
With $\pi(Z)\define \E[X|Z]$ and $\upsilon = X - \pi(Z)$,
\begin{equation}\label{eq:iv-mdl}
    \begin{aligned}
        Y &= X\transp \theta + Z_1\transp\beta   + \epsilon\\
        X &= \pi(Z) + \upsilon
    \end{aligned},\qquad   \E[U|Z] = 0, \quad   U \define (\epsilon, \upsilon\transp)\transp.
    \end{equation}
If $\pi(Z)$ is constant the instruments $Z$ provide no information about $\theta$. 
Weak identification in this model can be very different from in the IV model with a linear first stage: there are many data configurations in which $\Var(\E[X|Z])$ may be ``large'' whilst $\E[XZ\transp]\E[ZZ\transp]^{-1}\approx 0$.
In such situations, tests which can exploit such non-linear identifying information can provide substantially more power than tests which (implicitly) use a linear first stage. In this section I develop a $\psi_{n, \theta_0}$ test which can capture such identifying information whilst remaining robust to weak identification.\footnote{This does not contradict optimality results that are known for, e.g., the AR test \cite[][]{M09, CHJ09} as these results assume a linear first stage. 
}$^,$ \footnote{An alternative approach to capturing this non-linear identifying information (whilst remaining robust to weak instruments) is to use a large number of transformations of the instruments, $f_1(Z), \ldots, f_M(Z)$, in a linear first stage, combined with a testing procedure which remains robust in the presence of many weak instruments. In the simulation study below, I compare the $\psi_{n, \theta_0}$ test to this approach, using the test of \cite{MS21}.}

Let $\zeta$ denote the density of $\xi \define (\epsilon, \upsilon\transp, Z\transp)$ with respect to a $\sigma$-finite measure $\nu$.
The parameters of the IV model are $\gamma = (\theta, \eta)$ with the nuisance parameters collected in $\eta = (\beta, \pi, \zeta)$. The density of one observation is 
\begin{equation}\label{eq:iv-dens}
    p_{\gamma}(W) = \zeta(Y - X\transp\theta - Z_1\transp \beta, X - \pi(Z), Z),
\end{equation}
with respect to a $\sigma$-finite measure $\tilde{\nu}$ and $P_{\gamma}$ denotes the corresponding measure. The model parameters are restricted as follows.

\begin{assumption}\label{ass:iv-mdl}
    The parameters $\gamma = (\theta, \beta, \pi, \zeta)\in \Gamma = \Theta \times \mc{B}\times \ms{P}\times \ms{Z}$ where
    \begin{enumerate}
        \item $\Theta$ is an open subset of $\R^{d_\theta}$ and $\mc{B}$ is an open subset of $\R^{d_\beta}$;
        \item $\ms{Z}$ is a subset of the set of density functions on $\R^{d_W}$ with respect to $\nu$;
        \item For $(\pi, \zeta)\in \ms{P}\times \ms{Z}$, if $\xi\define (U\transp, Z\transp)\transp$, then
        \begin{equation*}
            \E[U | Z] = 0,\qquad \E\|\xi\|^4 < \infty,\qquad  \E\|\pi(Z)\|^4 < \infty,\qquad \E\|\phi(\xi)\|^4 < \infty,
        \end{equation*}
        where $\phi_1 \define \nabla_{\epsilon}\log \zeta(\epsilon, \upsilon, Z)$, $\phi_2\define \nabla_{\upsilon}\log \zeta(\epsilon, \upsilon, Z)$ and $\phi \define(\phi_1, \phi_2\transp)\transp$.        
    \end{enumerate}
    Additionally,  $p_{\gamma}$ is a probability density for each $\gamma\in \Gamma$ with respect to a $\sigma$-finite measure $\tilde{\nu}$.
\end{assumption}

Assumption \ref{ass:iv-mdl} imposes the existence of certain moments and the (IV) conditional mean restriction. That $p_{\gamma}$ is a valid probability density holds automatically (with $\nu = \tilde{\nu}$) when $U|Z$ is continuously distributed.

\paragraph*{Local Asymptotic Normality} 
Consider local perturbations $P_{\gamma + \varphi_n(h)}$ for 
\begin{equation}\label{eq:iv-varphi}
    \varphi_n(h)\define \left(\frac{\tau}{\sqrt{n}},\,  \frac{b_0}{\sqrt{n}}, \varphi_{n, 1}(b_1),\, \varphi_{n, 2}(b_2)\right)~, \quad h = (\tau, b)\in H\define \R^{d_\theta}\times B,
\end{equation}
with $ B \define \R^{d_\beta} \times B_{1}\times B_{2}$. $B_{1}$ is a subset of the bounded functions $b_1:\R^{d_Z}\to \R^{d_\theta}$ and $B_{2}$ a subset of the functions $b_2:\R^{d_W}\to \R$ which are bounded and continuously differentiable in its first $1 + d_\theta$ components with bounded derivative and such that
\begin{equation}\label{eq:iv-b2-conditions}
    \E\left[b_2(U, Z)\right] = 0, \quad \E\left[U b_2(U, Z) | Z\right] = 0, \qquad \text{ for } (U\transp, Z\transp)\transp  \sim \zeta.
\end{equation}
The precise forms of $\varphi_{n, 1}, \varphi_{n, 2}$ are left unspecified. It is required only that the local perturbations satisfy LAN.\footnote{Examples of $\varphi_{n, 1}, \varphi_{n, 2}$ and $B$ for which Assumption \ref{ass:iv-LAN} holds are given in Section \ref{app:ssec:iv-LAN}. }

\begin{assumption}\label{ass:iv-LAN}
    Suppose that $\mc{W}_n = \prod_{i=1}^n \R^{d_W}$, 
    $P_{n, h}\define P_{\gamma + \varphi_n(h)}^n \ll \nu_n$ for all $\gamma\in \Gamma$ and $h\in H$ 
    and are such that Assumption \ref{ass:LAN} holds with
    \begin{equation}
        \log \frac{p_{n, h}}{p_{n,  0}} = \fracrootn\sumin [Ah](W_i) - \frac{1}{2}\sigma(h) + o_{P_{n, 0}}(1), \quad h\in H,
    \end{equation}
    where $\sigma(h) = \int [Ah]^2\darg{P}$ and $A$ is as in equation \eqref{eq:Agamh} with 
    \begin{align*}
        \dotscr(W)&\define -\phi_1(\epsilon(\theta, \beta), \upsilon(\pi),Z) X_1 \\
        [Db](W) &\define -\phi(\epsilon(\theta, \beta),\upsilon(\pi),Z)\transp \left[\begin{smallmatrix}
            b_0\transp Z_1 & b_1(Z)
        \end{smallmatrix}\right]
     + b_2(\epsilon(\theta, \beta), \upsilon(\pi), Z),
    \end{align*} 
    where $\epsilon(\theta, \beta)\define Y - X\transp\theta - Z_1\transp\beta$ and $\upsilon(\pi)\define X - \pi(Z)$. 
\end{assumption}

\paragraph*{The moment conditions}
The test will be based on moment conditions related to the efficient score function for $\theta$, $\effscr$. This is given in the following Lemma.\footnote{The last two conditions in \eqref{eq:iv-EphiU} hold if $\lim_{|u_i|\to\infty} |u_i|\zeta(u, z)=0$ for $i=1, 
\ldots, d_\alpha$. 
}

\begin{lemma}\label{lem:iv-effscr}
    Suppose Assumptions \ref{ass:iv-mdl}, \ref{ass:iv-LAN} hold, for  $J(Z)\define \E[UU\transp|Z]$,
    \begin{equation}\label{eq:iv-EphiU}
        \begin{aligned}
        0<c \le \lambda_{\min}(J(Z))\le 
        \lambda_{\max}(J(Z))
        \le C < \infty, &\qquad \lambda_{\min}(\E[Z_1Z_1\transp] )> 0,\\
        \E\left[\phi(\epsilon,\upsilon, Z)U\transp \middle| Z\right] = -I,&\qquad  \E[\phi_1(\epsilon, \upsilon, Z)\upsilon U\transp]= 0,
    \end{aligned}
    \end{equation}
    and that $B_1$ is dense in $L_2$. Define $\upomega(Z)\define \E[\epsilon^2|Z]^{-1}$. The efficient score for $\theta$ is 
    \begin{equation}\label{eq:iv-effscr}
        \effscr(W) = \upomega(Z)(Y - X\transp\theta - Z_1\transp\beta)\left[\pi(Z) - \E[\upomega(Z)XZ_1\transp]\E[\upomega(Z)Z_1Z_1\transp]^{-1}Z_1\right].
    \end{equation}
\end{lemma}
For simplicity, I will use the moment functions 
\begin{equation}\label{eq:iv-bar-effscr}
    g(W)\define  \E[\epsilon^2]^{-1}(Y - X\transp\theta - Z_1\transp\beta)\left[\pi(Z) - \E[XZ_1\transp]\E[Z_1Z_1\transp]^{-1}Z_1\right]. 
\end{equation}
$g$ belongs to the orthocomplement of $\{Db: b\in B\}$ and coincides with the efficient score function when $J(Z) = \E[UU\transp]$ a.s. (i.e. under homoskedasticity).\footnote{Nevertheless, homoskedasticity is \emph{not} assumed and the results below hold under heteroskedasticity. For full efficiency one could base the test on \eqref{eq:iv-effscr}. This is left for future work.
} 

\begin{lemma}\label{lem:iv-bar-effscr}
    Suppose that Assumptions \ref{ass:iv-mdl}, \ref{ass:iv-LAN} and equation \eqref{eq:iv-EphiU} hold. Then, the moment conditions $g \in \{Db: b\in B\}^\perp$. If $\E[\epsilon^2|Z] = \E[\epsilon^2]$ a.s., then $g = \effscr$ a.s.. 
\end{lemma}

\begin{proposition}\label{prop:iv-joint-conv}
    Suppose that Assumptions \ref{ass:iv-mdl}, \ref{ass:iv-LAN} and equation \eqref{eq:iv-EphiU} hold. Then Assumption \ref{ass:joint-conv} is satisfied with 
    $g_{n} = \G_n g$.
\end{proposition}

\paragraph*{A feasible test}

Suppose that $\hat{\beta}_n$ and $\hat{\pi}_{n, i}(Z_i)$ are estimators of $\beta$ and $\pi(Z_i)$ respectively. Let the $i$-th residual in \eqref{eq:iv-mdl-0} based on $\theta = \theta_0$ and $\hat\beta_n$ be $\hat{\epsilon}_{n, i} \define Y_i - X_i\transp \theta - Z_{1, i}\transp\hat{\beta}_n$. Let $\hat{s}_n\define \meanin \hat{\epsilon}_{n, i}^2$ and define

\begin{equation}\label{eq:iv-bar-effscr-est}
   \hat{g}_{n, \theta, i}\define \hat{s}_{n}^{-1}\hat{\epsilon}_{n, i} \left[\hat{\pi}_n(Z_i)  - \left[\meanin X_i Z_{1, i}\transp\right] \left[\meanin Z_{1,i} Z_{1, i}\transp\right]^{-1} Z_i\right],
\end{equation}
and 
\begin{equation}\label{eq:iv-bar-effinfo-est-check}
    \check{V}_{n, \theta} \define \meanin \hat{g}_{n, \theta, i}\hat{g}_{n, \theta, i}\transp.
\end{equation}

Based on $\check{V}_{n, \theta}$, form $\hat{V}_{n, \theta}$ according to the construction in Section S5 of \cite{LM21-S} using a truncation rate $\upnu_n$, set $\hat{\Lambda}_{n, \theta}\define \hat{V}_{n, \theta}^\dagger$ and $\hat{r}_{n, \theta}\define \rank(\hat{V}_{n, \theta})$.

The following assumption provides sufficient high-level conditions on the estimators $\hat\beta_n$ and $\hat{\pi}_{n, i}(Z_i)$ such that Assumption \ref{ass:consistent} holds. These conditions are compatible with $\hat{\pi}_{n, i}$ being a leave-one-out series estimator.\footnote{The discretisation of $\hat\beta_n$ is a technical device 
which permits the proof to go through under weaker conditions \cite[cf.][Chapter 6]{LCY00}. This 
can be arranged given a $\sqrt{n}$ -- consistent initial estimator, by replacing its value with the closest point in the set $\ms{S}_n$.}$^,\,$\footnote{See e.g. \cite{BCCK15} for sufficient conditions for \eqref{ass:iv-est2:eq:pi-est-bound} and Section \ref{sm:sec:IV} for a discussion of \eqref{ass:iv-est2:eq:pi-est-bound2}.}

\begin{assumption}\label{ass:iv-est2}
    Suppose that, given $\theta_0$, (i) $\hat{\beta}_n$ is an estimator valued in $\ms{S}_n\define \{CZ / \sqrt{n}: Z\in \Z^{d_{\beta}}\}$ for some $C\in \R^{d_\beta\times d_\beta}$ and satisfying  $\sqrt{n}(\hat{\beta}_n - \beta) = O_{P_{n, 0}}(1)$ and (ii) $\hat{\pi}_{n, i}(Z_i)$ are estimators such that $\hat{\pi}_{n, i}(Z_i)$ is $\sigma(Z_i, \mc{C}_{n,-i})$ measurable for $\mc{C}_{n,-i}\define \{W_j: j=1, \ldots, n,\, j\neq i\}$, and on events $F_n$ with $P_{n, 0}(F_n)\to 1$, 
    \begin{equation}\label{ass:iv-est2:eq:pi-est-bound}
        \left[\int \left\|\hat{\pi}_{n, i}(z) - \pi(z)\right\|^2 \darg{\zeta_Z(z)}\right]^{1/2} \le \delta_n = o(1),
    \end{equation}
    where $\zeta_Z$ is the marginal distribution of $Z$ and for each $k=1, \ldots, d_\theta$, $i\neq j$,
    \begin{equation}\label{ass:iv-est2:eq:pi-est-bound2}
        \E\left[\bm{1}_{F_n}\bm{1}_{G_n}(\hat{\pi}_{n, i, k}(Z_i) - \pi_k(Z_i))(\hat{\pi}_{n, j, k}(Z_j) - \pi_k(Z_j))\transp\epsilon_i\epsilon_j\right]\lesssim \delta_n^2 / n, \  P_{n, 0}(G_n)\to 1.
    \end{equation}
    Suppose also $\delta_n^2 + n^{-1/2}  = o(\upnu_n)$, \eqref{eq:iv-EphiU} holds and $\E\left[\epsilon^4(\|\pi(Z)\| + \|Z_1\|)^4\right] < \infty$.
\end{assumption}

There is no requirement on the rate $\delta_n$ in \eqref{ass:iv-est2:eq:pi-est-bound}, \eqref{ass:iv-est2:eq:pi-est-bound2} beyond $\delta_n = o(1)$.

\begin{proposition}\label{prop:iv-estimation-bar-effscr2}
    Suppose that Assumptions \ref{ass:iv-mdl}, \ref{ass:iv-LAN}, \& \ref{ass:iv-est2} hold. Then Assumption \ref{ass:consistent} holds with $\hat{g}_{n, \theta} \define \fracrootn\sumin \hat{g}_{n, \theta, i}$, $g_{n} \define \G_n g$ and $\hat{\Lambda}_{n, \theta}$ defined below equation \eqref{eq:iv-bar-effinfo-est-check}.
\end{proposition}

A consequence of Assumption \ref{ass:iv-LAN} and Propositions \ref{prop:iv-joint-conv} and \ref{prop:iv-estimation-bar-effscr2} is that the test $\psi_{n, \theta_0}$ formed as in \eqref{eq:psi-test} is locally regular by Theorem \ref{thm:psi-pwr-local-alt}.

\paragraph{Simulation study}
I test $H_0: \theta=\theta_0 = 0$ at a nominal level of 5\%. Each study reports the results of 5000 monte carlo replications with a sample size of $n\in \{200, 400, 600\}$.\footnote{The power surfaces in Design 1 are computed with 2500 replications.} Two simulation designs are considered. 

Design 1 is a bivariate, just identified design. Here $d_\theta = 2$ and $Z_2$ is drawn from a zero-mean multivariate normal distribution with covariance matrix $\left[\begin{smallmatrix}
    1 & 0.4\\
    0.4 & 1
\end{smallmatrix}\right]$. The error terms $\epsilon, \upsilon$ are drawn from a zero-mean multivariate normal such that each has variance 1 and the covariances are $\Cov(\epsilon, \upsilon_i) = 0.9$ and $\Cov(\upsilon_1,\upsilon_2) = 0.7$. $Z_1=1$ with $\beta = 1$ and $\pi(Z) = \pi(Z_2) = (\pi_1(Z_{2,1}), \pi_2(Z_{2, 2}))\transp$ with each $\pi_i$ ($i=1, 2$) being one of the exponential or logistic functions $f_j$ in Table \ref{tbl:SIM-index-fcns}. The exponential form is a prototypical function shape for which the linear projection of $X$ on $Z$ will provide essentially no identifying information; for the logistic form this linear projection should perform well.\footnote{These functions are plotted in Figures \ref{fig:SIM-gaussian-fs} and \ref{fig:SIM-logistic-fs}. The separation  $\pi(Z_2) = (\pi_1(Z_{2,1}), \pi_2(Z_{2, 2}))\transp$ is assumed unknown and is not imposed in the estimation of $\pi$.
}$^,\, $\footnote{Results for the case where $\pi(Z_2)$ is linear are very similar to the ``approximately linear'' logistic case and are therefore unreported.}

I consider the $\psi_{n, \theta_0}$ test developed above, with a leave-one-out series estimator of $\pi$ based on (tensor product) Legendre polynomials. I consider both  fixing the number of polynomials at $k=3$ in each of the univariate series which form the tensor product basis and choosing $k\in \{3, 4, 5, 6, 7\}$ using information criteria. $\upnu$ is set to $10^{-2}$.
I additionally consider the \cite{AR49} (AR) test.\footnote{The AR test is computed with $Z_2$ as instruments, after partialling out $Z_1$. 
}$^,$\footnote{I do not consider alternative weak instrument robust tests based on a linear first stage (e.g. LM, CLR) in this design as the AR test is known to be optimal when the model is just-identified.}  

The empirical rejection frequencies under the null are shown in Tables \ref{tbl:IV-size1i} and \ref{tbl:IV-size1ii}. The parameter $j$ controls the level of identification: the larger is $j$ the closer $\pi_j$ is to a constant function. In each specification all the considered tests reject at close to the nominal level. Power surfaces for the $\psi_{n, \theta_0}$ and AR tests are shown in figures \ref{fig:IV-power-1-exp-exp-AR} -- \ref{fig:IV-power-1-exp-log-psi}. As can be seen in these figures, the $\psi_{n, \theta_0}$ test is able to detect deviations from the null when $\pi_j$ has the exponential form, unlike the AR test.\footnote{
One could consider an AR test using e.g. some basis functions $f_1(Z_{2}), \ldots, f_K(Z_{2})$ however as noted in \cite[][p. 2669]{MS21}, the AR statistic is not well behaved for large $K$. The jackknife AR test of \cite{MS21} applies only to the case where $d_\theta=1$. 
} For the logistic form, the power of the two tests is similar. Unsurprisingly, neither test provides non-trivial power when identification is very weak.

\begin{table*}[!htbp]
	\footnotesize
	\setlength{\tabcolsep}{4.5pt}
	\begin{center}
		\caption{\label{tbl:IV-size1i} Empirical rejection frequencies, IV, Design 1}
		\begin{threeparttable}
			
\begin{tabular}{rrrrrrrrrrrrrr}
\toprule
\multicolumn{1}{c}{} & \multicolumn{1}{c}{} & \multicolumn{4}{c}{Exponential - Exponential} & \multicolumn{4}{c}{Logistic - Logistic} & \multicolumn{3}{c}{Exponential - Logistic} \\
\cmidrule(l{3pt}r{3pt}){3-6} \cmidrule(l{3pt}r{3pt}){7-10} \cmidrule(l{3pt}r{3pt}){11-13}
\multicolumn{1}{c}{} & \multicolumn{1}{c}{} & \multicolumn{1}{c}{AR} & \multicolumn{3}{c}{$\psi$} & \multicolumn{1}{c}{AR} & \multicolumn{3}{c}{$\psi$} & \multicolumn{1}{c}{AR} & \multicolumn{3}{c}{$\psi$} \\
\cmidrule(l{3pt}r{3pt}){3-3} \cmidrule(l{3pt}r{3pt}){4-6} \cmidrule(l{3pt}r{3pt}){7-7} \cmidrule(l{3pt}r{3pt}){8-10} \cmidrule(l{3pt}r{3pt}){11-11} \cmidrule(l{3pt}r{3pt}){12-14}
$n$ & $j$ &  & $k=3$ & AIC & BIC &  & $k=3$ & AIC & BIC &  & $k=3$ & AIC & BIC\\
\midrule
200 & 1 & 5.52 & 5.10 & 4.00 & 5.18 & 5.52 & 4.88 & 3.40 & 4.48 & 5.52 & 4.56 & 3.76 & 4.68\\
200 & 2 & 5.52 & 6.24 & 6.26 & 6.24 & 5.52 & 5.46 & 5.30 & 5.46 & 5.52 & 5.78 & 5.86 & 5.78\\
200 & 3 & 5.52 & 8.36 & 7.80 & 8.36 & 5.52 & 8.44 & 7.98 & 8.44 & 5.52 & 7.92 & 7.74 & 7.92\\
\addlinespace
400 & 1 & 5.60 & 4.96 & 4.48 & 4.80 & 5.60 & 4.96 & 4.14 & 4.78 & 5.60 & 4.88 & 4.34 & 4.80\\
400 & 2 & 5.60 & 6.12 & 6.20 & 6.12 & 5.60 & 5.68 & 5.52 & 5.68 & 5.60 & 6.08 & 6.02 & 6.08\\
400 & 3 & 5.60 & 6.76 & 6.94 & 6.76 & 5.60 & 7.72 & 7.60 & 7.72 & 5.60 & 6.46 & 6.60 & 6.46\\
\addlinespace
600 & 1 & 5.38 & 5.30 & 4.90 & 5.22 & 5.38 & 5.20 & 3.80 & 5.16 & 5.38 & 5.32 & 4.64 & 5.04\\
600 & 2 & 5.38 & 6.08 & 6.10 & 6.08 & 5.38 & 4.98 & 4.98 & 4.98 & 5.38 & 5.56 & 5.76 & 5.56\\
600 & 3 & 5.38 & 4.36 & 4.78 & 4.36 & 5.38 & 5.02 & 5.50 & 5.02 & 5.38 & 4.40 & 5.06 & 4.40\\
\bottomrule
\end{tabular}

			\begin{tablenotes}
			\footnotesize
			\item \notes{E.g. ``Exponential - Logistic'' indicates that $\pi_1, \pi_2$ have the exponential and logistic form in Table \ref{tbl:SIM-index-fcns} respectively, with $c_j$ corresponding to column $j$.
			}
			\end{tablenotes}	
		\end{threeparttable}
	\end{center}
\end{table*}

\begin{table*}[!htbp]
	\setlength{\tabcolsep}{4.5pt}
	\footnotesize
	\begin{center}
		\caption{\label{tbl:IV-size1ii} Empirical rejection frequencies, IV, Design 1}
		\begin{threeparttable}
			
\begin{tabular}{rlrrrrrrrrrrrr}
\toprule
\multicolumn{1}{c}{} & \multicolumn{1}{c}{} & \multicolumn{4}{c}{Exponential - Exponential} & \multicolumn{4}{c}{Logistic - Logistic} & \multicolumn{4}{c}{Exponential - Logistic} \\
\cmidrule(l{3pt}r{3pt}){3-6} \cmidrule(l{3pt}r{3pt}){7-10} \cmidrule(l{3pt}r{3pt}){11-14}
\multicolumn{1}{c}{} & \multicolumn{1}{c}{} & \multicolumn{1}{c}{AR} & \multicolumn{3}{c}{$\psi$} & \multicolumn{1}{c}{AR} & \multicolumn{3}{c}{$\psi$} & \multicolumn{1}{c}{AR} & \multicolumn{3}{c}{$\psi$} \\
\cmidrule(l{3pt}r{3pt}){3-3} \cmidrule(l{3pt}r{3pt}){4-6} \cmidrule(l{3pt}r{3pt}){7-7} \cmidrule(l{3pt}r{3pt}){8-10} \cmidrule(l{3pt}r{3pt}){11-11} \cmidrule(l{3pt}r{3pt}){12-14}
$n$ & $j_1-j_2$ &  & $k=3$ & AIC & BIC &  & $k=3$ & AIC & BIC &  & $k=3$ & AIC & BIC\\
\midrule
200 & 1 - 3 & 5.52 & 6.36 & 5.30 & 6.48 & 5.52 & 7.18 & 6.40 & 7.12 & 5.52 & 6.34 & 5.40 & 6.30\\
200 & 2 - 3 & 5.52 & 6.26 & 6.44 & 6.26 & 5.52 & 6.94 & 6.78 & 6.94 & 5.52 & 6.46 & 6.64 & 6.46\\
200 & 3 - 3 & 5.52 & 8.36 & 7.80 & 8.36 & 5.52 & 8.44 & 7.98 & 8.44 & 5.52 & 7.92 & 7.74 & 7.92\\
\addlinespace
400 & 1 - 3 & 5.60 & 5.70 & 5.38 & 5.68 & 5.60 & 6.24 & 5.82 & 6.20 & 5.60 & 6.14 & 5.54 & 5.98\\
400 & 2 - 3 & 5.60 & 6.22 & 6.48 & 6.22 & 5.60 & 6.24 & 6.40 & 6.24 & 5.60 & 6.74 & 6.90 & 6.74\\
400 & 3 - 3 & 5.60 & 6.76 & 6.94 & 6.76 & 5.60 & 7.72 & 7.60 & 7.72 & 5.60 & 6.46 & 6.60 & 6.46\\
\addlinespace
600 & 1 - 3 & 5.38 & 5.66 & 5.40 & 5.60 & 5.38 & 5.46 & 4.52 & 5.34 & 5.38 & 5.58 & 5.26 & 5.36\\
600 & 2 - 3 & 5.38 & 6.00 & 6.10 & 6.00 & 5.38 & 5.40 & 5.40 & 5.40 & 5.38 & 6.12 & 6.32 & 6.12\\
600 & 3 - 3 & 5.38 & 4.36 & 4.78 & 4.36 & 5.38 & 5.02 & 5.50 & 5.02 & 5.38 & 4.40 & 5.06 & 4.40\\
\bottomrule
\end{tabular}

			\begin{tablenotes}
			\footnotesize
			\item \notes{E.g. ``Exponential - Logistic'' indicates that $\pi_1, \pi_2$ have the exponential and logistic form in Table \ref{tbl:SIM-index-fcns} respectively, with $c_{j_1}$ and $c_{j_2}$ corresponding to column $j_1$ - $j_2$.
			}
			\end{tablenotes}	
		\end{threeparttable}
	\end{center}
\end{table*}

\begin{figure}[!htbp]
		\centering
		\caption{\label{fig:IV-power-1-exp-exp-AR} IV Design 1, AR power, $\pi_i$ exponential}
		\begin{subfigure}[b]{0.3\textwidth}
			\centering
			\includegraphics[width=0.99\textwidth]{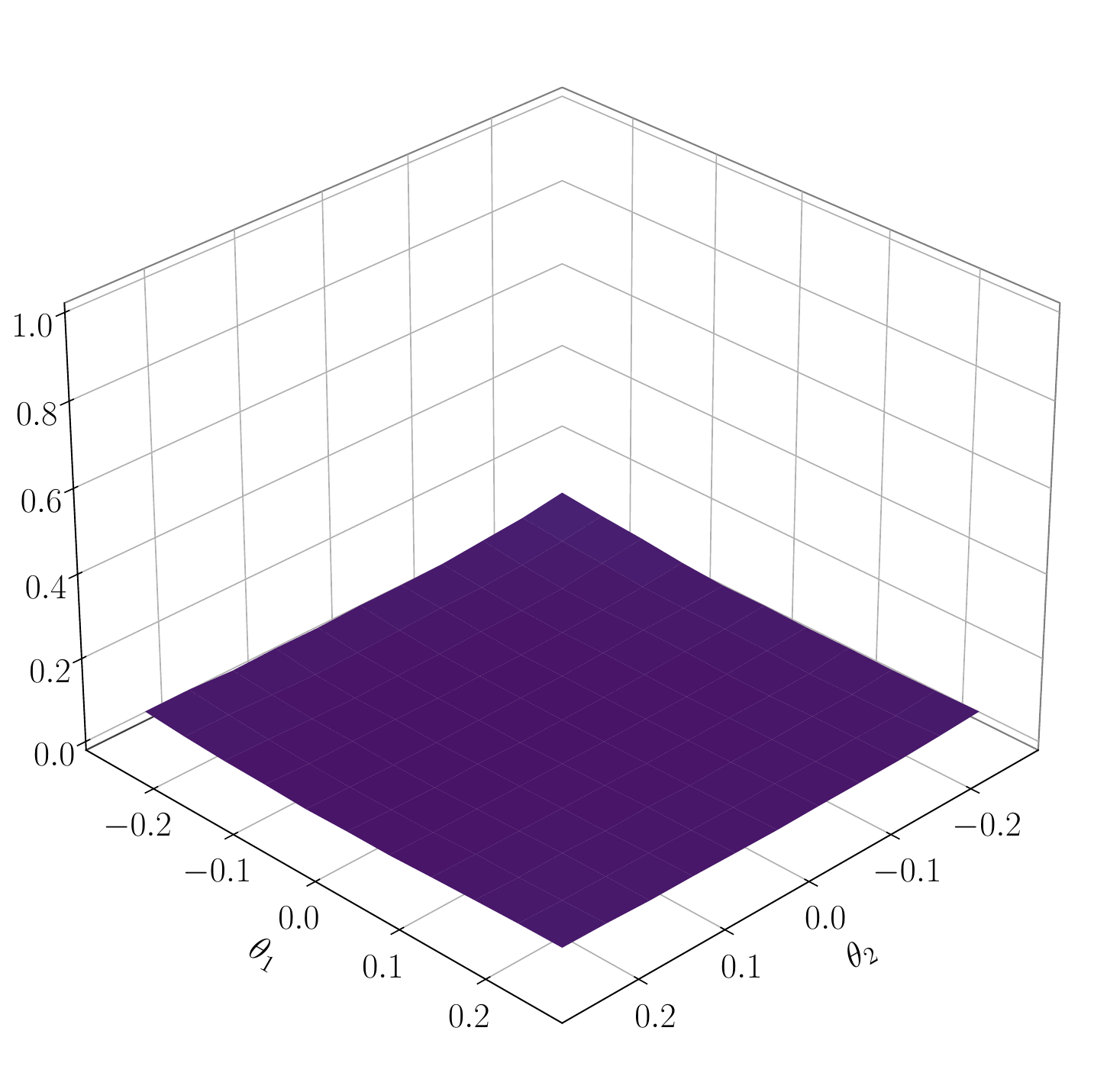}
			\caption[]%
			{{\small $j_1=1$, $j_2=1$}}
		\end{subfigure}
		\hfill
		\begin{subfigure}[b]{0.3\textwidth}
			\centering
			\includegraphics[width=0.99\textwidth]{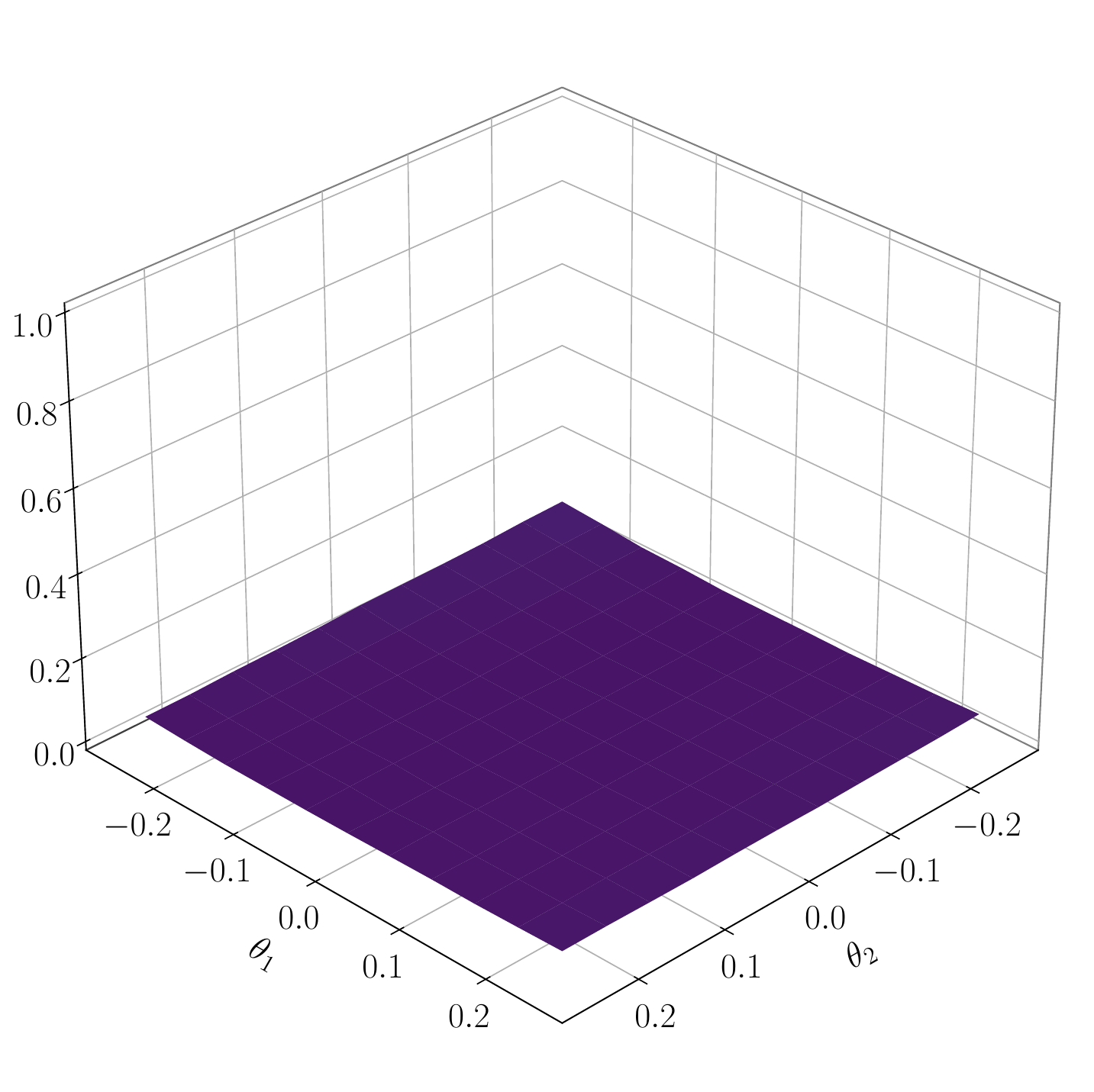}
			\caption[]%
			{{\small $j_1=1$, $j_2=3$}}
		\end{subfigure}
		\hfill
		\begin{subfigure}[b]{0.3\textwidth}
			\centering
			\includegraphics[width=0.99\textwidth]{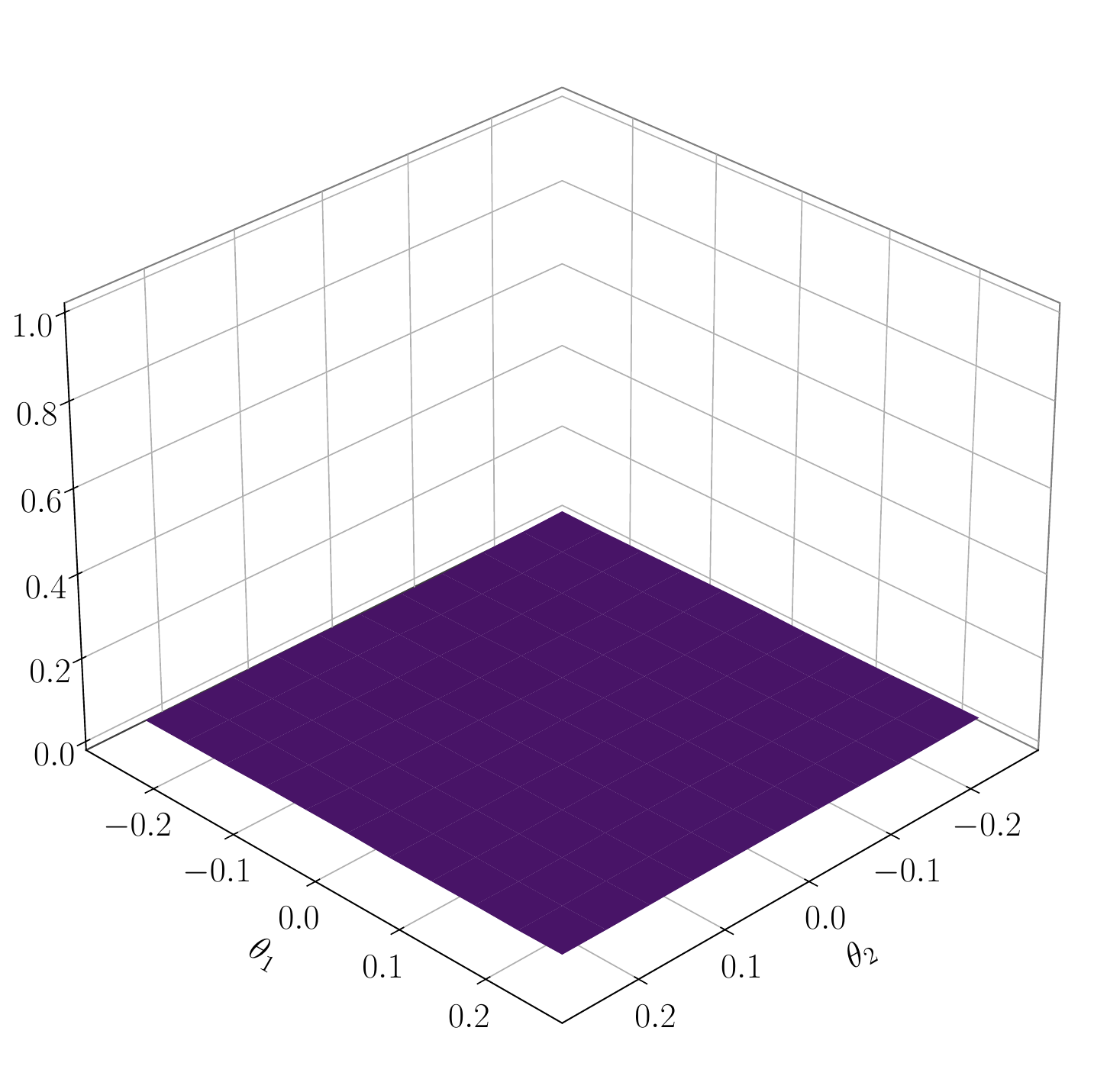}
			\caption[]%
			{{\small $j_1=3$, $j_2=3$}}
		\end{subfigure}
\end{figure}

\begin{figure}[!htbp]
		\centering
		\caption{\label{fig:IV-power-1-exp-exp-psi} IV Design 1, $\psi$ $(k=3)$ power, $\pi_i$ exponential}
		\begin{subfigure}[b]{0.3\textwidth}
			\centering
			\includegraphics[width=0.99\textwidth]{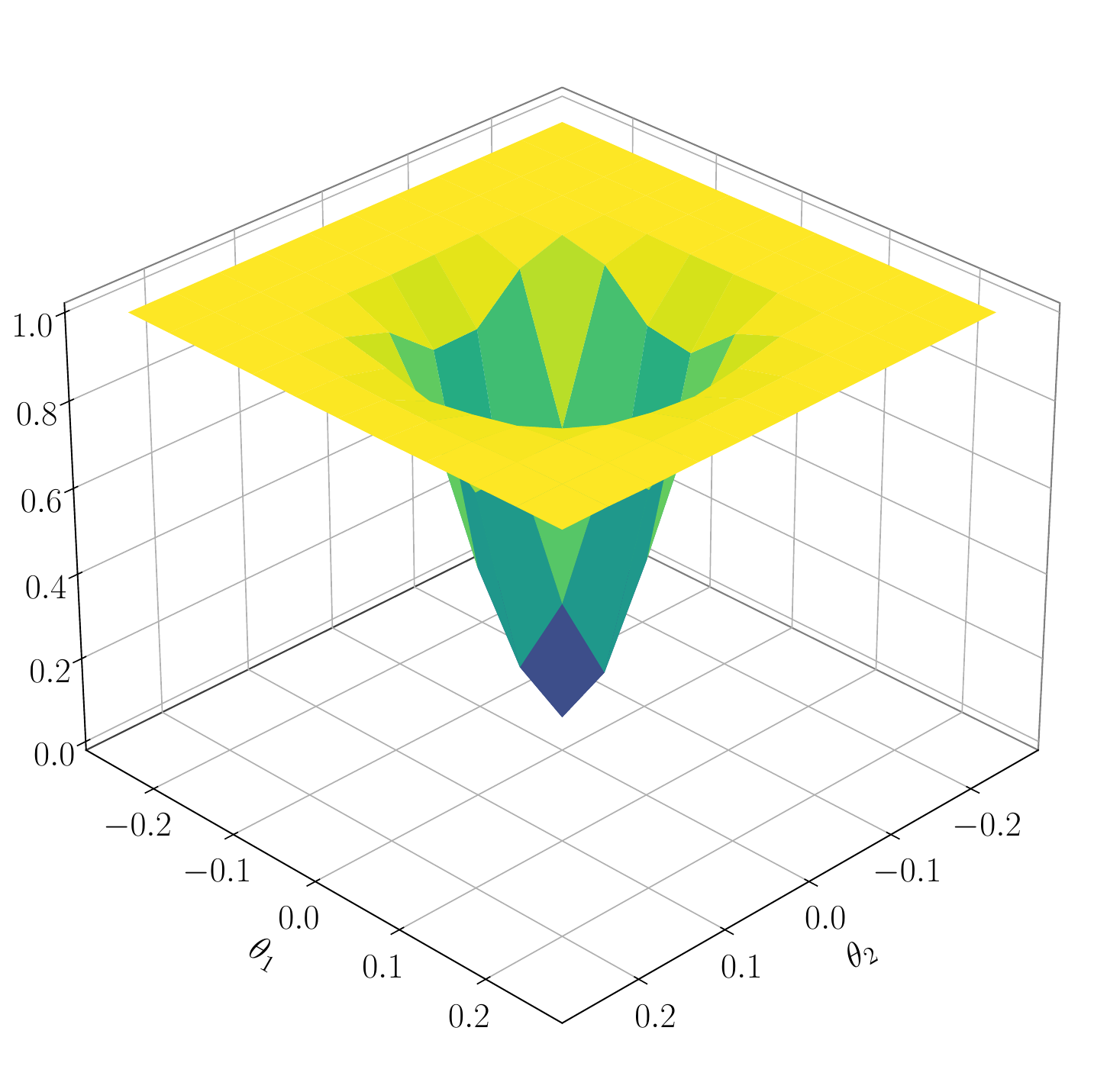}
			\caption[]%
			{{\small $j_1=1$, $j_2=1$}}
		\end{subfigure}
		\hfill
		\begin{subfigure}[b]{0.3\textwidth}
			\centering
			\includegraphics[width=0.99\textwidth]{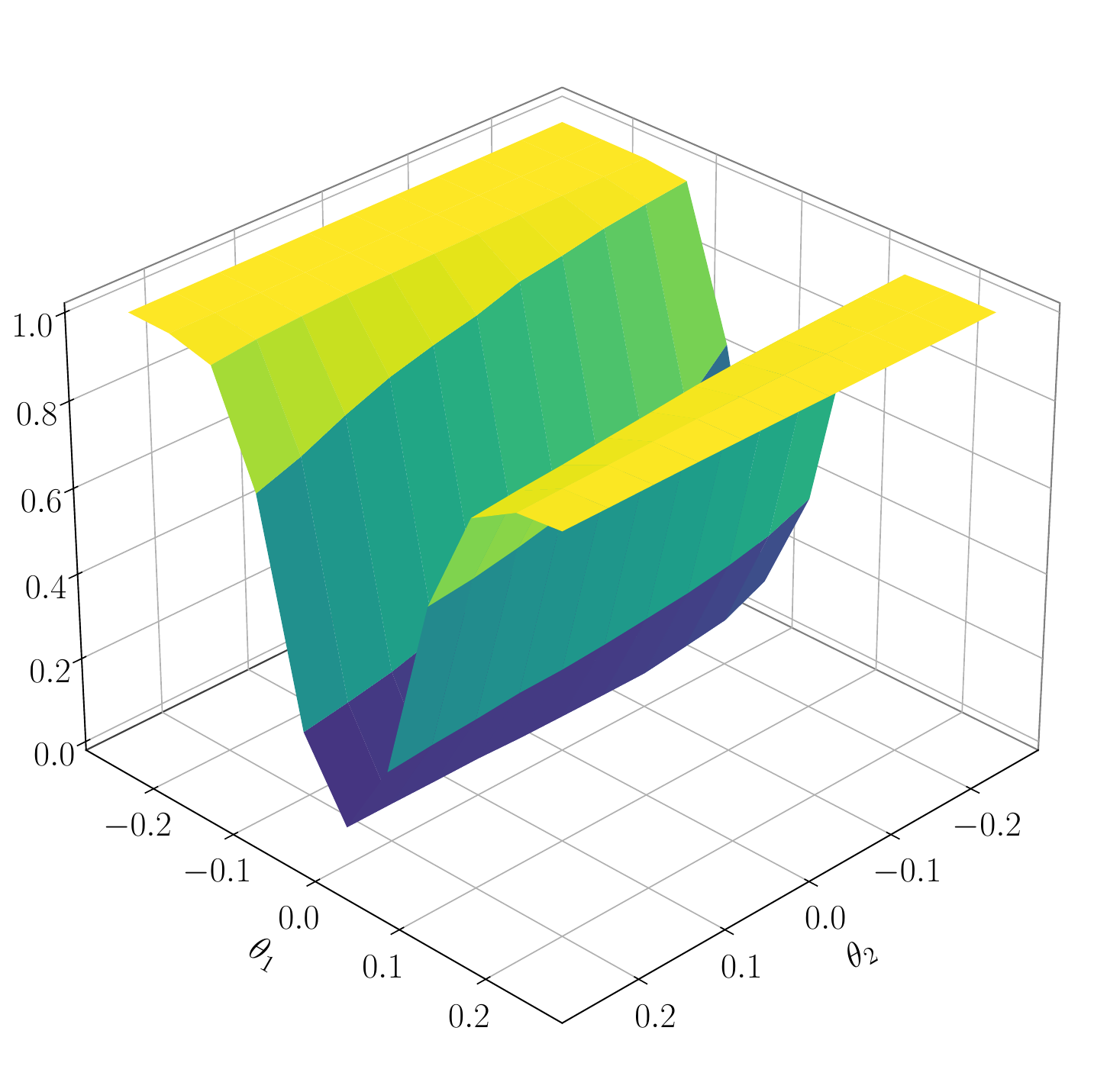}
			\caption[]%
			{{\small $j_1=1$, $j_2=3$}}
		\end{subfigure}
		\hfill
		\begin{subfigure}[b]{0.3\textwidth}
			\centering
			\includegraphics[width=0.99\textwidth]{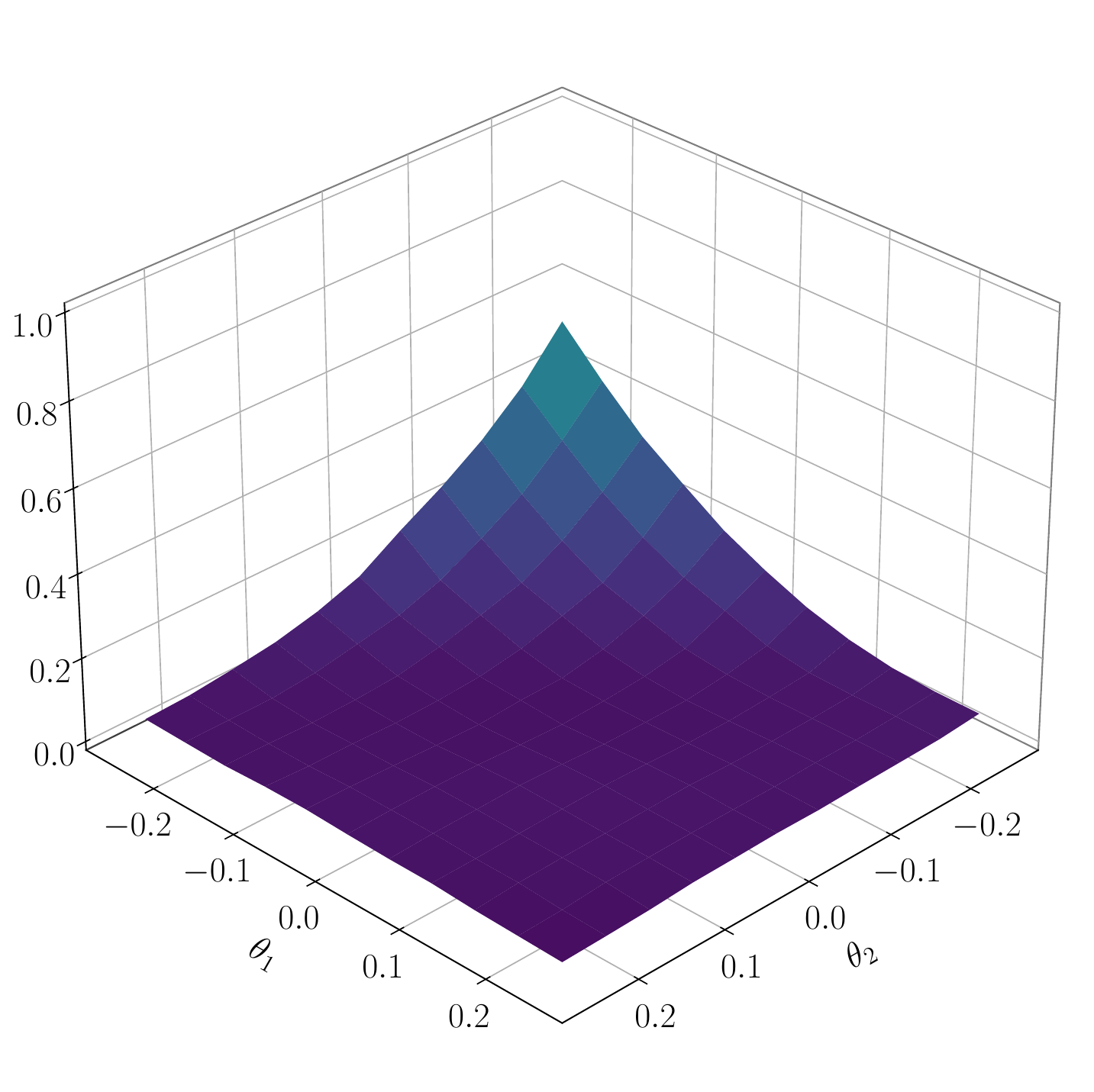}
			\caption[]%
			{{\small $j_1=3$, $j_2=3$}}
		\end{subfigure}
\end{figure}

\begin{figure}[!htbp]
		\centering
		\caption{\label{fig:IV-power-1-log-log-AR} IV Design 1, AR power, $\pi_i$ logistic}
		\begin{subfigure}[b]{0.3\textwidth}
			\centering
			\includegraphics[width=0.99\textwidth]{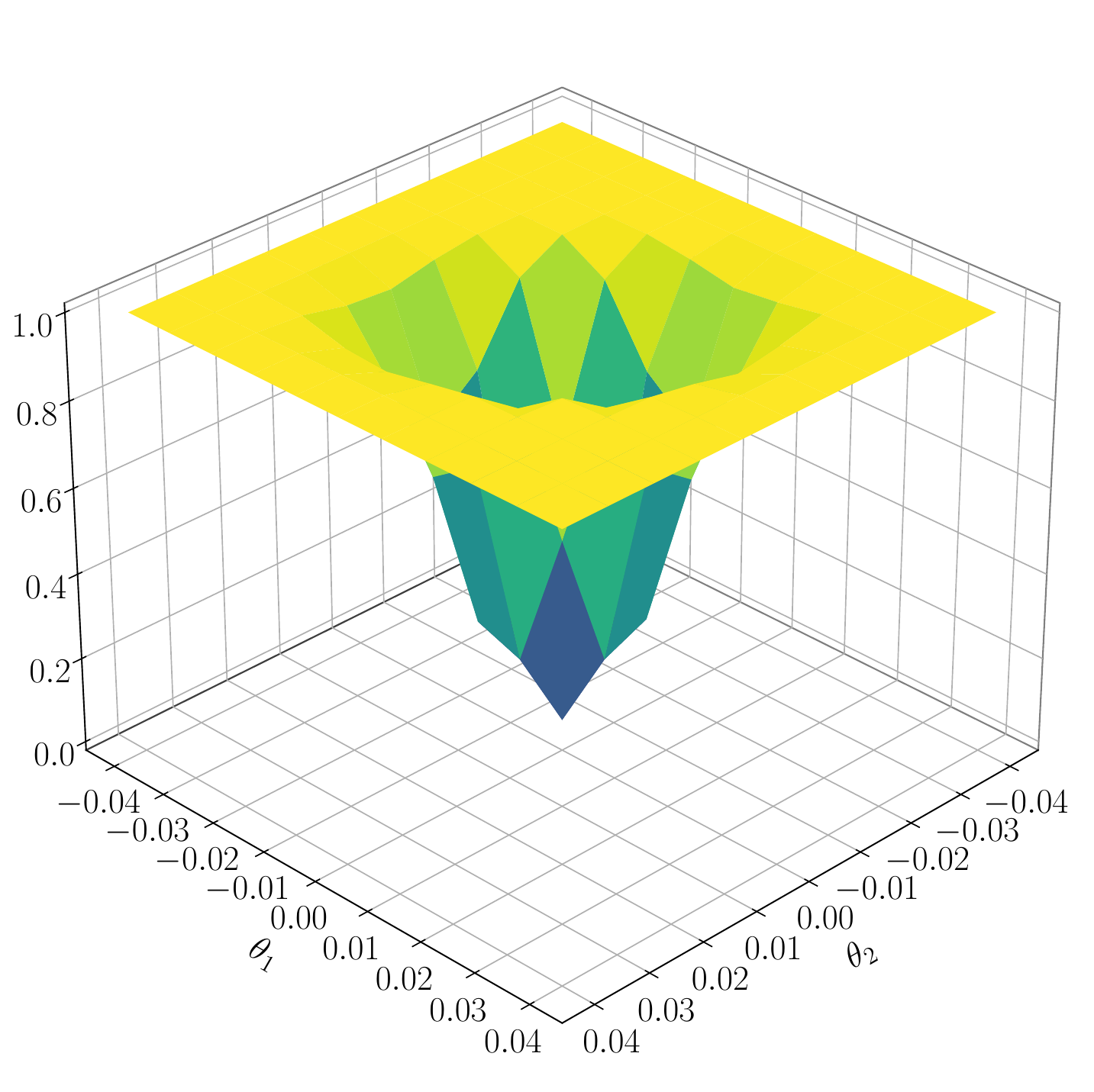}
			\caption[]%
			{{\small $j_1=1$, $j_2=1$}}
		\end{subfigure}
		\hfill
		\begin{subfigure}[b]{0.3\textwidth}
			\centering
			\includegraphics[width=0.99\textwidth]{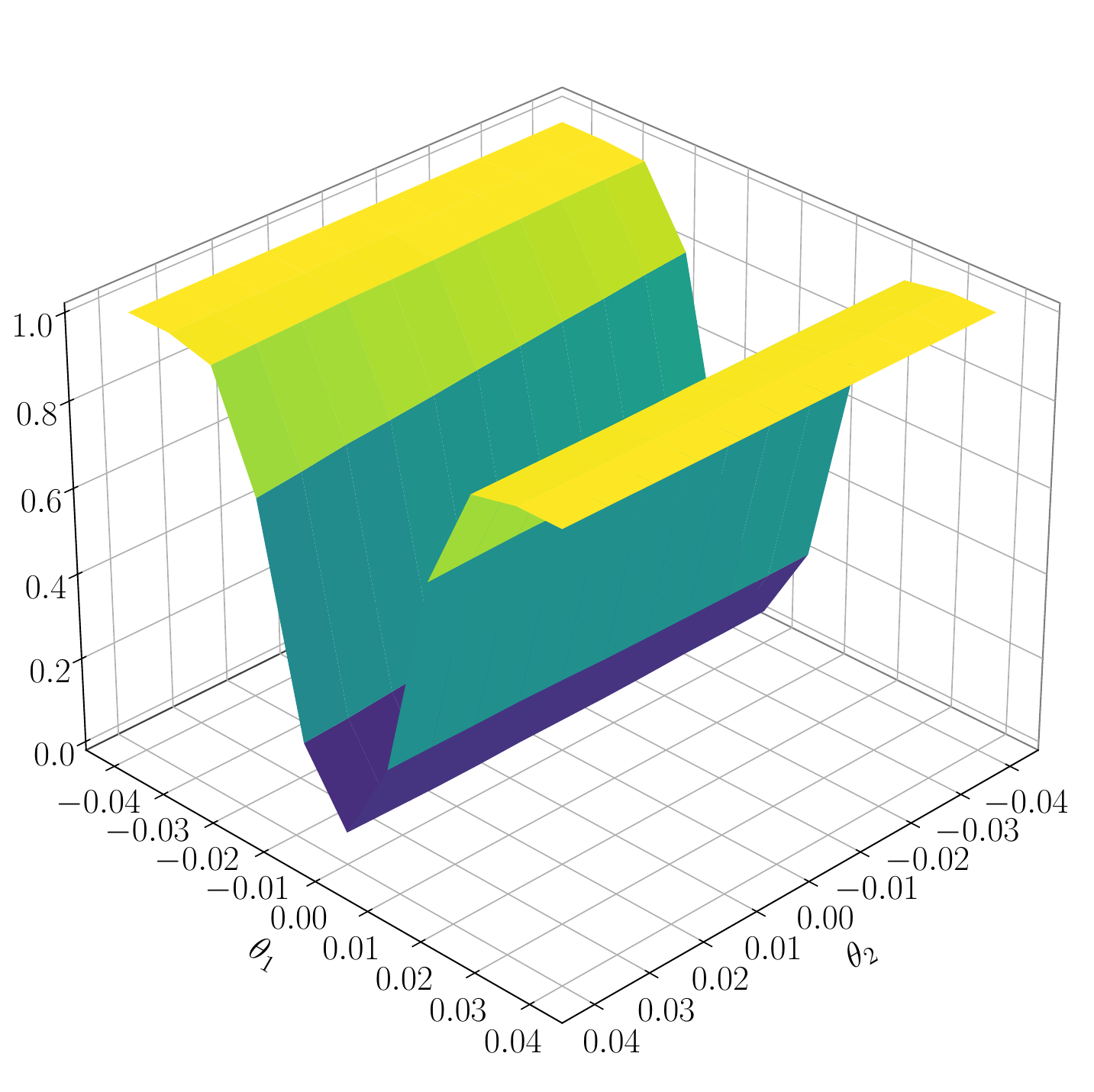}
			\caption[]%
			{{\small $j_1=1$, $j_2=3$}}
		\end{subfigure}
		\hfill
		\begin{subfigure}[b]{0.3\textwidth}
			\centering
			\includegraphics[width=0.99\textwidth]{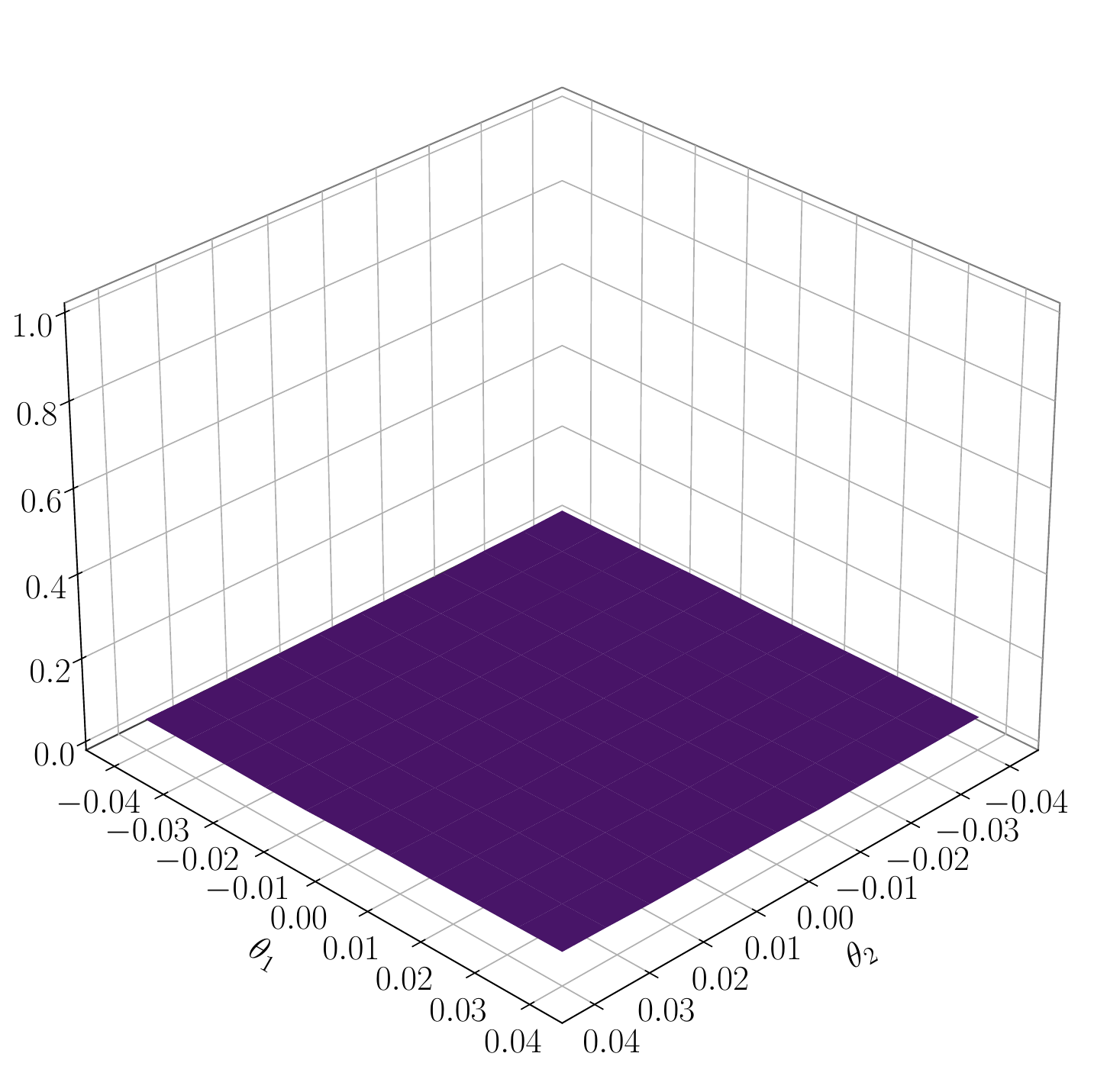}
			\caption[]%
			{{\small $j_1=3$, $j_2=3$}}
		\end{subfigure}
\end{figure}

\begin{figure}[!htbp]
		\centering
		\caption{\label{fig:IV-power-1-log-log-psi} IV Design 1, $\psi$ $(k=3)$ power, $\pi_i$ logistic}
		\begin{subfigure}[b]{0.3\textwidth}
			\centering
			\includegraphics[width=0.99\textwidth]{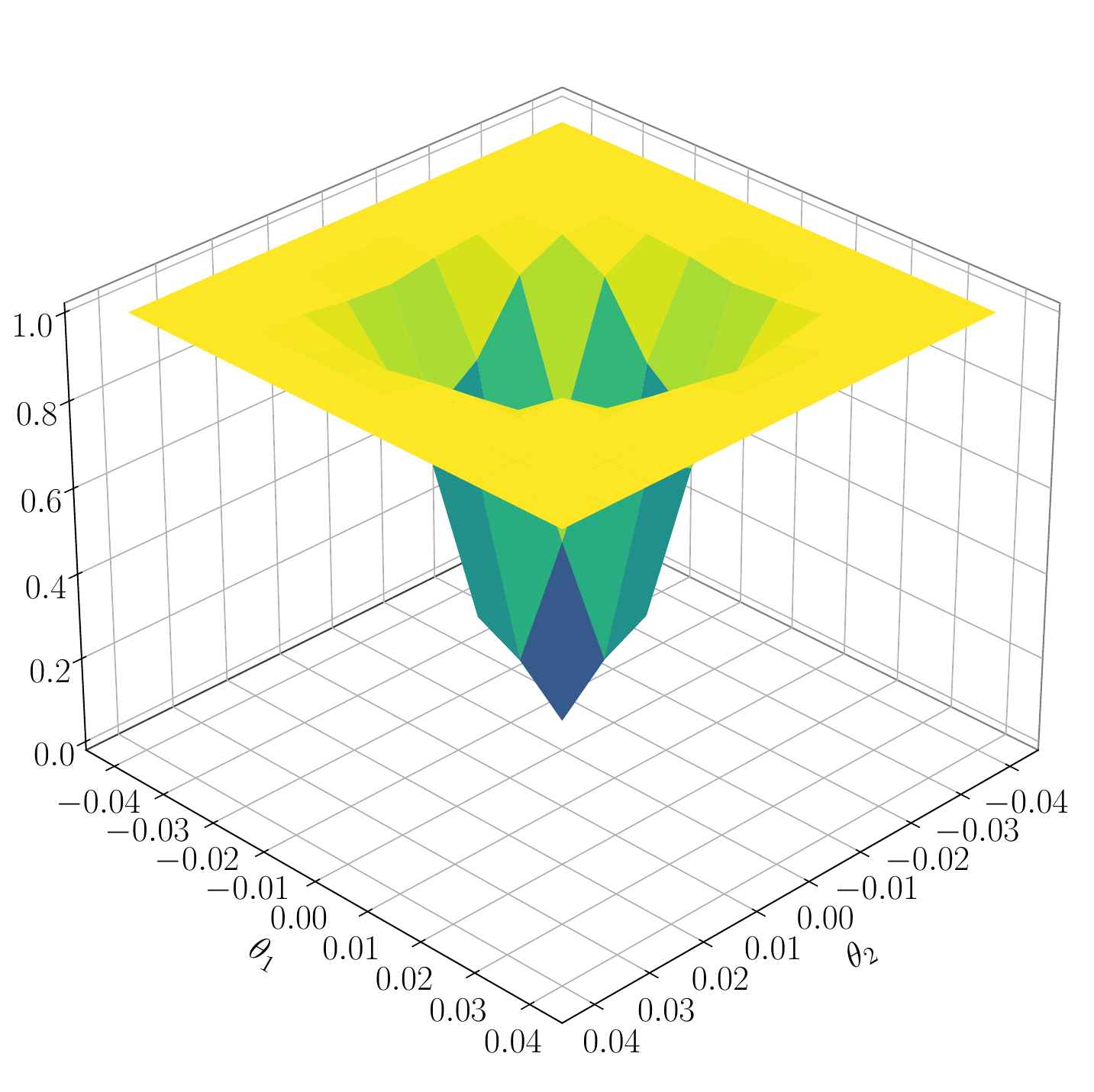}
			\caption[]%
			{{\small $j_1=1$, $j_2=1$}}
		\end{subfigure}
		\hfill
		\begin{subfigure}[b]{0.3\textwidth}
			\centering
			\includegraphics[width=0.99\textwidth]{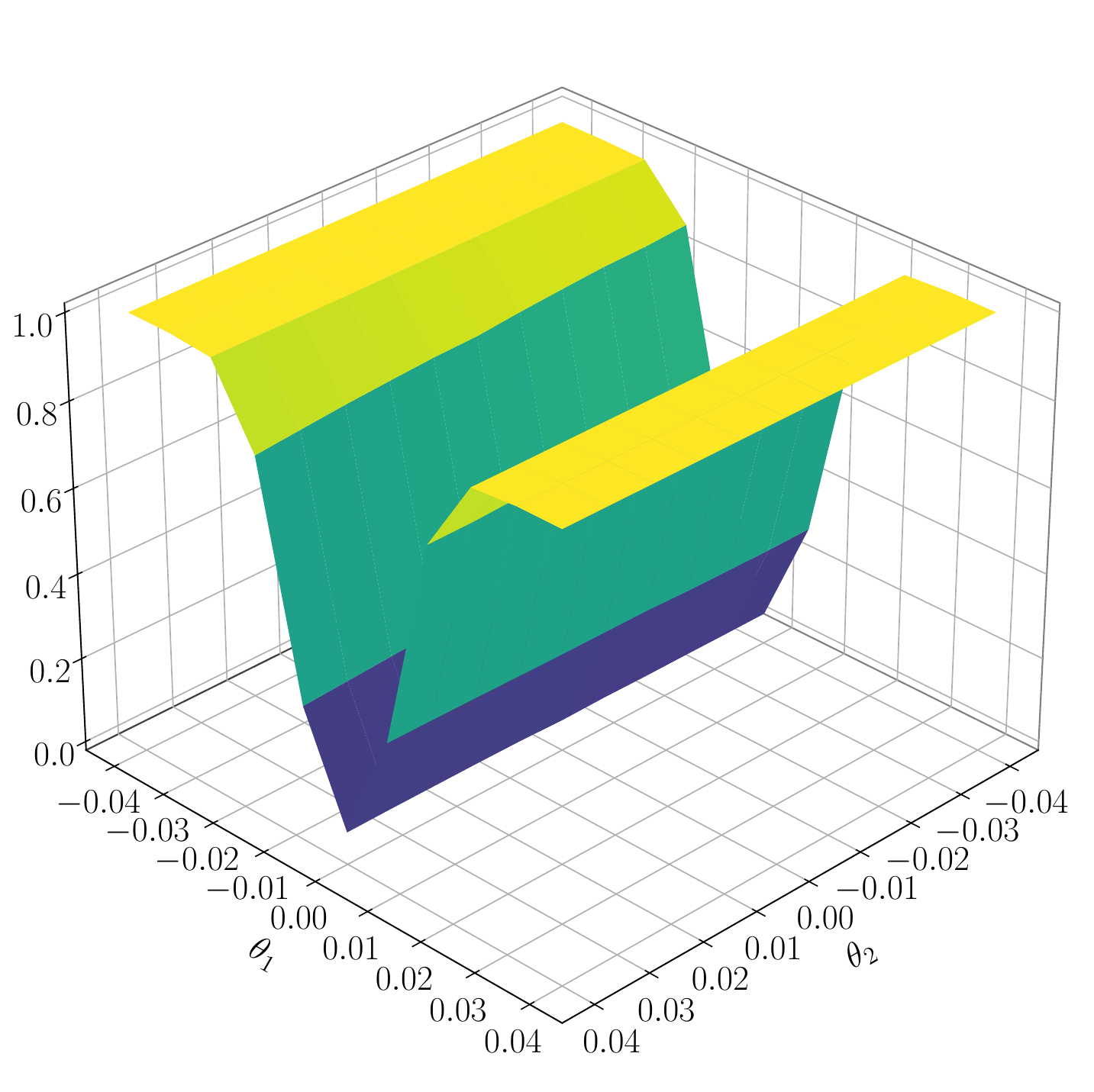}
			\caption[]%
			{{\small $j_1=1$, $j_2=3$}}
		\end{subfigure}
		\hfill
		\begin{subfigure}[b]{0.3\textwidth}
			\centering
			\includegraphics[width=0.99\textwidth]{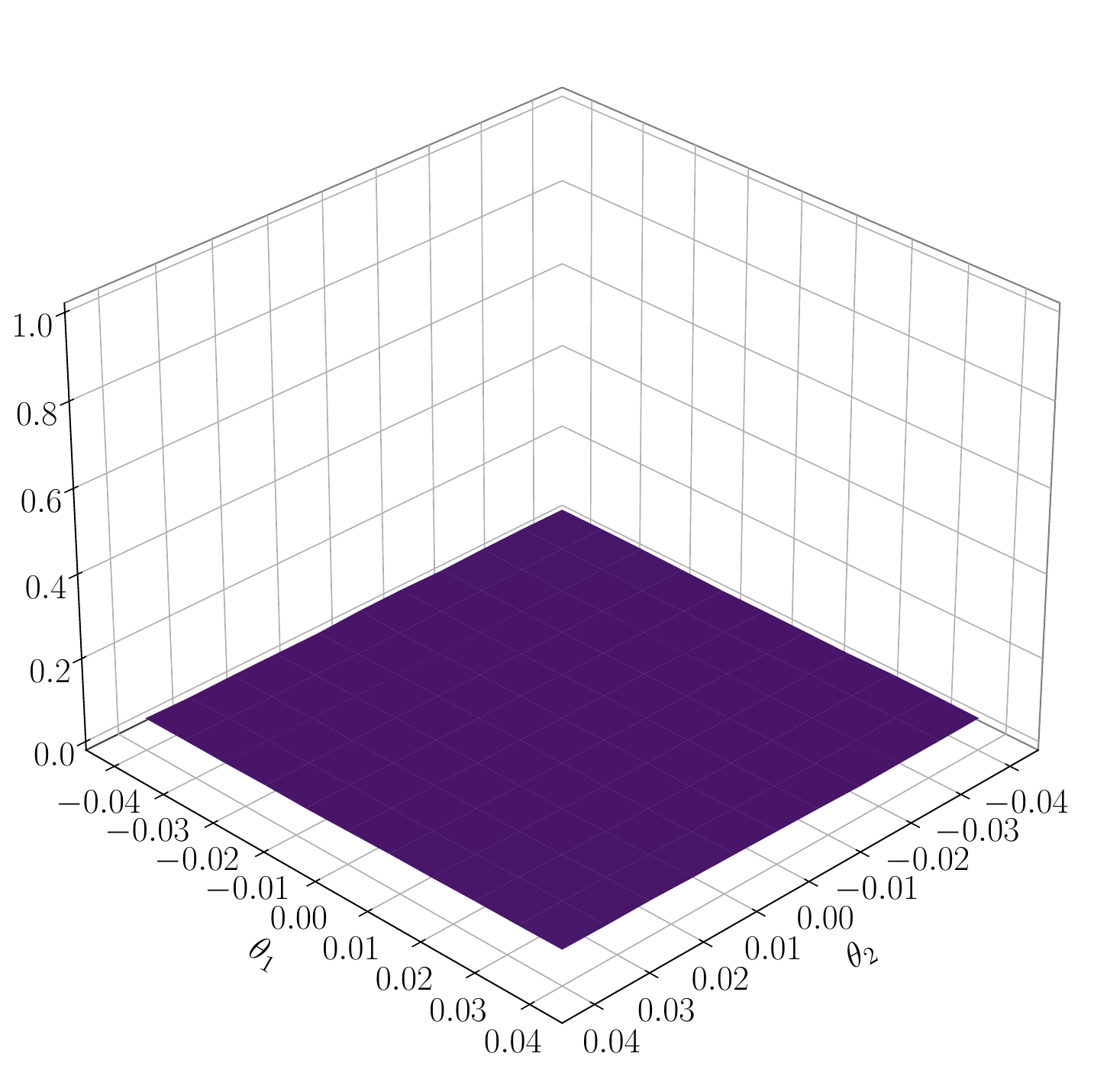}
			\caption[]%
			{{\small $j_1=3$, $j_2=3$}}
		\end{subfigure}
\end{figure}

\begin{figure}[!htbp]
		\centering
		\caption{\label{fig:IV-power-1-exp-log-AR} IV Design 1, AR power, $\pi_1$ exponential, $\pi_2$ logistic}
		\begin{subfigure}[b]{0.3\textwidth}
			\centering
			\includegraphics[width=0.99\textwidth]{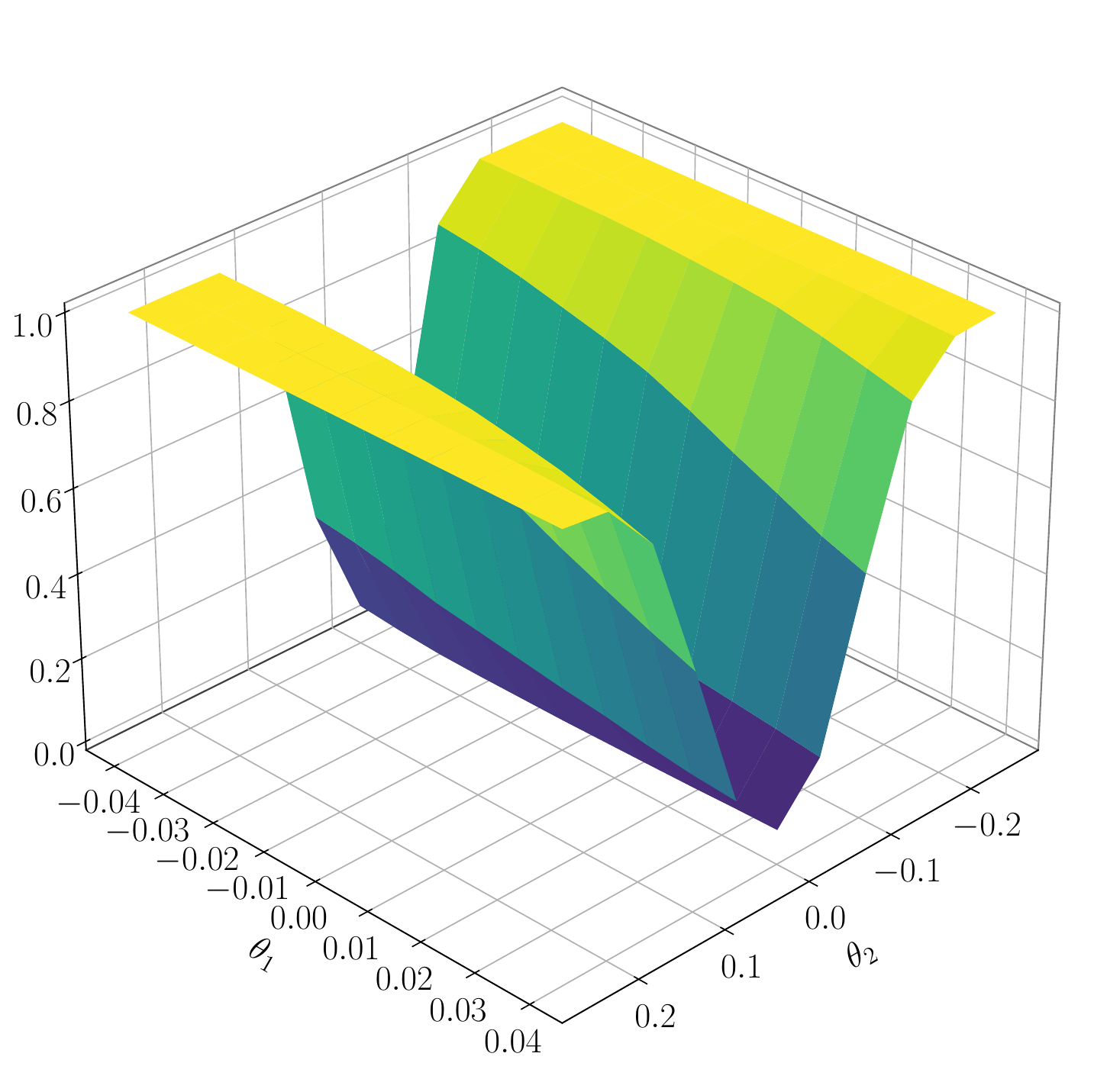}
			\caption[]%
			{{\small $j_1=1$, $j_2=1$}}
		\end{subfigure}
		\hfill
		\begin{subfigure}[b]{0.3\textwidth}
			\centering
			\includegraphics[width=0.99\textwidth]{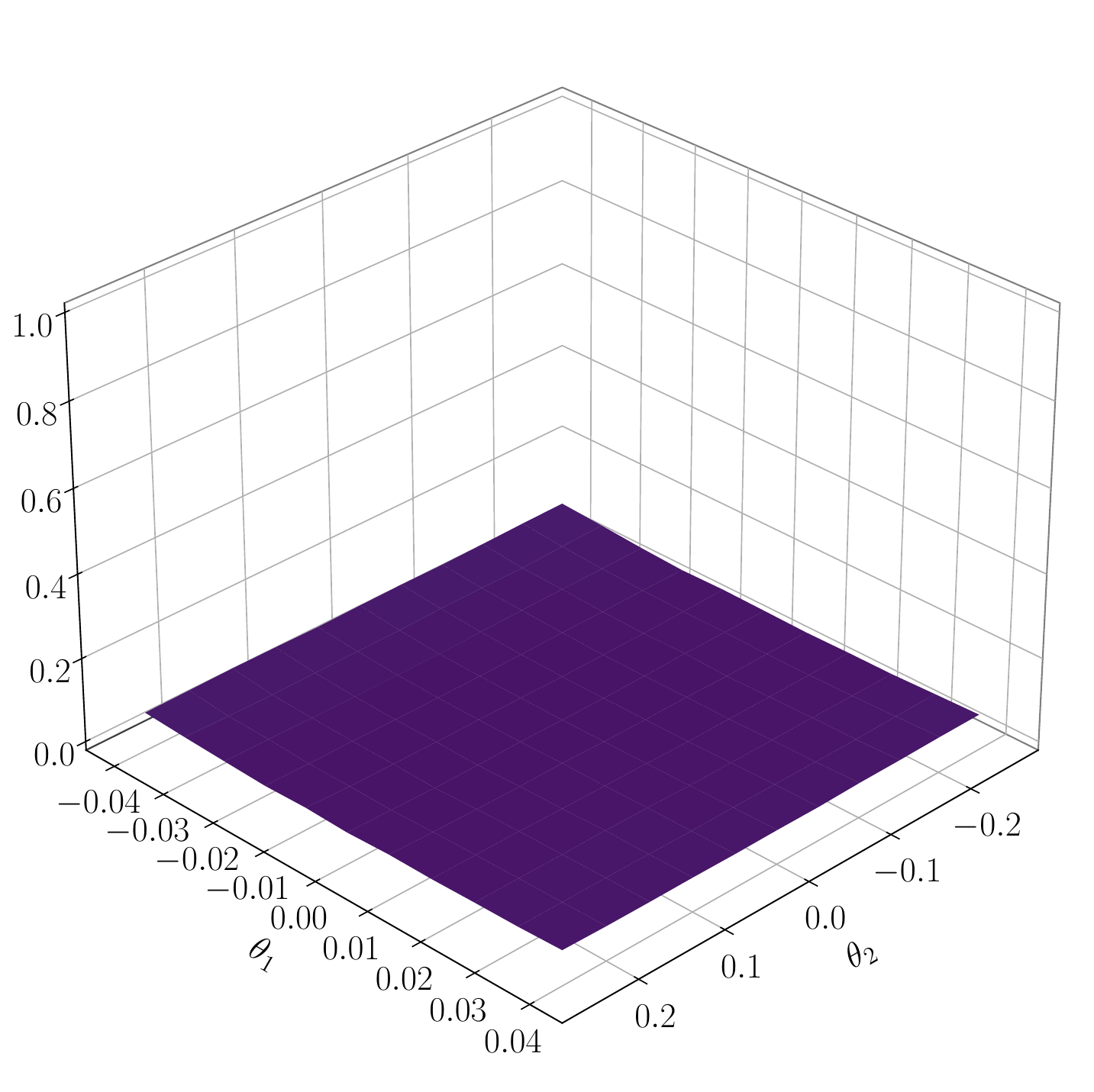}
			\caption[]%
			{{\small $j_1=1$, $j_2=3$}}
		\end{subfigure}
		\hfill
		\begin{subfigure}[b]{0.3\textwidth}
			\centering
			\includegraphics[width=0.99\textwidth]{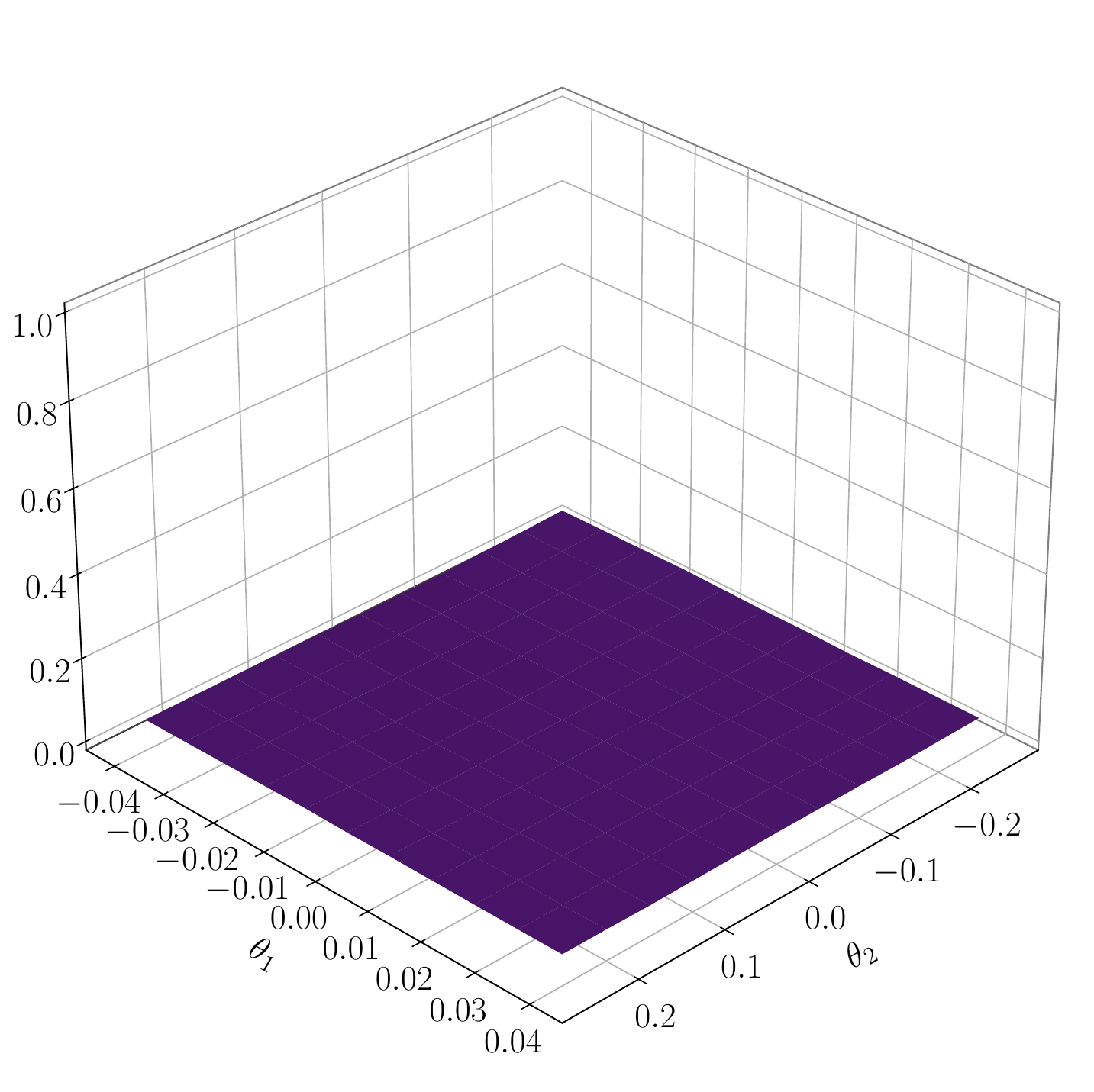}
			\caption[]%
			{{\small $j_1=3$, $j_2=3$}}
		\end{subfigure}
\end{figure}

\begin{figure}[!htbp]
		\centering
		\caption{\label{fig:IV-power-1-exp-log-psi} IV Design 1, $\psi$ $(k=3)$ power, $\pi_1$ exponential, $\pi_2$ logistic}
		\begin{subfigure}[b]{0.3\textwidth}
			\centering
			\includegraphics[width=0.99\textwidth]{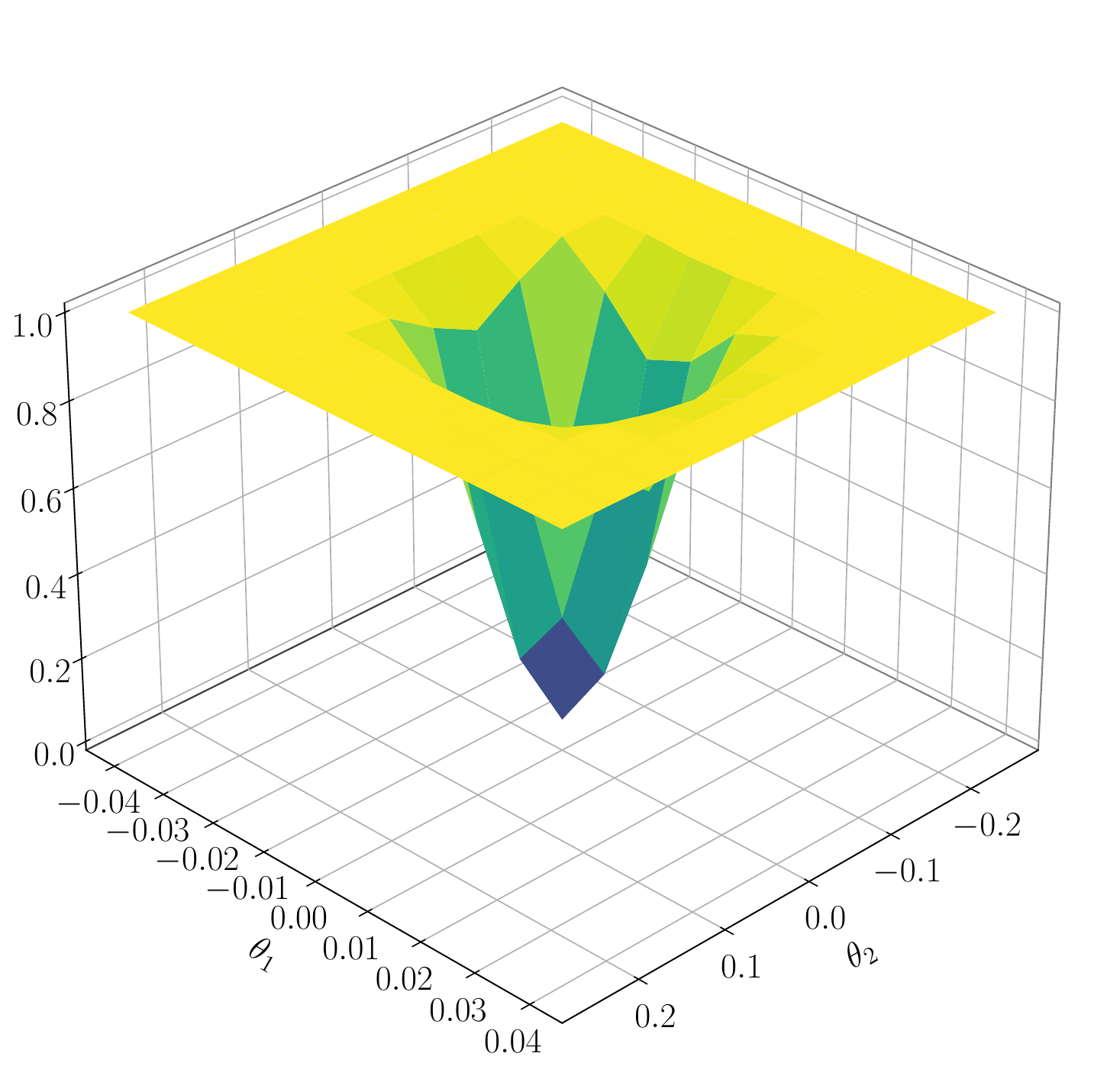}
			\caption[]%
			{{\small $j_1=1$, $j_2=1$}}
		\end{subfigure}
		\hfill
		\begin{subfigure}[b]{0.3\textwidth}
			\centering
			\includegraphics[width=0.99\textwidth]{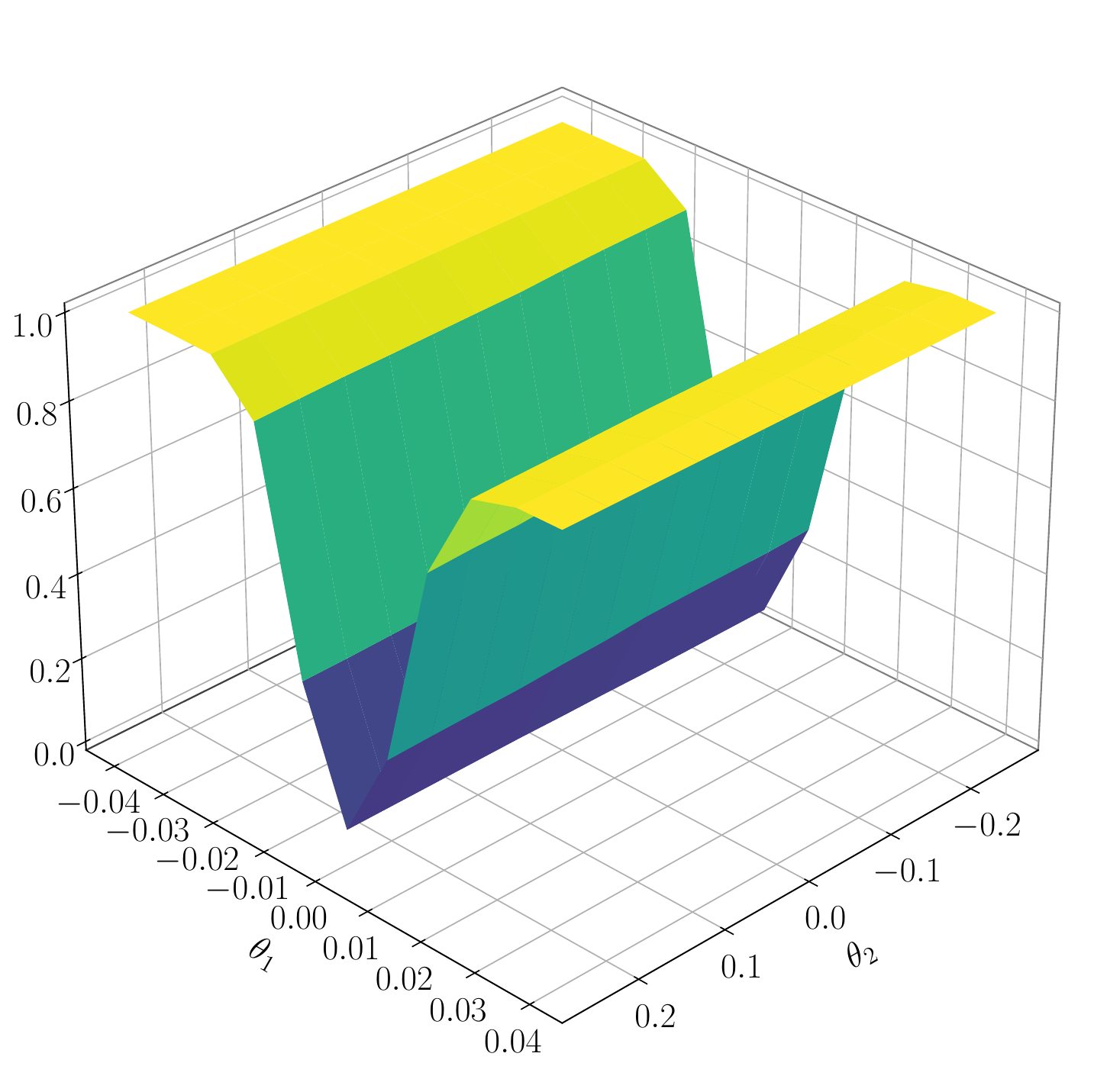}
			\caption[]%
			{{\small $j_1=1$, $j_2=3$}}
		\end{subfigure}
		\hfill
		\begin{subfigure}[b]{0.3\textwidth}
			\centering
			\includegraphics[width=0.99\textwidth]{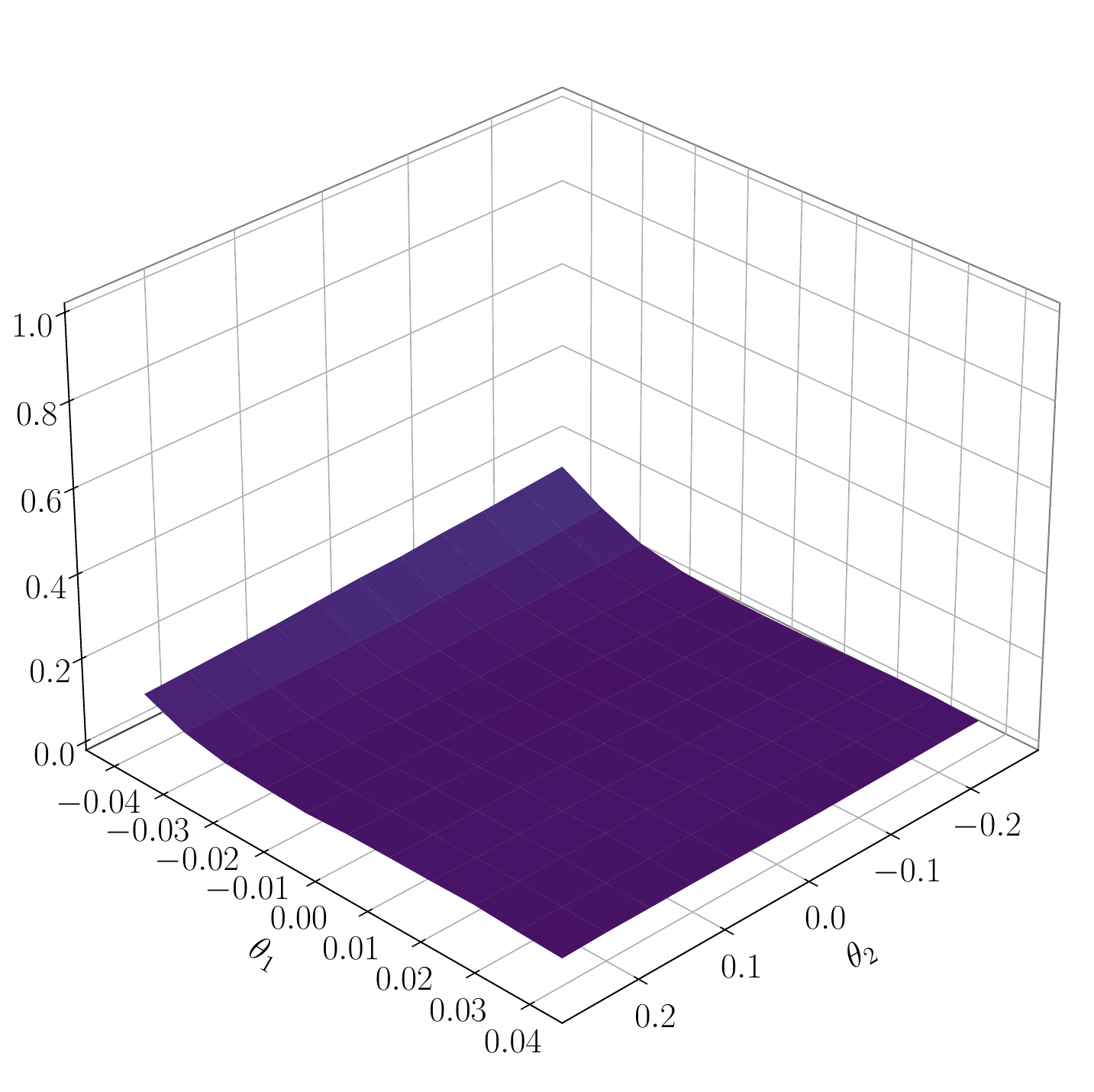}
			\caption[]%
			{{\small $j_1=3$, $j_2=3$}}
		\end{subfigure}
\end{figure}

Design 2 is a univariate, over identified model with heteroskedastic errors. $Z_1$, $\beta$ and $Z_2$ are as in Design 1, and $\pi(Z_2) = (\pi_1(Z_{2, 1}) + \pi_2(Z_{2, 2}))/2$ where the $\pi_i$ have one of the exponential or logistic forms of Table \ref{tbl:SIM-index-fcns}.\footnote{This functional form is treated as unknown and not imposed in the estimation of $\pi$.} I draw $(\tilde{\epsilon}, \tilde{\upsilon})$ from a zero-mean multivariate normal distribution with unit variances,  covariance $0.95$ and set $(\epsilon, \upsilon)\transp = \left[\begin{smallmatrix}
    \sqrt{1 + \sin(Z_{2, 1})^2} & 0\\
    0 & \sqrt{1 + \cos(Z_{2, 2})^2} 
\end{smallmatrix}\right](\tilde{\epsilon}, \tilde{\upsilon})\transp$.

The $\psi$ tests are computed in the same manner as in Design 1. I also compute the AR, LM and CLR tests based on $Z_2$ (with $Z_1$ partialled out) as well as the many weak instrument robust jackknife AR test of \cite{MS21}. MS$_1$ uses $Z_2$ as instruments; MS$_2$ uses the (tensor product) of Legendre polynomials used to estimate $\pi$ as instruments. 

The empirical rejection frequencies under the null are shown in Tables \ref{tbl:IV-size2i} -- \ref{tbl:IV-size2iii}. As in Design 1, the parameter $j$ controls the level of identification: the larger is $j$ the closer $\pi_j$ is to a constant function and hence $\theta$ unidentified. In each specification all the considered tests reject close to the nominal level; the MS$_2$ test is somewhat oversized for smaller $n$.
The power of these tests is plotted in Figures \ref{fig:IV-power-2-exp-exp} -- \ref{fig:IV-power-2-exp-log}; the $\psi_{n, \theta_0}$ tests are  denoted by $k=3$, AIC and BIC, corresponding to how $\pi$ is estimated. For the design with both $\pi_i$ exponential, the $\psi_{n, \theta_0}$ test clearly delivers the highest power whenever there is non-trivial power available; of the other tests, only MS$_2$ delivers non-trivial power in this specification. For the case with both $\pi_i$ logistic, all tests except MS$_2$ perform similarly, with MS$_2$ offering lower power. The same holds for the final specification, where $\pi_1$ is exponential and $\pi_2$ logistic.

\begin{table*}[!htbp]
	\renewcommand*{\arraystretch}{0.95}
	\footnotesize
	\begin{center}
		\caption{\label{tbl:IV-size2i} Empirical rejection frequencies, IV, Design 2, Exponential - Exponential}
		\begin{threeparttable}
			
\begin{tabular}{rrrrrrrrrr}
\toprule
\multicolumn{1}{c}{} & \multicolumn{1}{c}{} & \multicolumn{1}{c}{AR} & \multicolumn{1}{c}{LM} & \multicolumn{1}{c}{CLR} & \multicolumn{1}{c}{MS$_1$} & \multicolumn{1}{c}{MS$_2$} & \multicolumn{3}{c}{$\psi$} \\
\cmidrule(l{3pt}r{3pt}){3-3} \cmidrule(l{3pt}r{3pt}){4-4} \cmidrule(l{3pt}r{3pt}){5-5} \cmidrule(l{3pt}r{3pt}){6-6} \cmidrule(l{3pt}r{3pt}){7-7} \cmidrule(l{3pt}r{3pt}){8-10}
$n$ & $j$ &  &  &  &  &  & $k=3$ & AIC & BIC\\
\midrule
200 & 1 & 5.68 & 5.52 & 5.94 & 7.56 & 9.24 & 6.36 & 5.74 & 6.36\\
200 & 2 & 5.68 & 6.08 & 5.88 & 7.56 & 9.24 & 7.62 & 7.62 & 7.62\\
200 & 3 & 5.68 & 6.46 & 6.16 & 7.56 & 9.24 & 7.36 & 7.44 & 7.36\\
\addlinespace
400 & 1 & 5.28 & 5.48 & 5.26 & 7.18 & 8.78 & 6.32 & 5.38 & 6.36\\
400 & 2 & 5.28 & 5.72 & 5.48 & 7.18 & 8.78 & 7.98 & 7.98 & 7.98\\
400 & 3 & 5.28 & 6.06 & 5.84 & 7.18 & 8.78 & 3.42 & 3.68 & 3.42\\
\addlinespace
600 & 1 & 5.86 & 5.48 & 6.04 & 7.58 & 7.92 & 5.16 & 5.16 & 5.14\\
600 & 2 & 5.86 & 5.50 & 5.86 & 7.58 & 7.92 & 6.28 & 6.20 & 6.28\\
600 & 3 & 5.86 & 5.76 & 6.30 & 7.58 & 7.92 & 1.32 & 1.74 & 1.32\\
\bottomrule
\end{tabular}

			\begin{tablenotes}
			\footnotesize
			\item \notes{The functions $\pi_i$ have the exponential form in Table \ref{tbl:SIM-index-fcns} with $c_j$ corresponding to column $j$.
			}
			\end{tablenotes}	
		\end{threeparttable}
	\end{center}
\end{table*}

\begin{table*}[!htbp]
	\renewcommand*{\arraystretch}{0.95}
	\footnotesize
	\begin{center}
		\caption{\label{tbl:IV-size2ii} Empirical rejection frequencies, IV, Design 2, Logistic - Logistic}
		\begin{threeparttable}
			
\begin{tabular}{rrrrrrrrrr}
\toprule
\multicolumn{1}{c}{} & \multicolumn{1}{c}{} & \multicolumn{1}{c}{AR} & \multicolumn{1}{c}{LM} & \multicolumn{1}{c}{CLR} & \multicolumn{1}{c}{MS$_1$} & \multicolumn{1}{c}{MS$_2$} & \multicolumn{3}{c}{$\psi$} \\
\cmidrule(l{3pt}r{3pt}){3-3} \cmidrule(l{3pt}r{3pt}){4-4} \cmidrule(l{3pt}r{3pt}){5-5} \cmidrule(l{3pt}r{3pt}){6-6} \cmidrule(l{3pt}r{3pt}){7-7} \cmidrule(l{3pt}r{3pt}){8-10}
$n$ & $j$ &  &  &  &  &  & $k=3$ & AIC & BIC\\
\midrule
200 & 1 & 5.68 & 5.78 & 6.94 & 7.56 & 9.24 & 4.82 & 4.46 & 4.80\\
200 & 2 & 5.68 & 5.94 & 7.22 & 7.56 & 9.24 & 5.76 & 5.74 & 5.76\\
200 & 3 & 5.68 & 5.82 & 6.86 & 7.56 & 9.24 & 7.68 & 7.78 & 7.68\\
\addlinespace
400 & 1 & 5.28 & 4.76 & 5.98 & 7.18 & 8.78 & 4.32 & 4.04 & 4.18\\
400 & 2 & 5.28 & 4.80 & 6.26 & 7.18 & 8.78 & 4.90 & 4.98 & 4.90\\
400 & 3 & 5.28 & 4.94 & 5.80 & 7.18 & 8.78 & 3.18 & 3.52 & 3.18\\
\addlinespace
600 & 1 & 5.86 & 5.12 & 6.42 & 7.58 & 7.92 & 4.84 & 4.50 & 4.66\\
600 & 2 & 5.86 & 5.20 & 6.32 & 7.58 & 7.92 & 5.00 & 4.90 & 5.00\\
600 & 3 & 5.86 & 5.14 & 6.28 & 7.58 & 7.92 & 1.56 & 1.84 & 1.56\\
\bottomrule
\end{tabular}

			\begin{tablenotes}
			\footnotesize
			\item \notes{The functions $\pi_i$ have the logistic form in Table \ref{tbl:SIM-index-fcns} with $c_j$ corresponding to column $j$.
			}
			\end{tablenotes}	
		\end{threeparttable}
	\end{center}
\end{table*}

\begin{table*}[!htbp]
	\renewcommand*{\arraystretch}{0.95}
	\footnotesize
	\begin{center}
		\caption{\label{tbl:IV-size2iii} Empirical rejection frequencies, IV, Design 2, Exponential - Logistic}
		\begin{threeparttable}
			
\begin{tabular}{rrrrrrrrrr}
\toprule
\multicolumn{1}{c}{} & \multicolumn{1}{c}{} & \multicolumn{1}{c}{AR} & \multicolumn{1}{c}{LM} & \multicolumn{1}{c}{CLR} & \multicolumn{1}{c}{MS$_1$} & \multicolumn{1}{c}{MS$_2$} & \multicolumn{3}{c}{$\psi$} \\
\cmidrule(l{3pt}r{3pt}){3-3} \cmidrule(l{3pt}r{3pt}){4-4} \cmidrule(l{3pt}r{3pt}){5-5} \cmidrule(l{3pt}r{3pt}){6-6} \cmidrule(l{3pt}r{3pt}){7-7} \cmidrule(l{3pt}r{3pt}){8-10}
$n$ & $j$ &  &  &  &  &  & $k=3$ & AIC & BIC\\
\midrule
200 & 1 & 5.68 & 6.14 & 7.12 & 7.56 & 9.24 & 5.28 & 4.62 & 5.30\\
200 & 2 & 5.68 & 5.92 & 7.18 & 7.56 & 9.24 & 6.90 & 6.90 & 6.90\\
200 & 3 & 5.68 & 5.90 & 6.60 & 7.56 & 9.24 & 7.16 & 7.26 & 7.16\\
\addlinespace
400 & 1 & 5.28 & 4.94 & 6.44 & 7.18 & 8.78 & 4.74 & 4.28 & 4.70\\
400 & 2 & 5.28 & 5.24 & 6.54 & 7.18 & 8.78 & 5.52 & 5.64 & 5.52\\
400 & 3 & 5.28 & 5.00 & 5.66 & 7.18 & 8.78 & 3.10 & 3.42 & 3.10\\
\addlinespace
600 & 1 & 5.86 & 5.20 & 6.30 & 7.58 & 7.92 & 4.86 & 4.56 & 5.00\\
600 & 2 & 5.86 & 5.14 & 6.48 & 7.58 & 7.92 & 5.62 & 5.50 & 5.62\\
600 & 3 & 5.86 & 5.22 & 6.24 & 7.58 & 7.92 & 1.20 & 1.58 & 1.20\\
\bottomrule
\end{tabular}

			\begin{tablenotes}
			\footnotesize
			\item \notes{$\pi_1$, $\pi_2$ have the exponential and logistic form in Table \ref{tbl:SIM-index-fcns} respectively with $c_{j_i}$ corresponding to column $j$.
			}
			\end{tablenotes}	
		\end{threeparttable}
	\end{center}
\end{table*}

\begin{figure}[htbp]
	\begin{minipage}{.99\textwidth}
		\centering
		\caption{\label{fig:IV-power-2-exp-exp} \small IV design 2, Power curves, exponential $\pi_i$ }
		\includegraphics[scale = 0.45]{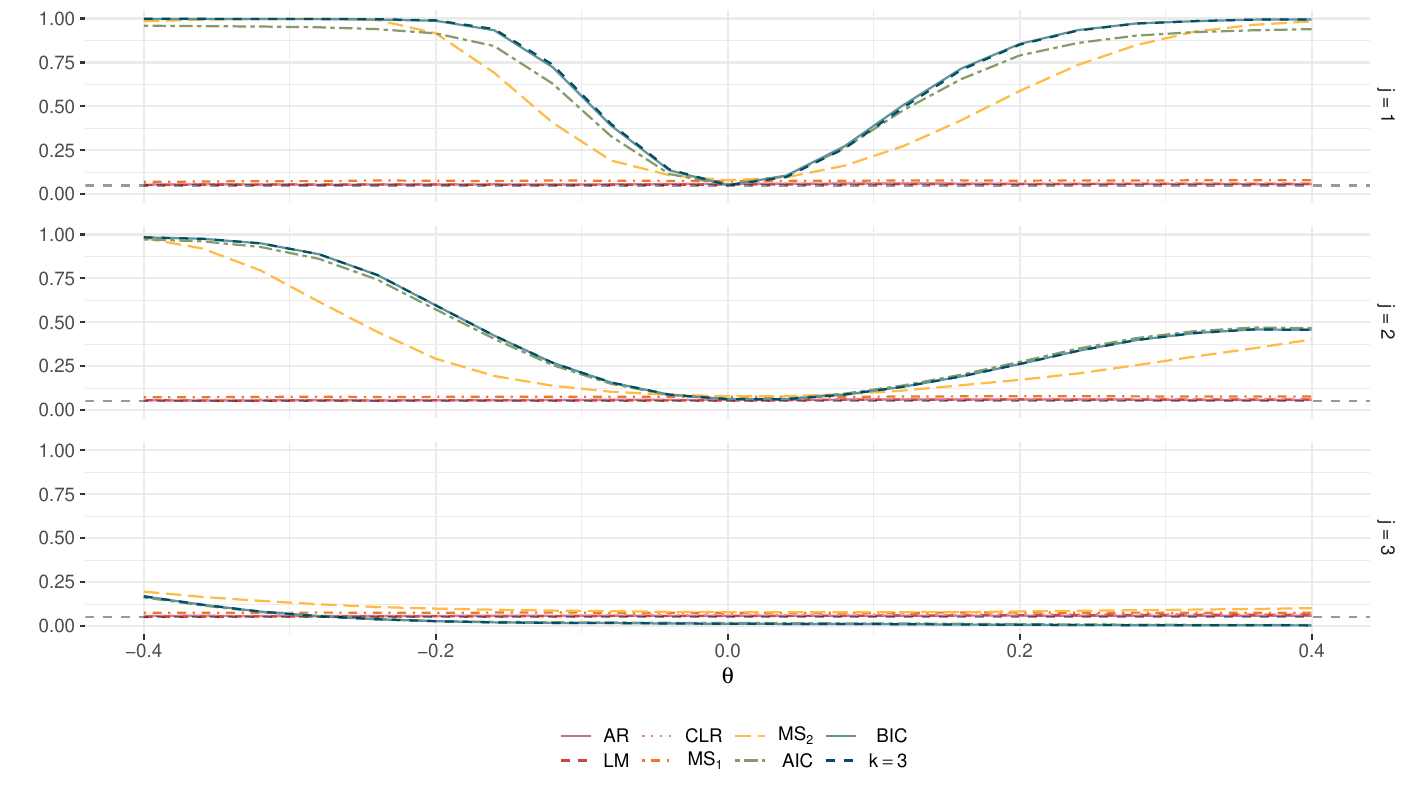}
		\end{minipage}
\end{figure}
\begin{figure}[htbp]
	\begin{minipage}{.99\textwidth}
		\centering
		\caption{\label{fig:IV-power-2-log-log} \small IV Design 2, Power curves, logistic $\pi_i$ }
		\includegraphics[scale = 0.45]{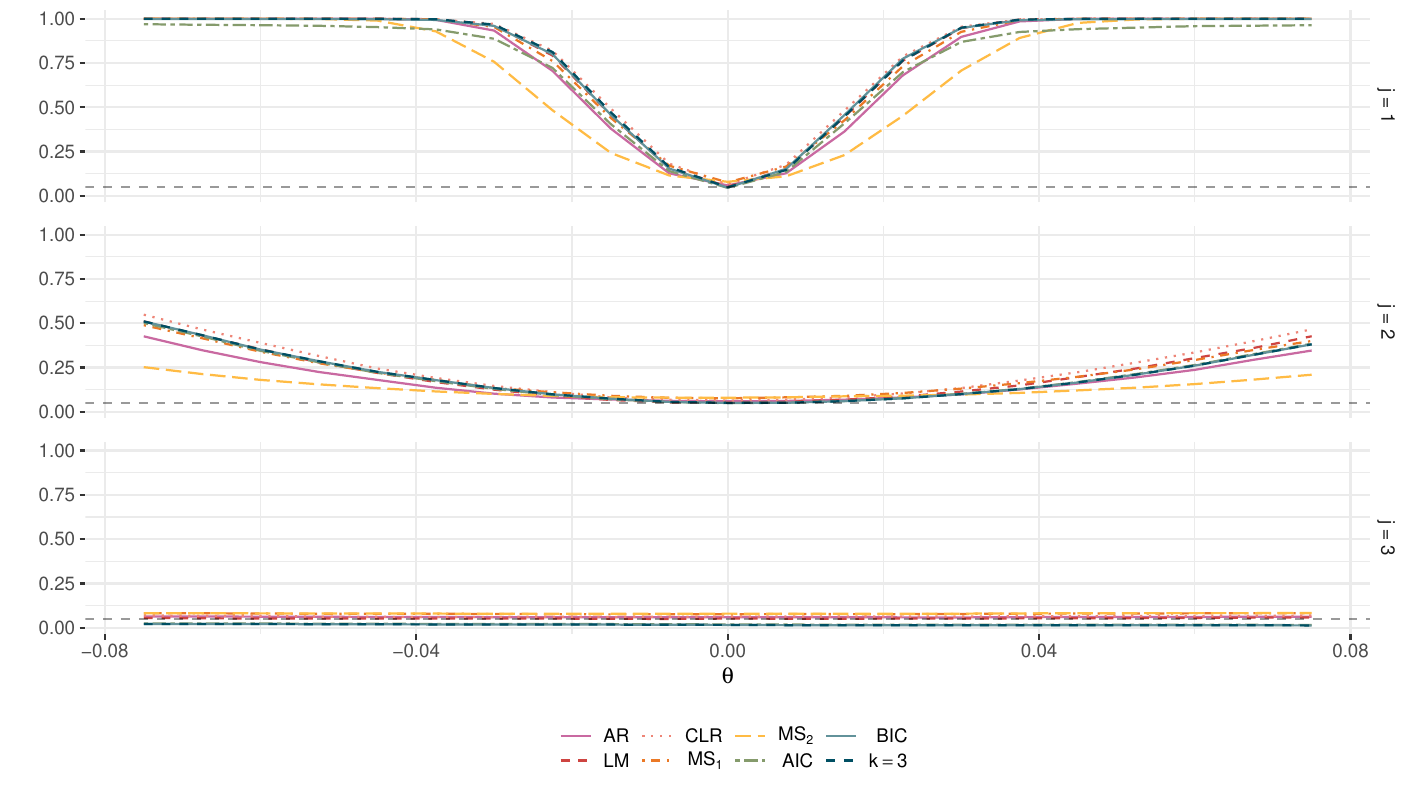}
		\end{minipage}
\end{figure}

\begin{figure}[htbp]
	\begin{minipage}{.99\textwidth}
		\centering
		\caption{\label{fig:IV-power-2-exp-log} \small IV Design 2, Power curves, exponential $\pi_1$, logistic $\pi_2$ }
		\includegraphics[scale = 0.45]{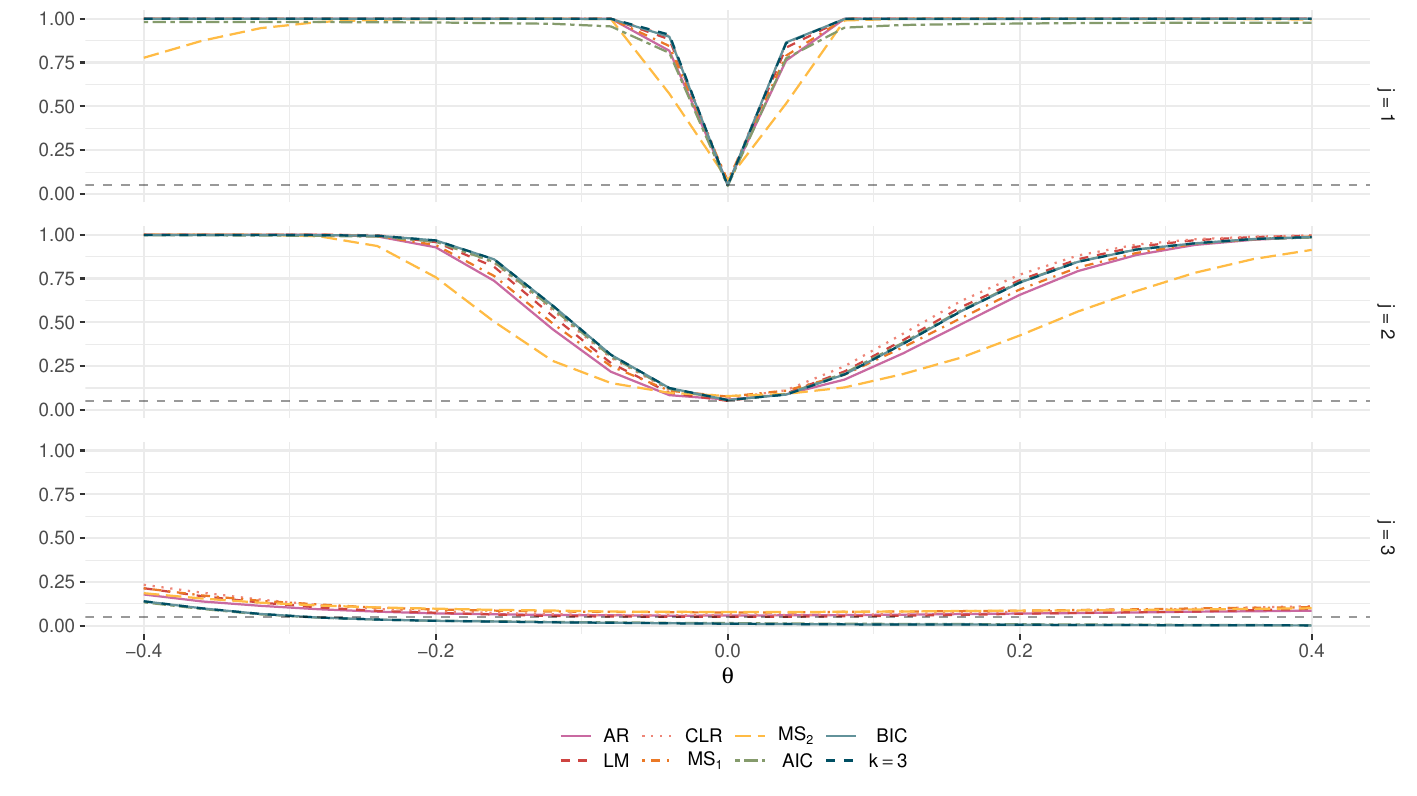}
		\end{minipage}
\end{figure}

\section{Empirical applications}\label{sec:ea}

In this section I re-analyse two IV studies with potentially weak instruments by inverting the $\psi_{n, \theta_0}$ test developed in section \ref{ssec:iv} to construct weak-instrument robust confidence intervals (CIs). In each case the $\psi_{n, \theta_0}$ test is able to exploit non-linearities to yield substantial reductions in CI length relative to AR CIs.

\subsection{The effect of skilled immigration on productivity}
\cite{H14} studies the long term effect of skilled immigration on productivity using a natural experiment in which the skilled but religiously persecuted French protestants (Hugenots) fled and settled in Prussia.\footnote{That the Hugenots were more skilled (on average) than the Prussian population seems to be broadly accepted cf. pp. 85-86, 93-95 in \cite{H14}.} 
In the notation of Example \ref{ex:IV}, $Y$ is log output in textile manufacturing, $X$ is the proportion of Hugenots in each town and $Z_1$ contains various control variables, see \cite{H14} for details. \cite{H14} argues that ``by the order of centralized ruling by the king and his agents Huguenots were channeled into Prussian towns in order to compensate for severe population losses during the Thirty Years' War'', motivating the instrument $Z_2$: the percentage population losses during the war. In particular, three different measurements of this population loss are used and refered to as specifications (1), (2) and (3) hereafter.\footnote{Specifications (1) \& (2) are those considered in the left and and right hand parts of Table 4 of \cite{H14}; Specification (3) is that considered in the left hand part of Table 5.}

As noted in \cite{H14}, this instrument may be weak: in each case the first stage F statistic is ``small''.\footnote{It is less than the cutoff of 10 suggested by \cite{SS97} for homoskedastic IV.}
I implement the $\psi_{n, \theta_0}$ test using a leave-one-out series estimator of $\pi$ of the form 
\begin{equation}\label{eq:pi-hat-ea}
    \hat{\pi}_{i}(Z_i) = \hat{\pi}_{i}\transp p_{K}(Z_{2,i}) + \hat{\beta}_i Z_{1, i},
\end{equation}
where the $i$ subscript on the estimated coefficients indicates they have been estimated on all observations except for the $i$-th. $p_{K}$ is a vector of a constant and the first $K$ Legendre polynomials. I choose $K\in \{1, 2, 3, 4\}$ and whether to include $Z_1$ in the model for $\pi$ by using BIC: all specifications include $Z_1$ and $K=4$. 

Table \ref{tbl:H14-1} reports 2SLS estimates of $\theta$ along with 2SLS (Wald) CIs, AR CIs and CIs found by inverting the $\psi_{n, \theta_0}$ test.\footnote{The inversion is performed over a grid of 5000 equally spaced points from -1 to 7.}  The resulting CIs provide a similar interpretation as that based on the AR CIs: the effect of (skilled) Hugenot immigration was positive on textile output. However, the $\psi_{n, \theta_0}$ based CIs are smaller than the AR CIs, achieving approximately a 40\% - 50\% reduction in length.

\begin{table*}[!htbp]
    \small
    \renewcommand{\arraystretch}{1.05}
	\begin{center}
		\caption{\label{tbl:H14-1} Point estimates and confidence intervals}
		\begin{threeparttable}
			
\begin{tabular}{lccc}
\toprule
 & (1) & (2) & (3)\\
\midrule
n & 150 & 150 & 186\\
F & 3.668 & 4.791 & 5.736\\
\midrule\multicolumn{4}{l}{\textit{Point estimate}}\\
2SLS & 3.475 & 3.38 & 1.671\\
\midrule\multicolumn{4}{l}{\textit{Confidence intervals}}\\
2SLS & {}[1.27, 5.68] & {}[1.294, 5.467] & {}[0.032, 3.31]\\
AR & {}[1.427, 6.303] & {}[1.43, 5.985] & {}[-0.022, 3.379]\\
$\psi$ & {}[1.626, 4.099] & {}[1.637, 4.073] & {}[1.136, 3.228]\\
\midrule\multicolumn{4}{l}{\textit{Relative length of confidence intervals to AR}}\\
2SLS & 0.904 & 0.916 & 0.964\\
$\psi$ & 0.507 & 0.535 & 0.615\\
\bottomrule
\end{tabular}

            \begin{tablenotes}
                \scriptsize
                \item \notes{F is the first stage $F$ statistic. All confidence intervals have nominal coverage of 95\%.
                }
                \end{tablenotes}	
		\end{threeparttable}
	\end{center}
\end{table*}

\subsection{The effect of racial segregation on inequality}
\cite{A11} estimates the effect of racial segregation ($X$) on poverty and inequality ($Y$, measured respectively by the poverty rate and log gini coefficient for black / white city residents), instrumenting segregation by a ``railroad division index'' (RDI, $Z_2$), ``a variation on a Herfindahl index that measures the dispersion of a city's land into subunits'' via the layout of railroad tracks.\footnote{Section 3 and Appendix A of \cite{A11} provides evidence that the choice of railroad placement was not related to local social or economic concerns.} The first and second stages also include an intercept and control for railroad track length ($Z_1$).  

The instrument may be weak: I calculate the first stage F statistic to be 2.307.\footnote{\cite{A11} refers to Column 1 of Table 1 when discussing the first stage F statistic. The values in this table imply a first stage F of $(0.357 / 0.088)^2 \approx 16.458$. This ``discrepancy'' arises from different default choices of robust covariance estimate in \textsf{R}'s \texttt{sandwich} package (HC3; my calculation) and STATA's \texttt{robust} command (HC1; \citealp{A11}).} I implement the $\psi_{n, \theta_0}$ test using a leave-one-out series estimator of $\pi$ of the form \eqref{eq:pi-hat-ea}. I choose $K\in \{1, 2, 3, 4\}$ and whether to include $Z_1$ in the model for $\pi$ by using BIC: this excludes $Z_1$ and chooses $K=2$.  

Table \ref{tbl:A11-1} reports 2SLS estimates of $\theta$ along with 2SLS (Wald) CIs, AR CIs and CIs found by inverting the $\psi_{n, \theta_0}$ test.\footnote{The inversion is performed over a grid of 5000 equally spaced points from -1 to 2.}  The resulting CIs provide a similar interpretation as that based on the AR CIs: racial segregation increases poverty and inequality within the Black community and decreases poverty and inequality within the White community. The $\psi_{n, \theta_0}$ based CIs are shorter than the AR CIs, achieving a reduction in length varying from 5\% to around 38\%.

\begin{table*}[!htbp]
    \small
    \renewcommand{\arraystretch}{1.05}
	\begin{center}
		\caption{\label{tbl:A11-1} Point estimates and confidence intervals}
		\begin{threeparttable}
			
\begin{tabular}{lcccc}
\toprule
\multicolumn{1}{c}{ } & \multicolumn{2}{c}{Poverty rate} & \multicolumn{2}{c}{Gini coefficient} \\
\cmidrule(l{3pt}r{3pt}){2-3} \cmidrule(l{3pt}r{3pt}){4-5}
  & White & Black & White & Black\\
\midrule\multicolumn{5}{l}{\textit{Point estimate}}\\
\midrule
2SLS & 0.258 & -0.196 & 0.875 & -0.334\\
\midrule\multicolumn{5}{l}{\textit{Confidence intervals}}\\
2SLS & {}[-0.026, 0.543] & {}[-0.334, -0.058] & {}[0.277, 1.474] & {}[-0.558, -0.111]\\
AR & {}[-0.04, 0.598] & {}[-0.394, -0.075] & {}[0.319, 1.684] & {}[-0.674, -0.149]\\
$\psi$ & {}[0.072, 0.568] & {}[-0.376, -0.097] & {}[0.29, 1.138] & {}[-0.608, -0.106]\\
\midrule\multicolumn{5}{l}{\textit{Relative length of confidence intervals to AR}}\\
2SLS & 0.891 & 0.866 & 0.877 & 0.853\\
$\psi$ & 0.775 & 0.877 & 0.622 & 0.955\\
\bottomrule
\end{tabular}

            \begin{tablenotes}
                \scriptsize
                \item \notes{All confidence intervals have nominal coverage of 95\%.
                }
                \end{tablenotes}	
		\end{threeparttable}
	\end{center}
\end{table*}

\section{Conclusion}\label{sec:concl}

In this paper I establish that C($\alpha$)-style tests are locally regular under mild conditions, including in non-regular cases where locally regular estimators do not exist. As a consequence, these tests do not overreject under semiparametric weak identification asymptotics. Additionally I generalise the classical local asymptotic power bounds for LAN models to the case where the efficient information matrix has positive, but potentially deficient, rank, such that these results also apply in cases of underidentification (or weak underidentification). I show that, if the C($\alpha$) test is based on the efficient score function, it attains these power bounds. This (attainment) result improves on results known in the literature in two ways: (i) it applies also to non-regular models and (ii) it does not require the data to be i.i.d. nor the information operator to be boundedly invertible. A simulation study based on two examples shows that the asymptotic theory provides an accurate approximation to the finite sample performance of the proposed tests.

\bibliographystyle{asa}
\bibliography{Bib}

\begin{thebibliography}{72}
\newcommand{\enquote}[1]{``#1''}
\expandafter\ifx\csname natexlab\endcsname\relax\def\natexlab#1{#1}\fi

\bibitem[{Amengual et~al.(2023)Amengual, Bei, and Sentana}]{ABS23}
Amengual, D., Bei, X., and Sentana, E. (2023), \enquote{Hypothesis Tests with a Repeatedly Singular Information Matrix,} Working paper.

\bibitem[{Ananat(2011)}]{A11}
Ananat, E.~O. (2011), \enquote{The Wrong Side(s) of the Tracks: The Causal Effects of Racial Segregation on Urban Poverty and Inequality,} \textit{American Economic Journal: Applied Economics}, 3, 34--66.

\bibitem[{Anderson and Rubin(1949)}]{AR49}
Anderson, T.~W. and Rubin, H. (1949), \enquote{{Estimation of the Parameters of a Single Equation in a Complete System of Stochastic Equations},} \textit{The Annals of Mathematical Statistics}, 20, 46 -- 63.

\bibitem[{Andrews(1987)}]{A87}
Andrews, D. W.~K. (1987), \enquote{Asymptotic Results for Generalized Wald Tests,} \textit{Econometric Theory}, 3, 348--358.

\bibitem[{Andrews(1999)}]{A99}
--- (1999), \enquote{Estimation When a Parameter is on a Boundary,} \textit{Econometrica}, 67, 1341--1383.

\bibitem[{Andrews(2001)}]{A01}
--- (2001), \enquote{{Testing When a Parameter is on the Boundary of the Maintained Hypothesis},} \textit{Econometrica}, 69, 683--734.

\bibitem[{Andrews and Cheng(2012)}]{AC12}
Andrews, D. W.~K. and Cheng, X. (2012), \enquote{Estimation and Inference With Weak, Semi-Strong, and Strong Identification,} \textit{Econometrica}, 80, 2153--2211.

\bibitem[{Andrews and Guggenberger(2009)}]{AG09}
Andrews, D. W.~K. and Guggenberger, P. (2009), \enquote{Hybrid and Size-Corrected Subsampling Methods,} \textit{Econometrica}, 77, 721--762.

\bibitem[{Andrews and Guggenberger(2010)}]{AG10b}
--- (2010), \enquote{Asymptotic Size and a Problem with Subsampling and with the $m$ out of $n$ Bootstrap,} \textit{Econometric Theory}, 26, 426--468.

\bibitem[{Andrews and Guggenberger(2019)}]{AG19}
--- (2019), \enquote{Identification- and singularity-robust inference for moment condition models,} \textit{Quantitative Economics}, 10, 1703--1746.

\bibitem[{Andrews and Mikusheva(2015)}]{AM15}
Andrews, I. and Mikusheva, A. (2015), \enquote{Maximum likelihood inference in weakly identified dynamic stochastic general equilibrium models,} \textit{Quantitative Economics}, 6, 123--152.

\bibitem[{Andrews and Mikusheva(2016)}]{AM16b}
--- (2016), \enquote{Conditional Inference With a Functional Nuisance Parameter,} \textit{Econometrica}, 84, 1571--1612.

\bibitem[{Andrews and Mikusheva(2022)}]{AM22}
--- (2022), \enquote{Optimal Decision Rules for Weak GMM,} \textit{Econometrica}, 90, 715--748.

\bibitem[{Andrews and Mikusheva(2023)}]{AM23}
--- (2023), \enquote{GMM is Inadmissible Under Weak Identification,} ArXiv:econ.EM/2204.12462.

\bibitem[{Belloni et~al.(2015)Belloni, Chernozhukov, Chetverikov, and Kato}]{BCCK15}
Belloni, A., Chernozhukov, V., Chetverikov, D., and Kato, K. (2015), \enquote{Some new asymptotic theory for least squares series: Pointwise and uniform results,} \textit{Journal of Econometrics}, 186, 345--366.

\bibitem[{Bernstein(2009)}]{B09}
Bernstein, D.~S. (2009), \textit{Matrix mathematics: theory, facts, and formulas}, Princeton, NJ, USA: Princeton University Press, 2nd ed.

\bibitem[{Bickel et~al.(1998)Bickel, Klaassen, Ritov, and Wellner}]{BKRW98}
Bickel, P.~J., Klaassen, C. A.~J., Ritov, Y., and Wellner, J.~A. (1998), \textit{Efficient and Adaptive Estimation for Semiparametric Models}, New York, NY, USA: Springer.

\bibitem[{Bickel et~al.(2006)Bickel, Ritov, and Stoker}]{BRS06}
Bickel, P.~J., Ritov, Y., and Stoker, T.~M. (2006), \enquote{{Tailor-made tests for goodness of fit to semiparametric hypotheses},} \textit{The Annals of Statistics}, 34, 721 -- 741.

\bibitem[{Billingsley(1999)}]{B99}
Billingsley, P. (1999), \textit{Convergence of Probability Measures}, Wiley.

\bibitem[{Bravo et~al.(2020)Bravo, Escanciano, and Van~Keilegom}]{BEvK20}
Bravo, F., Escanciano, J.~C., and Van~Keilegom, I. (2020), \enquote{Two-step semiparametric empirical likelihood inference,} \textit{Ann. Statist.}, 48, 1--26.

\bibitem[{Chamberlain(1986)}]{C86}
Chamberlain, G. (1986), \enquote{Asymptotic efficiency in semi-parametric models with censoring,} \textit{Journal of Econometrics}, 32, 189--218.

\bibitem[{Chamberlain(1992)}]{C92}
--- (1992), \enquote{Efficiency Bounds for Semiparametric Regression,} \textit{Econometrica}, 60, 567--596.

\bibitem[{Chernozhukov et~al.(2022)Chernozhukov, Escanciano, Ichimura, Newey, and Robins}]{CEINR22}
Chernozhukov, V., Escanciano, J.~C., Ichimura, H., Newey, W.~K., and Robins, J.~M. (2022), \enquote{Locally Robust Semiparametric Estimation,} \textit{Econometrica}, 90, 1501--1535.

\bibitem[{Chernozhukov et~al.(2009)Chernozhukov, Hansen, and Jansson}]{CHJ09}
Chernozhukov, V., Hansen, C., and Jansson, M. (2009), \enquote{Admissible Invariant Similar Tests for Instrumental Variables Regression,} \textit{Econometric Theory}, 25, 806--–818.

\bibitem[{Chernozhukov et~al.(2015)Chernozhukov, Hansen, and Spindler}]{CHS15}
Chernozhukov, V., Hansen, C., and Spindler, M. (2015), \enquote{Valid Post-Selection and Post-Regularization Inference: An Elementary, General Approach,} \textit{Annual Review of Economics}, 7, 649--688.

\bibitem[{Choi et~al.(1996)Choi, Hall, and Schick}]{CHS96}
Choi, S., Hall, W.~J., and Schick, A. (1996), \enquote{Asymptotically uniformly most powerful tests in parametric and semiparametric models,} \textit{Ann. Statist.}, 24, 841--861.

\bibitem[{Davidson(2021)}]{D21}
Davidson, J. (2021), \textit{Stochastic limit theory}, Oxford University Press, 2nd ed.

\bibitem[{Dufour(1997)}]{D97}
Dufour, J.-M. (1997), \enquote{Some Impossibility Theorems in Econometrics With Applications to Structural and Dynamic Models,} \textit{Econometrica}, 65, 1365--1387.

\bibitem[{Dufour and Val\'{e}ry(2016)}]{DV15}
Dufour, J.-M. and Val\'{e}ry, P. (2016), \enquote{Rank-robust Regularized Wald-type tests,} Working paper.

\bibitem[{Elliott et~al.(2015)Elliott, M\"{u}ller, and Watson}]{EMW15}
Elliott, G., M\"{u}ller, U.~K., and Watson, M.~W. (2015), \enquote{Nearly Optimal Tests When a Nuisance Parameter Is Present Under the Null Hypothesis,} \textit{Econometrica}, 83, 771--811.

\bibitem[{Escanciano(2022)}]{E22}
Escanciano, J.~C. (2022), \enquote{Semiparametric Identification and Fisher Information,} \textit{Econometric Theory}, 38, 301–338.

\bibitem[{Geyer(1994)}]{G94}
Geyer, C.~J. (1994), \enquote{{On the asymptotics of constrained M-estimation},} \textit{Annals of Statistics}, 22, 1993--2010.

\bibitem[{Han and McCloskey(2019)}]{HM19}
Han, S. and McCloskey, A. (2019), \enquote{Estimation and inference with a (nearly) singular Jacobian,} \textit{Quantitative Economics}, 10, 1019--1068.

\bibitem[{Hoesch et~al.(2024{\natexlab{a}})Hoesch, Lee, and Mesters}]{HLM22}
Hoesch, L., Lee, A., and Mesters, G. (2024{\natexlab{a}}), \enquote{Locally Robust Inference for Non-Gaussian SVAR Models,} \textit{Quantitative Economics}, 15, 523--570.

\bibitem[{Hoesch et~al.(2024{\natexlab{b}})Hoesch, Lee, and Mesters}]{HLM22-S}
--- (2024{\natexlab{b}}), \enquote{Supplement to ``Locally Robust Inference for Non-Gaussian SVAR Models'',} \textit{Quantitative Economics Supplementary Material}, 15.

\bibitem[{Hornung(2014)}]{H14}
Hornung, E. (2014), \enquote{Immigration and the Diffusion of Technology: The Huguenot Diaspora in Prussia,} \textit{The American Economic Review}, 104, 84--122.

\bibitem[{Ichimura(1993)}]{I93}
Ichimura, H. (1993), \enquote{Semiparametric least squares (SLS) and weighted SLS estimation of single-index models,} \textit{Journal of Econometrics}, 58, 71--120.

\bibitem[{Janson(1997)}]{J97}
Janson, S. (1997), \textit{Gaussian Hilbert Spaces}, Cambridge University Press.

\bibitem[{Kaji(2021)}]{Kaji21}
Kaji, T. (2021), \enquote{Theory of Weak Identification in Semiparametric Models,} \textit{Econometrica}, 89, 733--763.

\bibitem[{Ketz(2018)}]{K18}
Ketz, P. (2018), \enquote{{Subvector inference when the true parameter vector may be near or at the boundary},} \textit{Journal of Econometrics}, 207, 285--306.

\bibitem[{Kleibergen(2005)}]{K05}
Kleibergen, F. (2005), \enquote{Testing Parameters in GMM Without Assuming that They Are Identified,} \textit{Econometrica}, 73, 1103--1123.

\bibitem[{{Le Cam}(1986)}]{LC86}
{Le Cam}, L.~M. (1986), \textit{Asymptotic Methods in Statistical Decision Theory}, New York, NY, USA: Springer.

\bibitem[{{Le Cam} and Yang(2000)}]{LCY00}
{Le Cam}, L.~M. and Yang, G.~L. (2000), \textit{Asymptotics in Statistics: Some Basic Concepts}, New York, NY, USA: Springer, 2nd ed.

\bibitem[{Lee and Mesters(2024{\natexlab{a}})}]{LM21}
Lee, A. and Mesters, G. (2024{\natexlab{a}}), \enquote{Locally Robust Inference for Non-Gaussian Linear Simultaneous Equations Models,} \textit{Journal of Econometrics}, 240, 105647.

\bibitem[{Lee and Mesters(2024{\natexlab{b}})}]{LM21-S}
--- (2024{\natexlab{b}}), \enquote{Supplement to `Locally Robust Inference for Non-Gaussian Linear Simultaneous Equations Models',} \textit{Journal of Econometrics Supplementary Material}, 240.

\bibitem[{Lehmann and Romano(2005)}]{LR05}
Lehmann, E.~L. and Romano, J.~P. (2005), \textit{Testing Statistical Hypotheses}, New York, NY, USA: Springer, 3rd ed.

\bibitem[{McCloskey(2017)}]{McC17}
McCloskey, A. (2017), \enquote{Bonferroni-based size-correction for nonstandard testing problems,} \textit{Journal of Econometrics}, 200, 17--35.

\bibitem[{Mikusheva and Sun(2022)}]{MS21}
Mikusheva, A. and Sun, L. (2022), \enquote{{Inference with Many Weak Instruments},} \textit{The Review of Economic Studies}, 89, 2663--2686.

\bibitem[{Moreira(2009)}]{M09}
Moreira, M.~J. (2009), \enquote{Tests with correct size when instruments can be arbitrarily weak,} \textit{Journal of Econometrics}, 152, 131--140.

\bibitem[{Newey(1990)}]{N90}
Newey, W.~K. (1990), \enquote{Semiparametric efficiency bounds,} \textit{Journal of Applied Econometrics}, 5, 99--135.

\bibitem[{Newey(1991)}]{N91}
--- (1991), \enquote{Estimation of Tobit models under conditional symmetry,} in \textit{Nonparametric and Semiparametric Methods in Econometrics and Statistics: Proceedings of the Fifth International Symposium in Economic Theory and Econometrics}, eds. Barnett, W.~A., Powell, J., and Tauchen, G.~E., Cambridge University Press.

\bibitem[{Newey and McFadden(1994)}]{NMcF94}
Newey, W.~K. and McFadden, D. (1994), \enquote{Chapter 36 Large sample estimation and hypothesis testing,} Elsevier, vol.~4 of \textit{Handbook of Econometrics}, pp. 2111--2245.

\bibitem[{Newey and Stoker(1993)}]{NS93}
Newey, W.~K. and Stoker, T.~M. (1993), \enquote{Efficiency of Weighted Average Derivative Estimators and Index Models,} \textit{Econometrica}, 61, 1199--1223.

\bibitem[{Neyman(1959)}]{N59}
Neyman, J. (1959), \enquote{Optimal Asymptotic Tests of Composite Statistical Hypotheses,} in \textit{Probability and Statistics, the Harald Cram\'{e}r Volume}, ed. Grenander, U., New York, USA: Wiley.

\bibitem[{Neyman(1979)}]{N79}
--- (1979), \enquote{C($\alpha$) Tests and Their Use,} \textit{Sankhyā: The Indian Journal of Statistics, Series A (1961-2002)}, 41, 1--21.

\bibitem[{Ramsay(1988)}]{R88}
Ramsay, J.~O. (1988), \enquote{Monotone Regression Splines in Action,} \textit{Statistical Science}, 3, 425--441.

\bibitem[{Rao and Mitra(1971)}]{RM71}
Rao, C.~R. and Mitra, S.~K. (1971), \textit{Generalized Inverse of Matrices and its Applications}, New York, NY, USA: John Wiley \& Sons, Inc.

\bibitem[{Ritov and Bickel(1990)}]{RB90}
Ritov, Y. and Bickel, P.~J. (1990), \enquote{{Achieving Information Bounds in Non and Semiparametric Models},} \textit{The Annals of Statistics}, 18, 925 -- 938.

\bibitem[{Romano and Shaikh(2012)}]{RS12}
Romano, J.~P. and Shaikh, A.~M. (2012), \enquote{{On the uniform asymptotic validity of subsampling and the bootstrap},} \textit{The Annals of Statistics}, 40, 2798 -- 2822.

\bibitem[{Rothenberg(1971)}]{R71}
Rothenberg, T.~J. (1971), \enquote{Identification in Parametric Models,} \textit{Econometrica}, 39, 577--591.

\bibitem[{Rotnitzky et~al.(2000)Rotnitzky, Cox, Bottai, and Robins}]{RCBR00}
Rotnitzky, A., Cox, D.~R., Bottai, M., and Robins, J. (2000), \enquote{Likelihood-Based Inference with Singular Information Matrix,} \textit{Bernoulli}, 6, 243--284.

\bibitem[{Royden and Fitzpatrick(2010)}]{RF10}
Royden, H. and Fitzpatrick, P. (2010), \textit{Real Analysis}, Pearson Prentice Hall.

\bibitem[{Rudin(1991)}]{R91}
Rudin, W. (1991), \textit{Functional analysis}, McGraw Hill, Inc., 2nd ed.

\bibitem[{Serfozo(1982)}]{S82}
Serfozo, R. (1982), \enquote{Convergence of Lebesgue Integrals with Varying Measures,} \textit{Sankhyā: The Indian Journal of Statistics, Series A (1961-2002)}, 44, 380--402.

\bibitem[{Staiger and Stock(1997)}]{SS97}
Staiger, D. and Stock, J.~H. (1997), \enquote{Instrumental Variables Regression with Weak Instruments,} \textit{Econometrica}, 65, 557--586.

\bibitem[{Stock and Wright(2000)}]{SW00}
Stock, J.~H. and Wright, J.~H. (2000), \enquote{{GMM} with Weak Identification,} \textit{Econometrica}, 68, 1055--1096.

\bibitem[{Strasser(1985)}]{S85}
Strasser, H. (1985), \textit{Mathematical Theory of Statistics: Statistical Experiments and Asymptotic Decision Theory}, W. de Gruyter.

\bibitem[{{van der Vaart}(1991)}]{vdV91b}
{van der Vaart}, A.~W. (1991), \enquote{An Asymptotic Representation Theorem,} \textit{International Statistical Review / Revue Internationale de Statistique}, 59, 97--121.

\bibitem[{{van der Vaart}(1998)}]{vdV98}
--- (1998), \textit{Asymptotic Statistics}, New York, NY, USA: Cambridge University Press.

\bibitem[{{van der Vaart}(2002)}]{vdV02}
--- (2002), \enquote{Semiparametric Statistics,} in \textit{{Lectures on Probability Theory and Statistics: Ecole d'Et\'{e} de Probabilit\'{e}s de Saint-Flour XXIX - 1999}}, ed. Bernard, P., Springer.

\bibitem[{{van der Vaart} and Wellner(1996)}]{vdVW96}
{van der Vaart}, A.~W. and Wellner, J.~A. (1996), \textit{Weak Convergence and Empirical Processes}, New York, NY, USA: Springer-Verlag New York, Inc.

\bibitem[{Wald(1943)}]{W43}
Wald, A. (1943), \enquote{Tests of Statistical Hypotheses Concerning Several Parameters When the Number of Observations is Large,} \textit{Transactions of the American Mathematical Society}, 54, 426--482.

\end{thebibliography}
\appendix

\section{Proofs of the main results}\label{sec:proofs}

\begin{proof}[Proof of Proposition \ref{prop:asymp-dist}]
	Combination of Assumptions \ref{ass:LAN} and \ref{ass:joint-conv} yields
	 \begin{equation*}
		 \left(g\transp_{n}, L_{n}(h)\right)\transp \overset{P_{n}}{\weakconv }\mc{N}\left(
				 \begin{pmatrix}
					 0 \\
					 -\frac{1}{2}\sigma(h)
				 \end{pmatrix}
					 , \begin{pmatrix}
						 V & \tau\transp\Sigma_{21}\transp\\
						  \Sigma_{21}\tau & \sigma(h)
				 \end{pmatrix}
		 \right).
	 \end{equation*}
	 By Le Cam's third Lemma $ g_{n} \overset{P_{n, h}}{\weakconv} Z_{\tau}\sim \mc{N}\left(\Sigma_{21}\tau, V\right)$.
	 The second claim follows from the first with Assumption \ref{ass:consistent}\ref{ass:consistent:itm:ghat}, Remark \ref{rem:mutual-contiguity} and Slutsky's Theorem. By the second claim, Assumption \ref{ass:consistent}\ref{ass:consistent:itm:lambdahat} and standard arguments, $\hat{S}_{n, \theta_0}\overset{P_{n, h}}{\weakconv} Z_{\tau}\transp V^\dagger Z_{\tau}$.
 \end{proof}

 \begin{proof}[Proof of Theorem \ref{thm:psi-pwr-local-alt}]
    If $r\ge 1$, as $\hat{r}_n\xrightarrow{P_{n}}r$, $P_{n}\{c_n = c_\alpha\} \to 1$. By Proposition \ref{prop:asymp-dist}, Remark \ref{rem:mutual-contiguity} and Slutsky's Theorem, $\hat{S}_{n, \theta_0} - c_n \weakconv S - c$ under $P_{n, h}$ where $S \sim \chi^2_r\left(a\right)$.
    By the Portmanteau Theorem, 
    \begin{equation*}
        \limn P_{n, h} \psi_{n, \theta_0} = \limn P_{n, h} \left(\hat{S}_{n, \theta_0} > c_n\right) = L\{S-c > 0\} = 1 - \P(\chi_r^2(a)\le c_\alpha),
    \end{equation*}
    for $L$ the law of $S$.
    If $r=0$, $\rank(\hat{\Lambda}_{n, \theta_0}) \xrightarrow{P_{n}}0$ $\implies$  $P_{n}R_n \to 1$ for $R_n\define \{\hat{\Lambda}_{n, \theta_0} = 0\}$. On $R_n$ $\hat{S}_{n, \theta_0} = 0$ $\implies$ $\psi_{n, \theta_0} = 0$. By Remark \ref{rem:mutual-contiguity}, $P_{n,  h} \psi_{n, \theta_0}  \le 1 - P_{n, h}R_n \to 0$.
\end{proof}

\begin{proof}[Proof of Corollary \ref{cor:psi-locally-uniformly-regular}]
    By Theorem \ref{thm:psi-pwr-local-alt}, $\uppi_n(h) \to \uppi(h)$ ($h\in H$), pointwise. Since the $\uppi_n$ are asymptotically equicontinuous on $K$, the convergence is uniform on $K$.
\end{proof}

\begin{proof}[Proof of Lemma \ref{lem:level-alpha-unif-equicontinuity-TV}]
    Immediate from $  | P_{n, h} \psi_{n, \theta_0} - P_{n, h'} \psi_{n, \theta_0}| \le d_{TV}(P_{n, h}, P_{n, h'})$.
\end{proof}

\begin{proof}[Proof of Lemma \ref{lem:IP-GP}]
    For $a, b\in \R$, $h_1, h_2\in H$, $\Delta_{n}(a_1 h_1 + a_2 h_2) = a_1 \Delta_{n}h_1 + a_2  \Delta_{n}h_2$
    and so $\Delta(a_1 h_1 + a_2 h_2) = a_1 \Delta(h_1) + a_2\Delta(h_2)$, hence $\Delta$ is linear.
    We now establish $K$ is a well-defined covariance kernel. 
    For $h\in H$, 
    $(\|\Delta_{n}h\|^2)_{n\in \N}$ is Cauchy.
     Letting $K_{n}(h, g)\define P_{n}\left[\Delta_{n}h\Delta_{n}g\right]$ and using  Cauchy -- Schwarz
    \begin{align*}
        \left|
            K_{n} (h, g)  - K_{m} (h, g)
         \right| 
        \le \|\Delta_{n} h - \Delta_{m}h\|\|\Delta_{n}g\| + \|\Delta_{m}h\| \|\Delta_{n}g - \Delta_{m} g\|,
    \end{align*}  
    hence $(K_{n}(h, g))_{n\in \N}$ is also Cauchy and thus has a limit.
     Bilinearity and symmetry are straightforward to check.
    For positive semi-definiteness, let $h_1, \ldots, h_K \in H$, $a\in \R^K$. As $\Delta_{n}h\in L_2^0(P_{n})$, $\mc{K}_n\define [K_{n}(h_k, h_j)]_{k, j=1}^K$ is a covariance matrix, hence $ \sum_{k=1}^K\sum_{j=1}^K a_ka_j K_{n}(h_k, h_j) = a\transp \mc{K}_n a \ge 0$ for each $n\in \N$ and 
    hence the same holds with $K_{n}$ and $\mc{K}_n$ replaced by $K$ and $\mc{K}\define [K(h_k, h_j)]_{k, j=1}^K$.
    
    By Assumption \ref{ass:LAN} and the fact that $K(h, h) = \sigma(h)$, $\Delta h\sim \mc{N}(0, K(h, h))$. That $\Delta$ is a mean-zero Gaussian process with covariance kernel $K$ then follows from the Cram\'{e}r -- Wold Theorem as $\sum_{k=1}^K a_k \Delta h_k  \sim \mc{N}\left(0, a\transp \mc{K}a\right)$ and 
    \begin{equation*}
        \sum_{k=1}^K a_k \Delta_{n}(h_k) = \Delta_{n}\left(\sum_{k=1}^K a_k  h_k\right) \overset{P_{n}}{\weakconv} \Delta\left(\sum_{k=1}^K a_k  h_k\right) = \sum_{k=1}^K a_k \Delta h_k. \qedhere
    \end{equation*}
\end{proof}

\begin{proof}[Proof of Proposition \ref{prop:conv-exp}]
    Remark \ref{rem:mutual-contiguity} and the transitivity of (mutual) contiguity ensures that the experiments $\ms{E}_{n}$ are contiguous.
    By Theorem 61.6 of \cite{S85} it suffices to show that the finite dimensional marginal distributions (fdds) of $L_{n}$ converge (under $P_{n}$) to those of $L$, where $L(h) \define\Delta h - \frac{1}{2}\|h\|^2$. This follows as the fdds of $\Delta_{n}$ converge to those of $\Delta$ (under $P_n$), by the Cram\'{e}r -- Wold Theorem.
\end{proof}

\begin{proof}[Proof of Proposition \ref{prop:exp-equiv-shift}]
    Let $G_{[0]}\define P_{0}$. Define $Z:\mb{H} \to L_2(\Omega, \F, G_{[0]})$ by $Z[h] \define \Delta(h)$ for an arbitrary $h \in \pi_V^{-1}([h])$, where $\pi_V$ is the quotient map from $H\to \mb{H}$.
    This is a standard Gaussian process for $\mb{H}$.
    Define $G_{[h]}$ by $\deriv{G_{[h]}}{G_{[0]}} = \exp\left(Z[h] - \frac{1}{2}\|[h]\|_K^2\right)$.
    $\ms{G}$ is a Gaussian shift on $(\mb{H}, \IP{\cdot}{\cdot}_K)$ \cite[][Theorem 69.4]{S85}. 
    For any $h\in H$ we have that $Z[h] = \Delta g$ for some $g\in \pi_V^{-1}([h])$ and $\Delta h = \Delta g$ $P_{ 0}$--almost surely.
     Since $\|h\|_K = \|[h]\|_K$, $P_0$-a.s., 
    $\deriv{G_{[h]}}{G_{[0]}} =  \exp\left(Z[h] - \frac{1}{2}\|[h]\|_K^2\right) = \exp\left(\Delta h - \frac{1}{2}\|h\|_K^2\right) = \deriv{P_{ h}}{P_{0}}$.
    As each $P_{ h} \ll P_0$ and $G_{[h]}\ll G_{[0]}$, and $P_0 = G_{[0]}$,
    $ d_{TV}(P_{ h}, G_{[h]}) = \frac{1}{2}\int \left| \deriv{P_{ h}}{P_{0}}- \deriv{G_{[h]}}{P_{0}}\right| \darg{P_{ 0}} = \frac{1}{2}\int \left| \deriv{P_{ h}}{P_{0}} - \deriv{G_{[h]}}{G_{[0]}}\right| \darg{P_{ 0}} = 0$.
\end{proof}

\begin{proof}[Proof of Lemma \ref{lem:ker-effinfo-is-fzero}]
    By straightforward calculation
    \begin{equation}\label{lem:ker-effinfo-is-fzero:eq:IP-decomp}
    \IP{[\tau, b]}{[t, g]}_K = \tau\transp \effinfo t + \IP{\Pi [\tau, 0] + [0, b]}{\Pi [t, 0] +[ 0, g]}_K
    \end{equation}
    This and $\Pi[(\tau, 0)]\in \ker \pi_1$ imply
    $ \|[\tau]\|^2 = \inf_{b\in B}\|[\tau, b]\|^2_{K} 
    = \tau\transp \effinfo\tau + \inf_{[h]\in \ker \pi_1} \| \Pi[\tau, 0] - [h]\|^2_{K} = \tau\transp \effinfoarg{\gamma}\tau$.
    Hence, $\|\tau\| = \|[\tau]\| = 0$ $\implies$ $\tau\transp \effinfoarg{\gamma}\tau = 0$ $\implies$ $\effinfoarg{\gamma}^{1/2}\tau=0$, and so
    $\effinfoarg{\gamma} \tau =  0$.
    Conversely $\tau\in \ker \effinfoarg{\gamma}$ $\implies$ $\|\tau\|^2 = 0$ $\implies$ $\|\tau\| = 0$.
\end{proof}

\begin{proof}[Proof of Theorem \ref{thm:two-sided-pwr-bound}]
    Define the bounded linear map $T: \overline{\mb{H}} \to \R$ according to 
    $T[h] \define \IP{\Pi^\perp [1, 0]}{\Pi^\perp [h]}_{K} = \IP{\Pi^\perp [1,0]}{[h]}_{K}$.
    For any $[h] = [\tau, b]\in \mb{H}$, 
    \begin{equation}\label{thm:two-sided-pwr-bound-contrast:eq:T-equiv-effinfo-IP}
        T[h]  = \IP{\Pi^\perp [1, 0]}{\Pi^\perp [\tau, b]}_{K} = \effinfo\tau.
    \end{equation}

    If $\effinfo = 0$, \eqref{thm:two-sided-pwr-bound:eq:power-bound} follows from Proposition \ref{prop:asy-size-alpha-and-r-0-imply-no-power}, so assume $\effinfo \neq 0$.
    Any unbiased level $\alpha$ test $\phi$ of $T[h] = 0$ against $T[h] \neq 0$ in the (restricted) Gaussian shift $\ms{G}$ satisfies 
      \begin{equation}\label{thm:two-sided-pwr-bound-contrast:eq:limit-exp-pwr-bound}
        G_{[h]} \phi\le 1 - \Phi\left(z_{\alpha/2} - \effinfo^{1/2}\tau\right)  + 1- \Phi\left(z_{\alpha/2} +  \effinfo^{1/2}\tau\right),   
    \end{equation}
    \cite[][Lemma 71.5]{S85}.
    By Proposition \ref{prop:conv-exp}, $\ms{E}_{n} \weakconv \ms{E}$; $\ms{E}$ is dominated.
    Let $\uppi_n(h)\define P_{n, h}\phi_n$ and fix an arbitrary $h^\star$.
    There is a subsequence $(\uppi_{n_m})_{m\in \N}$ such that $\lim_{m\to\infty} \uppi_{n_m}(h^\star) = \limsupn \uppi_n(h^\star)$. Since $[0, 1]^{H}$ is compact in the product topology, there is a subnet $(\uppi_{n_{m(s)}})_{s\in S}$ and a function $\uppi: H \to [0, 1]$ such that $\lim_{s\in S} \uppi_{n_{m(s)}}(h) = \uppi(h)$ for all $h\in H$. 
    By our hypotheses and equation \eqref{thm:two-sided-pwr-bound-contrast:eq:T-equiv-effinfo-IP} for any $h_0$ such that $[h_0]\in \ker T\cap \mb{H}$ and any $h_1$ such that $[h_1] \in \mb{H}\setminus (\ker T\cap \mb{H})$
    \begin{equation}\label{thm:two-sided-pwr-bound-contrast:eq:unbiased}
        \uppi(h_0) = \lim_{s\in S}\uppi_{n_{m(s)}}(h_0) \le \alpha \le \lim_{s\in S}\uppi_{n_{m(s)}}(h_1) = \uppi(h_1).
    \end{equation}
    There exists a test $\phi$ in $\ms{E}$ with power function $\uppi$ \cite[][Theorem 7.1]{vdV91b}. \eqref{thm:two-sided-pwr-bound-contrast:eq:unbiased} and Proposition \ref{prop:exp-equiv-shift} ensure that $\phi$ is an unbiased, level $\alpha$ test of $\ker T\cap \mb{H}$ against $\mb{H} \setminus (\ker T\cap \mb{H})$ in $\ms{G}$. Conclude by combining \eqref{thm:two-sided-pwr-bound-contrast:eq:limit-exp-pwr-bound} and (by Proposition \ref{prop:exp-equiv-shift})
    \begin{equation*}
        \limsupn P_{n, h^\star} \phi_n = \lim_{m\to\infty} \uppi_{n_m}(h^\star) = \uppi(h^\star) = P_{h^\star} \phi =  G_{[h^\star]} \phi.\qedhere
    \end{equation*}

\end{proof}

\begin{proof}[Proof of Corollary \ref{cor:psi-two-sided}]
    By Theorem \ref{thm:psi-pwr-local-alt},  $   \limn P_{n, h_n}\psi_{n, \theta} = 1 - \P(Z^2> c_\alpha)$ where $Z\sim \mc{N}\left(\effinfo^{1/2}\tau,1\right)$.
    $1 - \P(Z^2 > c_\alpha)$ is equal to the RHS of \eqref{eq:two-sided-power-bound}.
\end{proof}

\begin{proof}[Proof of Theorem \ref{thm:maximin-pwr-bound}]
    By Lemma \ref{lem:ker-effinfo-is-fzero}, $\mb{H}_1 = \quotient{\R^{d_\theta}}{\ker \effinfo}$. $\pi_1:\mb{H}\to \mb{H}_1$ is surjective: for any $[\tau]\in\mb{H}_1$ let $t\in \pi_{\ker \effinfo}^{-1}(\{[\tau]\})$ where $\pi_{\ker \effinfo}$ is the quotient map from $\R^{d_\theta}$ to $\mb{H}_1$. Then $\pi_1[t, 0] = [t] = [\tau]$. It follows that $\dim \ran \pi_1 = \operatorname{codim}\ker \pi_1 = r$.
    By linearity and $[0, b]\in \ker \pi_1$, $\Pi[\tau, b] = \Pi[\tau, 0] + [0, b]$. This with Lemma \ref{lem:ker-effinfo-is-fzero}
     yields $\|[\tau, b] - \Pi[\tau, b]\|_{K}^2 = \|[\tau, 0] - \Pi[\tau, 0]\|_{K}^2 = \|[\tau]\|^2 = \tau\transp \effinfo\tau$. Define the sets
    \begin{equation*}
        M_a\define \left\{[\tau, b]\in \mb{H}: \tau\transp \effinfo\tau = a\right\}, \qquad
        \overline{M}_a\define \left\{[\tau, b]\in \overline{\mb{H}}: \tau\transp \effinfo\tau = a\right\}.
        \end{equation*}  
        It is straightforward to check that $\cl M_a = \overline{M}_a$.
    Suppose that $\phi$ is a test on $\mathscr{G}$ with $G_{[0]}\phi \le \alpha$.
    First suppose $a>0$. $\phi$ is a level $\alpha$ test of  $K_0:\{[0]\}$ against $K_1:[\ker \pi_1]^\perp \setminus \{[0]\}$ in the restriction of the standard Gaussian shift experiment on $[\ker \pi_1]^\perp$.
    By Theorem 30.2 in \cite{S85} 
    \begin{equation} \label{thm:maximin-pwr-bound:eq:inf-bound-1}
        \inf_{[h]\in M_a}  G_{[h]}\phi = \inf_{[h]\in \bar{M}_a}  G_{[h]}\phi \le  \inf_{[h]\in \bar{M}_a \cap [\ker \pi_1]^\perp}  G_{[h]}\phi \le 1 - \P(\chi^2_r(a) \le c_r),
    \end{equation}
    since $[h]\mapsto G_{[h]} \phi$ is continuous.
    If, instead, $a =0$, 
    note that $[0]\in M_0$ and so, 
     \begin{equation}\label{thm:maximin-pwr-bound:eq:inf-bound-2}
        \inf_{[h]\in M_0} G_{[h]} \phi \le G_{[0]} \phi \le \alpha = 1 - \P(\chi_r^2(0)\le c_r).
    \end{equation}
    
    Fix $a\ge 0$ and let $\mc{R}\define 1 - \P(\chi^2_r(a) \le c_r)$. 
    Let $\uppi_n(h)\define P_{n,h} \phi_n$ and define $\beta_{n}\define \inf \left\{P_{n, h} \phi_n: h = (\tau, b) \in H,\ \tau\transp \effinfo\tau = a \right\}$. Suppose that $\limsup_{n\to\infty} \beta_n \ge \mc{R}+ \varepsilon$ for some $\varepsilon>0$.
    Hence, for some subsequence $(n_m)_{m\in \N}$, $\lim_{m\to\infty} \beta_{n_m} \ge \mc{R} + \varepsilon$.
    Since $[0, 1]^{H}$ is compact in the product topology, there is a subnet $(\uppi_{n_{m(s)}})_{s\in S}$ and a function $\uppi: H \to [0, 1]$ such that $\lim_{s\in S} \uppi_{n_{m(s)}}(h) = \uppi(h)$ for all $h\in H$. Take any $h$ such that     
   $[h]\in M_a$. The preceding display implies 
   \begin{equation}\label{thm:maximin-pwr-bound:eq:contradict}
    \uppi(h) = \lim_{s\in S} \uppi_{n_{m(s)}}(h) \ge \lim_{s\in S} \inf \left\{\uppi_{{n_{m(s)}}}(h): h = (\tau, b) \in H,\ \tau\transp \effinfo\tau =a \right\}\ge \mc{R} + \varepsilon.
   \end{equation} 
   
    By Proposition \ref{prop:conv-exp}, $\ms{E}_{n} \weakconv \ms{E}$; $\ms{E}$ is dominated. There is a test $\phi$ in $\ms{E}$ with power function $\uppi$ \cite[][Theorem 7.1]{vdV91b}.
   Consider the restriction of $\ms{G}$ to $[\ker \pi_1]^\perp$. 
    By hypothesis, Corollary \ref{cor:h-eqiv-g-no-pwr-diff} and Proposition \ref{prop:exp-equiv-shift} $ G_{[0]}\phi  =  P_{0}\phi = \uppi(0) = \lim_{s\in S}\uppi_{n_m(s)}(0)\le \limsup\uppi_{n}(0)\le \alpha$, 
    hence $\phi$ is a test of level $\alpha$ of $K_0$ against $K_1$ in this experiment, and $ \inf_{[h]\in M_a}G_{[h]} \phi  =\inf_{h: [h]\in M_a}P_{h} \phi  =\inf_{h: [h]\in M_a}\uppi(h) \ge \mc{R} + \varepsilon$,
    by \eqref{thm:maximin-pwr-bound:eq:contradict} and Proposition \ref{prop:exp-equiv-shift}, but this contradicts  $\eqref{thm:maximin-pwr-bound:eq:inf-bound-1}$ if $a>0$ or $\eqref{thm:maximin-pwr-bound:eq:inf-bound-2}$ if $a=0$.
\end{proof}

\begin{proof}[Proof of Corollary \ref{cor:psi-maximin}]
    Equation \eqref{cor:psi-maximin:eq-1} follows from Theorem \ref{thm:psi-pwr-local-alt}.
    For equation \eqref{cor:psi-maximin:eq-2}, let $f_n(h)\define P_{n, h}\psi_{n, \theta_0}$. By \eqref{cor:psi-maximin:eq-1} and the asymptotic equicontinuity,
    $\limn f_n(h) = 1 - \P\left(\chi^2_r\left(\tau\transp \effinfo \tau\right) \le c_r\right) \eqqcolon f(h)$, 
    uniformly on $K_a$.
    Conclude that if $h_n \to h \in K_a$,
    \begin{equation}\label{cor:psi-maximin:eq-3} 
        \limn f_n(h_n) =  f(h) \ge f_\star \define 1 - \P\left(\chi^2_r\left(a\right) \le c_r\right).
    \end{equation}
    If \eqref{cor:psi-maximin:eq-2} fails there is a sequence $h_n\in K_a$ with $\limsupn f_n(h_n) < f_\star$. Extract a subsequence $h_{n_m}\to h\in K_a$.
     Let $h_m^*\define h_{n_1}$ for $m =1, \ldots, n_1$ and $h_m^*\define h_{n_k}$ for $n_{k} \le m < n_{k+1}$. $f_{n_m}(h_{n_m})$ is a subsequence of $f_{m}(h_m^*)$  and $h_m^*\to h$, so by \eqref{cor:psi-maximin:eq-3} $ \lim_{m\to\infty} f_{n_m}(h_{n_m}) = \lim_{m\to\infty} f_m(h_m^*) = f(h) \ge f_\star > \limsupn f_n(h_n)$. 
\end{proof}

\begin{proof}[Proof of Proposition \ref{prop:asy-size-alpha-and-r-0-imply-no-power}]
    By \eqref{lem:ker-effinfo-is-fzero:eq:IP-decomp}, $r =0$ implies $\|[h] - \Pi [h]\|_{K} = 0$ and so $[h] = \Pi[h] \in \ker \pi_1$.
    Hence there is a $h^* \in H_{0}$ with $\|h - h^*\|_{K} = 0$. By Corollary \ref{cor:h-eqiv-g-no-pwr-diff},
    \begin{equation*}
        \limsup P_{n, \gamma, h}\phi_n \le \limsupn P_{n, \gamma, h^*} \phi_n + \limsupn | P_{n, \gamma, h^*} \phi_n - P_{n, \gamma, h} \phi_n | \le \alpha.\qedhere
    \end{equation*}
\end{proof}

\begin{proof}[Proof of Theorem \ref{thm:pwr-bounds-effscr}]
    Since $\limn P_n[\dotscrarg{n}g_n\transp] = \limn P_n[\dotscrarg{n}\effscrarg{n}\transp]$ 
    we may assume $g_n = \effscrarg{n}$. By Theorem 12.14 in \cite{R91},
    $\effinfoarg{n}\define P_n[\effscrarg{n}\effscrarg{n}\transp] = P_n[\dotscrarg{n}\effscrarg{n}\transp]$.
    Set $K_n(h, g)\define P_{n}[\Delta_{n}h\Delta_{n} g]$ and let $\mathsf{G}_{n}$ be a zero-mean Gaussian process with covariance kernel $K_n$. There exists a Hilbert space isomorphism, $Z_n:\cl \{\Delta_{n}h : h\in H\} \to \cl \{\mathsf{G}_nh: h\in H\}$ \cite[][Theorem 1.23]{J97}.
    Let $R\define \Pi\left[\cdot\middle|\left\{\mathsf{G}_{n}(0, b):b \in B\right\}^\perp\right]$
    and $Q\define \Pi\left[\cdot\middle|\{\Delta_{n}(0, b): b\in B\}^\perp\right]$.
    $ R\mathsf{G}_nh = R Z_n(\Delta_{n} h) = Z_n Q Z_n^{-1}Z_n(\Delta_{n} h) = Z_n Q\Delta_{n} h$ for $h\in H$ and 
    extends to elements in the closure by continuity. Hence 
    \begin{equation}\label{thm:pwr-bounds-effscr:eq:effinfo-ij-alt}
        \effinfoarg{n,ij}
         = P_n\left[\Delta_{n}(e_i, 0) Q\Delta_{n}(e_j, 0) \right]
          = \E\left[\mathsf{G}_n(e_i, 0)R \mathsf{G}_n(e_j, 0)\right].
    \end{equation}
   By Theorem 9.1 in \cite{J97}, 
    \begin{equation}\label{thm:pwr-bounds-effscr:eq:proj-J97}
        \E\left[\mathsf{G}_n(e_j, 0)\middle | \left\{\mathsf{G}_{n}(0, b):b \in B\right\}\right] = \Pi\left[\mathsf{G}_n(e_j, 0) \middle| \cl\left\{\mathsf{G}_{n}(0, b):b \in B\right\}\right],
    \end{equation}
    and so $\tilde{\mathsf{G}}_n(e_j, 0) \define \mathsf{G}_n(e_j, 0)-\E\left[\mathsf{G}_n(e_j, 0)\middle | \left\{\mathsf{G}_{n}(0, b):b \in B\right\}\right] =  R\mathsf{G}_n(e_j, 0)$. Then, 
    $\effinfoarg{n, ij} = \E\left[\mathsf{G}_n(e_i, 0)\tilde{\mathsf{G}}_n(e_j, 0)\right]$ 
    by \eqref{thm:pwr-bounds-effscr:eq:effinfo-ij-alt}. Set $\ms{G}_{n}\define \sigma(\{\mathsf{G}_n(0, b): b\in B\})$, $\ms{G}_{n}\define \sigma(\{\mathsf{G}(0, b): b\in B\})$ and define $X_n \define  \left(
        \mathsf{G}_n (e_i, 0),
        \E[\mathsf{G}_n(e_j, 0)  | \ms{G}_n]\right)$
        and $X\define \left(
            \mathsf{G} (e_i, 0),
            \E[\mathsf{G}(e_j, 0)| \ms{G}]
        \right)$,
    where $\mathsf{G}\define \Delta$. By \eqref{thm:pwr-bounds-effscr:eq:proj-J97} and $K_n(h, h)\to K(h, h)$ (Lemma \ref{lem:IP-GP}), 
    $(X_n)_{n\in \N}$ are uniformly square integrable Gaussian random vectors and $X_n\weakconv X$ (by Theorem \ref{sm:thm:cond-exp-weak-conv}). Combine with Lemma \ref{lem:calculation-of-effinfo} and Theorem 9.1 of \cite{J97}.
\end{proof}

\begin{proof}[Proof of Lemma \ref{lem:iid-LAN-DQM}]
    $R_n(h)\xrightarrow{P_n}0$ in \eqref{eq:LAN} and $Ah\in L_2^0(P)$ follows from \eqref{eq:dqm} \cite[][Lemma 3.10.11]{vdVW96}. Hence $\Delta_n h$ is uniformly square integrable (i.i.d) and $[\Delta_{n} h](W^{(n)}) = \G_n A h \weakconv \mc{N}(0, \int( A h)^2\dP)$ (CLT).
\end{proof}

\begin{proof}[Proof of Lemma \ref{lem:iid-orthocomp-joint-conv}]
    $P^n\left(\Delta_{n}h, g_{n}\transp\right) = 0$. By $g\in \{Db: b\in B\}^\perp$ and Assumption \ref{ass:iid}, the covariance matrix of $ \left(\Delta_{n}h, g_{n}\transp\right) $ (under $P^n$) is $\Sigma(h) = P\left[\begin{smallmatrix}
        [Ah]^2 & \tau\transp\dotscr g\transp\\
        g\dotscr\transp \tau & gg\transp
    \end{smallmatrix}\right]$.
    For each $h\in H$, the central limit theorem gives $  \left(\Delta_{n}h, g_{n}\transp\right) \overset{P^n}{\weakconv} \mc{N}(0, \Sigma(h))$.
\end{proof}

\begin{proof}[Proof of Corollary \ref{cor:iid-orth-projection-joint-conv}]
    $g\in \{D b: b\in B\}^\perp$. Apply Lemma \ref{lem:iid-orthocomp-joint-conv}.
\end{proof}

\clearpage


\renewcommand{\thesection}{S\arabic{section}}   
\renewcommand{\thetable}{S\arabic{table}}  
\renewcommand{\thefigure}{S\arabic{figure}} 
\renewcommand{\thelemma}{S\arabic{lemma}}
\renewcommand{\theexample}{S\arabic{example}}
\renewcommand{\theproposition}{S\arabic{proposition}}
\renewcommand{\theequation}{S\arabic{equation}}
\renewcommand{\theassumption}{S\arabic{assumption}}
\renewcommand{\thetheorem}{S\arabic{theorem}}
\renewcommand{\thecorollary}{S\arabic{corollary}}
\renewcommand{\thefootnote}{S\arabic{footnote}}
\renewcommand{\thepage}{S\arabic{page}}

\setcounter{footnote}{0}
\setcounter{section}{0}
\setcounter{table}{0}
\setcounter{figure}{0}
\setcounter{example}{0}
\setcounter{lemma}{0}
\setcounter{proposition}{0}
\setcounter{theorem}{0}
\setcounter{assumption}{0}
\setcounter{remark}{0}
\setcounter{corollary}{0}
\setcounter{equation}{0}

\title{\sc \Large Supplementary material for ``Locally regular and efficient tests in non-regular semiparametric models''}

\author{{Adam Lee\thanks{BI Norwegian Business School, \href{mailto:adam.lee@bi.no}{adam.lee@bi.no}.}}}

\date{\today}

\maketitle
\thispagestyle{empty}

\renewcommand{\abstractname}{}    

\begin{abstract}
    \singlespacing\noindent 
    This Supplementary Material contains the following sections: 
    \begin{itemize}
        \item[\ref{sm:sec:notation}:] Details on notation
        \item[\ref{sm:sec:additional_results}:] Additional results \& discussion
        \item[\ref{sm:sec:technicalities}:] Technicalities 
        \item[\ref{sm:sec:examples-detail}:] Additional details and proofs for the examples
        \item[\ref{sm:sec:extra-sim-results}:] Additional simulation details \& results
        \item[\ref{sm:sec:tbls-figs}:] Tables and Figures 
    \end{itemize}
\end{abstract}
\clearpage

\setcounter{page}{1}

\onehalfspacing
\section{Notation}\label{sm:sec:notation}

$x\define y$ means that $x$ is defined to be $y$. $d_\theta$ is the length of the vector $\theta$. The Lebesgue measure on $\R^K$ is denoted by $\lambda_{K}$ or $\lambda$ if the dimension is clear from context. $k$-times continuously differentiable functions belong to $\mc{C}^k$. $L_p(A, w)\define L_p(A, \mc{A}, \mu, w)$ for a measure space $(A, \mc{A}, \mu)$ and a weight function $w:A\to [0, \infty)$ is the weighted $L_p$ space consisting of (equivalence classes of) measurable functions $f:A\to \R$ such that $\int |f|^p w\dmu < \infty$. $L_p^0(\Omega, \mc{F}, P)$ is the subspace of $L_p(\Omega, \mc{F}, P)$ whose elements have zero-mean. 
 The standard basis vectors in $\R^K$ are $e_1, \ldots, e_K$. $M^\dagger$ is the Moore -- Penrose pseudoinverse of $M$. $Pf \define \int f \darg{P}$, $\EP f\define  \frac{1}{n}\sum_{i=1}^n f(Y_i)$ and $\G_n f \define  \sqrt{n}(\EP - P)f$. For sequences of probability measures $(Q_n)_{n\in \N}$ and $(P_n)_{n\in \N}$ where $Q_n$ and $P_n$ are defined on a common measurable space for each $n\in \N$, $Q_n \triangleleft P_n$ indicates that $(Q_n)_{n\in \N}$ is contiguous with respect to $(P_n)_{n\in \N}$. $Q_n \triangleleft \triangleright \;P_n$ indicates that $Q_n \triangleleft P_n$ and $P_n \triangleleft Q_n$. $X \independent Y$ indicates that random vectors $X$ and $Y$ are independent; $X\simeq Y$ indicates that they have the same distribution. $a\lesssim b$ means that $a\le Cb$ for some constant $C\in (0, \infty)$; $C$ may change from line to line. If $X$ is a topological space, $\cl X$ means the (topological) closure of $X$. $\mc{B}(X)$ are the Borel subsets of $X$. If $S$ is a subset of a vector space, $\lin S$ or $\Span S$ means the linear span of $S$. If $S$ is a subset of a topological vector space, $\cllin S$ or $\cl \Span S$ means the closure of the linear span of $S$. 
 If $S$ is a subset of an inner product space $(V, \IP{\cdot}{\cdot})$, $S^\perp$ is its orthogonal complement, i.e. $S^\perp = \{x\in V: \IP{x}{s} = 0 \text{ for all }s \in S\}$. If $S \subset V$ is complete the orthogonal projection of $x\in V$ onto $S$ is $\Pi(x | S)$. The total variation distance between measures $P$ and $Q$ defined on the measurable space $(\Omega, \mathcal F)$ is $d_{TV}(P,Q)$. $d_2$ is the Mallows-2 metric \cite[e.g.][Appendix 6]{BKRW98}.
 $\overset{P_n}{\weakconv}$ denotes weak convergence under the sequence of measures $(P_n)_{n\in \N}$. If the sequence of measures is clear from context, I write just $\weakconv$.

\section{Additional results \& discussion}\label{sm:sec:additional_results}

\subsection{Inference under shape constraints}\label{sm:shape-constraints}

A non-standard inference problem which has attracted substantial attention in statistics \& econometrics is inference when (finite-dimensional) nuisance parameters $\eta$ may be at, or close to, the boundary. See, amongst others, \cite{G94, A99, A01, K18}. In this scenario, as explained in detail in the aforementioned papers, the limiting distributions of extremum estimators are non-normal when the true parameter is at the boundary of the parameter space. In otherwise regular models, the same is true when the true parameter is ``close'' to this boundary, i.e. along local (contiguous) alternatives to such a boundary point, by virtue of Le Cam's third Lemma.

In consequence, the normal approximations which usually obtain for extremum estimators \cite[cf. ][]{NMcF94} can lead to either misleading inference or poor power. The literature contains examples of boundary problems where ``standard'' tests over-reject \cite[e.g.][]{AG10b} as well as examples where they are conservative and exhibit poor power \cite[e.g.][]{K18}.

Under regularity conditions, boundary - constrained estimators of the nuisance parameters typically remain $\sqrt{n}$ - consistent (albeit not asymptotically normal). Due to the approximate orthogonalisation \eqref{rem:orth:eq:orth}, plugging in any $\sqrt{n}$ - consistent estimator $\hat\eta_n$ of $\eta$ is typically sufficient to ensure that the resulting feasible moment function (i.e. $\hat{g}_{n, \theta} = g_{\theta, \hat\eta_n}$) achieves the same normal limit as in Proposition \ref{prop:asymp-dist}.

In the semiparametric setting a natural generalisation of this boundary - constrained phenomenon is that of inference when nuisance functions are estimated under shape restrictions which may be close to binding. 

\begin{example}[continues=ex:SIM-running-example, name = Single-index model]
    Suppose that the class $\ms{F}$ of permitted link functions $f$ in equation \eqref{eq:SIM-running-example-mdl} imposes a shape restriction. For instance, $\ms{F}$ may contain only monotonically increasing functions or convex functions.
\end{example}

Analogously to the parametric case, plugging in nuisance functions estimated under shape constraints causes no problems for C$(\alpha)$ style tests, which retain the same asymptotic distribution whether or not the constraints are (close to) binding. This phenomenon is explored in simulation (based on Example \ref{ex:SIM-running-example}) in Section \ref{sm:ssec:sim-extra-sim-results}. 

Note that the power results of Section \ref{sec:theory} typically do not apply to models with shape-constraints as -- like in the parametric boundary case -- the set $B$ of possible perturbations to $\eta$ will typically be a (linear) cone rather than a linear space.

\subsection{Uniform Local Asymptotic Normality}\label{ssec:ULAN}

$H$ is assumed to be a subset (containing 0) of a linear space equipped with a pseudometric.\footnote{Proposition \ref{prop:ULAN-LAN-equicontinuity} below is an adaptation of Theorem 80.13 in \cite{S85}.} 

\begin{assumption}[Uniform local asymptotic normality]\label{ass:ULAN}
	Equation \eqref{eq:LAN} holds and $R_n(h_n) \xrightarrow{P_{n}} 0$ for any $h_n\to  h$ in $H$. 
	Additionally, for each $h_n \to h$ in $H$, $(\Delta_{n}h_n)_{n\in \N}$ is uniformly square $P_{n}$-integrable and $\left(\Delta_{n}h_n, \Delta_{n}h
	\right)\transp
	\overset{P_{n}}{\weakconv} \mc{N}\left(0,  \sigma(h) \left[\begin{smallmatrix}
		1 & 1\\
	1 & 1
	\end{smallmatrix}\right]\right)$ where $ \sigma(h)\define \limn \|\Delta_{n}h\|^2$.
\end{assumption}

\begin{remark}
	The joint convergence of $(\Delta_n h_n, \Delta_n h)\transp$ in Assumption \ref{ass:ULAN} is needed because $H$ is not required to be linear. If $H$ is a linear space this follows from $d_2$ convergence of (the law of) $\Delta_n h_n$ to $\mc{N}(0, \sigma(h))$ and the Cram\'{e}r -- Wold Theorem.
\end{remark}

\begin{remark}\label{rem:Delta-unif-bounded}
	If $(\Delta_{n})_{n\in \N}$ is asymptotically equicontinuous on compact subsets $K\subset H$, then $h_n\to h$ in $H$ implies $\|\Delta_{n}(h_n - h)\| \to 0$. In consequence $(\Delta_{n} h)_{n\in \N}$ being uniformly square $P_{n}$-integrable and $\Delta_{n} h \overset{P_{n}}{\weakconv} \mc{N}(0, \sigma(h))$ for each $h\in H$, suffices for $(\Delta_{n} h_n)_{n\in \N}$ being uniformly square $P_{n}$-integrable and for any $h_n\to h\in H$
	\begin{equation*}
		\left(\Delta_{n} h_n,\  \Delta_{n}h\right)\transp = \begin{pmatrix}
			1 & 1\\
			0 & 1
		\end{pmatrix}\left(\Delta_{n} h_n - \Delta_{n}h,\  \Delta_{n}h\right)\transp
		 \overset{P_{n}}{\weakconv} \mc{N}\left(0, \sigma(h)\left[\begin{smallmatrix}
			1 & 1\\
			1 & 1
		\end{smallmatrix}\right]\right).
	\end{equation*}

    If $H$ is a Banach space (metrised by its norm), the equicontinuity of $(\Delta_{n})_{n\in \N}$ is guaranteed as uniform boundedness of $(\Delta_{n})_{n\in \N}$ (hence equicontinuity on $H$) is implied by uniform square $P_{n}$-integrability of $(\Delta_{n} h)_{n\in \N}$ for $h\in H$.
\end{remark}

\begin{proposition}\label{prop:ULAN-LAN-equicontinuity}
	Assumption \ref{ass:ULAN} is equivalent to Assumption \ref{ass:LAN} plus asymptotic equicontinuity on compact subsets $K\subset H$ of $(\Delta_{n})_{n\in \N}$ and $(h\mapsto P_{n, h})_{n\in \N}$ (in $d_{TV}$).
\end{proposition}
\begin{proof}

	Suppose Assumption \ref{ass:LAN} and the asymptotic equicontinuity conditions hold. Let $h_n\to h$ in $H$. By asymptotic equicontinuity of $(h\mapsto P_{n, h})_{n\in \N}$,
	\begin{equation*}
		\limn d_{TV}(P_{n, h_n}, P_{n, h}) =0 \quad \Longrightarrow\quad  \limn \int \left|
			\frac{p_{n, h_n}}{p_{n, 0}} - \frac{p_{n, h}}{p_{n, 0}}
		\right| \darg{P_{n, 0}} = 0.
	\end{equation*}
	In combination with (compact) asymptotic equicontinuity of $(\Delta_{n})_{n\in \N}$, this implies $R_{n}(h_n) - R_{n}(h) =  o_{P_{n}}(1)$. 
	That $(\Delta_{n}h_n)_{n\in \N}$ is uniformly square $P_{n}$-integrable and the joint weak convergence under $P_{n}$ follows from the
	argument in Remark \ref{rem:Delta-unif-bounded}.
	Conversely, suppose Assumption \ref{ass:ULAN} holds and let $h_n\to h$ in $H$.
	Then $\Delta_{n}(h_n-h) \xrightarrow{P_{n}}0$ and so $\|\Delta_{n} (h_n-h)\|^2\to 0$ by uniform square integrability \cite[][Theorem 2.7]{S82}.
	$d_{TV}(P_{n, h_n}, P_{n, h})\to 0$ holds by Lemma \ref{lem:seq-w-LR-to-0-no-pwr-diff} since
	\begin{equation*}
		L_{n}(h_n) - L_{n}(h) = \Delta_{n}h_n - \frac{1}{2} \|\Delta_{n}h_n\|^2 + R_{n}(h_n) - \left[\Delta_{n}h - \frac{1}{2} \|\Delta_{n}h\|^2 + R_{n}(h)\right],
	\end{equation*}
	and $R_{n}(h_n) =  o_{P_{n}}(1)$,  $R_{n}(h) =  o_{P_{n}}(1)$, $\|\Delta_{n} (h_n-h)\|^2\to 0$.
\end{proof}

In the i.i.d. case, Lemma \ref{lem:iid-LAN-DQM} recorded sufficient conditions for LAN (Assumption \ref{ass:LAN}). Similar sufficient conditions are available for ULAN (Assumption \ref{ass:ULAN}).

\begin{lemma}\label{lem:iid-ULAN-DQM}
	Suppose that Assumption \ref{ass:iid} holds and for each $h\in H$ equation \eqref{eq:dqm} holds with $A: \cllin H \to L_2(P)$ a bounded linear map. Then Assumption \ref{ass:ULAN} holds with $P_{n, h} = P_{h/\sqrt{n}}^n$ and $[\Delta_{n} h](W^{(n)}) = \G_n A h$.
\end{lemma}
\begin{proof}
	That $R_n(h_n)\xrightarrow{P_n}0$ in \eqref{eq:LAN} and that $Ah\in L_2^0(P)$ follows from \eqref{eq:dqm} by e.g. Lemma 3.10.11 in \cite{vdVW96}. This immediately implies that $\Delta_n h$ is uniformly square integrable by the i.i.d assumption and that $[\Delta_{n} h](W^{(n)}) = \G_n A h \weakconv \mc{N}(0, \sigma(h))$ for $\sigma(h)\define \int( A h)^2\dP$ by the central limit theorem. In view of Remark \ref{rem:Delta-unif-bounded} it remains to show that $\Delta_n$ is asymptotically equicontinuous on compact subsets $K\subset H$. This follows since $A$ is bounded: for $h_n \to h$, $\|\Delta_{n}(h_n - h)\| = \left\|\G_n A(h_n - h)\right\| = \|A (h_n - h)\| \le \|A\|\|h_n - h\| \to 0$.
\end{proof}

\subsection{Additional results on uniform local regularity}

\subsubsection{Asymptotic equicontinuity of power functions}

\begin{lemma}\label{lem:level-alpha-unif-equicontinuity-weak}
    Suppose the conditions of Theorem \ref{thm:psi-pwr-local-alt} hold and that $(H, d)$ is a pseudometric space. Let $\delta$ metrise weak convergence on the space of probability measures on $(\R, \mc{B}(\R))$ and let $Q_{n, h} = P_{n, h}\circ \hat{S}_{n, \theta_0}^{-1}$. 
    Suppose that on a subset $K\subset H$,
    \begin{enumerate}
        \item $(h\mapsto Q_{n, h})_{n\in \N}$ is asymptotically equicontinuous in $\delta$;\label{lem:level-alpha-unif-equicontinuity-weak.itm.Q}
        \item $(h\mapsto P_{n, h}(\hat{r}_{n, \theta_0} = r))_{n\in \N}$ is asymptotically equicontinuous;\label{lem:level-alpha-unif-equicontinuity-weak.itm.r}
        \item $(h\mapsto P_{n,  h}(\hat{\Lambda}_{n, \theta_0} = 0))_{n\in \N}$ is asymptotically equicontinuous;\label{lem:level-alpha-unif-equicontinuity-weak.itm.Lambda}
    \end{enumerate}
    then $(h\mapsto P_{n, h} \psi_{n, \theta_0})_{n\in \N}$ is asymptotically equicontinuous on $K$.
\end{lemma}
\begin{proof}
    First suppose $r\ge 1$. By asymptotic equicontinuity of $h\mapsto Q_{n, h}$ and $h\mapsto P_{n, h}(\hat{r}_{n, \theta_0} = r)$ on $K$, for any $h_n \to h$ (through $K$), $\delta(Q_{n, h_n}, Q_{n, h}) \to 0$ and $|P_{n, h}(\hat{r}_{n, \theta_0} = r) - P_{n, h_n}(\hat{r}_{n, \theta_0} = r )| \to0$. 
    Since $\hat{r}_{n, \theta_0} \xrightarrow{P_{n, h}}r$ (Assumption \ref{ass:consistent} \ref{ass:consistent:itm:rank} and Remark \ref{rem:mutual-contiguity}), $\hat{r}_{n, \theta_0} \xrightarrow{P_{n, h_n}}r$. Hence, under $P_{n, h_n}$,
    \begin{equation*}
        \hat{S}_{n, \theta_0} - c_n  \weakconv S - c_r, \quad S\sim \chi^2_r(a) \implies  P_{n, h_n}\psi_{n, \theta_0} \to 1 - \P(\chi^2_r(a) \le c_r)\eqqcolon \uppi(\tau),
    \end{equation*}
    by Proposition \ref{prop:asymp-dist} where $c_r$ and $a$ are as in Theorem \ref{thm:psi-pwr-local-alt}. Thus, by Theorem \ref{thm:psi-pwr-local-alt},    
    \begin{equation*}
        |P_{n, h_n} \psi_{n, \theta_0} - P_{n, h} \psi_{n, \theta_0}|\le  |P_{n, h_n} \psi_{n, \theta_0} - \uppi(\tau) |+ | P_{n, h} \psi_{n, \theta_0} - \uppi(\tau)| \to 0.
    \end{equation*}

    In the case where $r = 0$, the asymptotic equicontinuity on $K$ of $h\mapsto P_{n, h}(\hat{\Lambda}_{n, \theta_0} = 0)$ implies that if $h_n \to h$ (through $K$), $|P_{n, h_n}(\hat{\Lambda}_{n, \theta_0} = 0) - P_{n, h}(\hat{\Lambda}_{n, \theta_0} = 0) | \to 0$. In combination with $\rank(\hat{\Lambda}_{n, \theta_0}) \xrightarrow{P_{n, h}} 0$ (Assumption \ref{ass:consistent} \ref{ass:consistent:itm:rank} and Remark \ref{rem:mutual-contiguity}), this implies that $P_{n, h_n}(\hat{\Lambda}_{n, \theta_0} = 0) \to 1$ and thus $P_{n, h_n}\psi_{n, \theta_0} \to 0$. Thus, by Theorem \ref{thm:psi-pwr-local-alt}
    \begin{equation*}
        |P_{n, h_n} \psi_{n, \theta_0} - P_{n, h} \psi_{n, \theta_0}|\le  |P_{n, h_n} \psi_{n, \theta_0} |+ | P_{n, h} \psi_{n, \theta_0} | \to 0.\qedhere
    \end{equation*}
\end{proof}

\begin{remark}\label{rem:level-alpha-unif-equicontinuity-weak-condition-split}
    In Lemma \ref{lem:level-alpha-unif-equicontinuity-weak}, Conditions \ref{lem:level-alpha-unif-equicontinuity-weak.itm.Q} and \ref{lem:level-alpha-unif-equicontinuity-weak.itm.r} are required only in the case where $r\ge 1$ whilst Condition \ref{lem:level-alpha-unif-equicontinuity-weak.itm.Lambda} is required only in the case where $r=0$.
\end{remark}

\subsubsection{Uniform results under a measure structure}\label{ssec:unif-measure-structure}

Let $\uppi$ and $\uppi_n$ be as defined in Theorem \ref{thm:psi-pwr-local-alt} and Corollary \ref{cor:psi-locally-uniformly-regular} respectively.

\begin{corollary}\label{cor:psi-is-locally-uniformly-regular-measure}
	Suppose the conditions of Theorem \ref{thm:psi-pwr-local-alt} hold, $(H, \mc{S}, Q)$ is a finite measure space and the functions $h = (\tau, b)\mapsto \uppi_n(\tau, b)$ are measurable. Then, for any $\varepsilon>0$ there is a $K\in \mc{S}$ such that $Q(H\setminus K) < \varepsilon$ and 
	\begin{equation*}
		\limn \sup_{(\tau, b)\in K}|\uppi_n(\tau, b) - \uppi(\tau)| = 0.
	\end{equation*}
\end{corollary}
\begin{proof}
	The pointwise converge is the result of Theorem \ref{thm:psi-pwr-local-alt}.
	$\uppi$ is measurable as the pointwise limit of measurable functions.
	 By Egorov's theorem,
	   $\uppi_n(h) \to \uppi(h)$ uniformly on a $K$ satisfying the given requirements.
\end{proof}

\subsection{Alternative representations of \texorpdfstring{$\effinfo$}{the efficient information matrix}}

\begin{lemma}\label{lem:calculation-of-effinfo}
    Suppose Assumption \ref{ass:LAN} holds, $B$ is a linear space and let $(\Omega, \F, \P)$ be the probability space on which the Gaussian process $\Delta$ of Lemma \ref{lem:IP-GP} is defined.    
    If $\mc{T}\define \{\Delta(h): h = (0, b)\in H \}\subset L_2(\P)$ and $\tilde\Delta(e_i, 0) \define \Pi\left[\Delta(e_i, 0)\middle| \mc{T}^\perp\right]$, then, 
    \begin{equation*}
        \E\left[\Delta(e_i, 0)\tilde\Delta(e_j, 0)\right] =  \E\left[\tilde\Delta(e_i, 0)\tilde\Delta(e_j, 0)\right] = \effinfoarg{ij}.
    \end{equation*}
\end{lemma}
\begin{proof}
    Define $Z:\overline{\mb{H}}\to L_2(\P)$ as $Z[h] = \Delta(h)$. $Z$ is a mean-zero linear Gaussian process with covariance kernel 
    $K([h], [g]) = K(h, g) = \IP{[h]}{[g]}_{K}$.
    $Y\define  \ran Z\subset L_2(\P)$ is a Hilbert space since for $[h], [g]\in \overline{\mb{H}}$,
    $ \E\left[Z[h]Z[g]\right] = K([h], [g]) = \IP{[h]}{[g]}_{K}$, which along with the 
    completeness of $\overline{\mb{H}}$ yields the closedness of $\ran Z$. Hence $Z$ is a Hilbert space isomorphism from $\overline{\mb{H}}$ to $Y$. If $\pi_1' \define \left.\pi\right|_{\mb{H}}$,
    \begin{equation*}
        \mc{T}= \{\Delta(h): h = (0, b)\in H \} =  \{Z[h]: h = (0, b)\in H\} = \{ Z[h]: [h]\in \ker \pi_1'\}.
    \end{equation*}
    We next show that $\mc{T}^\perp = \{Z[h]: [h]\in (\ker \pi_1)^\perp\}$. For the first inclusion suppose that $Z[g]\in \mc{T}^\perp$. Then, for any $[h]\in \ker \pi_1'$, 
    \begin{equation}\label{eq:gh-isomorphism-IP-equality}
        \IP{[g]}{[h]}_{K} = \IP{Z[g]}{Z[h]}_{L_2(\P)} = 0,
    \end{equation}
    and the inclusion follows by taking limits as $\ker \pi_1 = \cl \ker \pi_1'$ by Lemma \ref{lem:closure-pi1-restriction}.
     For the other inclusion note that a corollary of Lemma \ref{lem:closure-pi1-restriction} is that $(\ker \pi_1)^\perp = (\ker \pi_1')^\perp$. Hence, if $[g]\in (\ker\pi_1)^\perp$, for any $[h]\in \ker \pi_1'$ \eqref{eq:gh-isomorphism-IP-equality} holds.  
    Finally, let $Q$ denote the orthogonal projection on $(\ker \pi_1)^\perp \subset \overline{\mb{H}}$ and $R$ that on $\mc{T}^\perp \subset Y$. Then for $[h]\in \overline{\mb{H}}$, $R \Delta(h) = R Z[h] = Z Q Z^* Z[h] =  Z Q Z^{-1} Z[h] = Z Q [h]$,
    since $Z$ is a Hilbert space isomorphism.  Hence $R Z[e_i, 0] = Z Q [e_i, 0]$ implying $   \E\left[\tilde\Delta(e_i,0)\tilde\Delta(e_j,0)\right] = 
     \IP{\Pi^\perp[e_i, 0]}{\Pi^\perp[e_j, 0]}_{K} = \effinfoarg{ij}$.
\end{proof}

In the i.i.d. setting the efficient information matrix $\effinfo$ coincides with the variance matrix of the efficient score function $\effscr = \Pi[\dotscr | \{Db: b\in B\}^\perp]$.

\begin{lemma}\label{lem:effscr-coincides-with-usual-defn}
    If Assumptions \ref{ass:LAN} and \ref{ass:iid} hold and $B$ is a linear space then $\effinfo = \int \effscr\effscr\transp \dP$. 
\end{lemma}
\begin{proof}
    For $h_1,h_2 \in H$, by the i.i.d. assumption and Lemma \ref{lem:IP-GP}, $ P_{n}[\Delta_{n}h_1 \Delta_{n} h_2]= P\left[Ah_1Ah_2\right] = \P[\Delta h_1\Delta h_2]$.
    $X \define \cl \ran A \subset L_2(P)$ and $Y \define \cl \ran \Delta \subset L_2(\P)$ are Hilbert spaces when equipped with the inner products given by $(h_1, h_2)\mapsto P[Ah_1Ah_2]$ and $(h_1, h_2)\mapsto\P[\Delta h_1\Delta h_2]$ respectively.
    Define $U: \ran A \to \ran \Delta$ by $U A h \define \Delta h$ for $h\in H$. $U$ is a bounded, linear, surjective isometry and can be uniquely extended to a Hilbert space isomorphism
      $U:X\to Y$.
    Let $R\define \Pi\left[\cdot| \mc{T}^\perp\right]$ ($\mc{T}^\perp$ as in Lemma \ref{lem:calculation-of-effinfo}) and $Q\define \Pi\left[\cdot | \{D b: b\in B\}^\perp\right]$. Then  $  R\Delta h = RUAh = UQU^* U  Ah  = U Q A h$,
    which implies the conclusion as $e_i\transp \dotscrarg{\gamma} = A(e_i, 0)$.
\end{proof}

\subsection{Most stringent tests}\label{sm:ssec:stringent}

Here I consider most stringent tests; this delivers a similar message to the maximin analysis in the main text.\footnote{The development here is based on Section 9, Chapter 11 in \cite{LC86}; in particular compare Theorem \ref{thm:most-stringent} with Corollary 2 of \citealp[Section 9, Chapter 11]{LC86} which treats the case of a Gausian shift experiment indexed by a Euclidean space.} 
Let $\mc{C}$ be the class of all tests of level $\alpha$ for the hypothesis $\mathrm{K}_0: h\in H_0$ against $\mathrm{K}_1: h\in H_1$ in the experiment $\mathscr{E}$. 
Define $\uppi^\star(h)\define \sup_{\phi\in \mc{C}} P_{h}\phi$ for all $h\in H_{1} $. The \emph{regret} of a test $\phi \in \mc{C}$ is 
\begin{equation}\label{eq:regret}
    R(\phi)\define \sup\left\{\uppi^\star(h) - P_{h}\phi: h\in H_{ 1}\right\}.
\end{equation}
A test $\phi\in \mc{C}$ is called \emph{most stringent at level $\alpha$} if it minimises $R(\phi)$ over $\mc{C}$.

\begin{theorem}\label{thm:most-stringent}
    Suppose Assumptions \ref{ass:LAN} and \ref{ass:IP} hold and $r = \rank(\effinfo)\ge 1$. The most stringent level $\alpha$ test of $\mathrm{K}_0: h\in H_0$ against $\mathrm{K}_1: h\in H_1$ in $\mathscr{E}$
    has power function 
    \begin{equation}\label{thm:most-stringent:eq:pwr-fcn}
        \uppi(h)
         =  1 - \P(\chi_r^2(a) \le c_r), \qquad a = \tau\transp \effinfo \tau,\quad h = (\tau, b)\in H.
    \end{equation}
\end{theorem}
\begin{proof}
    Denote by $\overline{\mathscr{G}}$ the Gaussian shift on $\overline{\mb{H}}$ and $\tilde{R}$ the regret 
    \begin{equation*}
        \tilde{R}(\phi)\define \sup\left\{\tilde{\uppi}^\star([h]) - G_{[h]}\phi: [h]\in \mb{H} \setminus \ker \pi_1\right\}, \quad \tilde{\uppi}^\star([h])\define \sup \left\{ G_{[h]}\varphi: \varphi\in \tilde{\mc{C}}\right\},
    \end{equation*}
    where $\tilde{\mc{C}}$ is the class of level--$\alpha$ tests of $\ker \pi_1$ against $\mb{H} \setminus \ker \pi_1$  in $\ms{G}$.
    The Neyman -- Pearson test, $\psi^\star$, of $[g] \in \ker \pi_1$ against $[g] + [h]$, with $[h]\in [\ker \pi_1]^\perp$ 
     rejects when
    \begin{equation*}
        \exp\left(Z[g + h] - Z [g] - \frac{1}{2}\|[g  +h]\|_{K}^2 + \frac{1}{2}\|[g]\|_{K}^2\right) = \exp\left(Z \Pi^\perp [h]- \frac{1}{2}\| \Pi^\perp[h]\|_{K}^2\right)
    \end{equation*}
    exceeds a $k$ chosen such that the test is of level $\alpha$, for $Z$ the central process of $\overline{\mathscr{G}}$.
    $k$ does not depend on the $[g]\in \ker \pi_1$ and the power of this test depends only on $\|\Pi^\perp [h]\|_{K}^2 = \tau\transp \effinfo \tau$ where $\pi_1 [h] = [\tau]$. Now let $[h]\in \mb{H} \setminus \ker\pi_1$ and consider testing $K_1: [h]$ against $K_0 : [h]\in \ker \pi_1$. One has $[h] = [g] + \Pi^\perp[h]$ where $[g] = \Pi [h] \in \ker \pi_1$. By the preceding observations, $\psi^\star$ is a most powerful level-$\alpha$ test for this hypothesis.\footnote{Suppose there were another level $\alpha$ test $\phi$ of $K_0$ against $K_1$, with strictly higher power than $\psi^\star$. Then, this would also be a test of level $\alpha$ for $[g]$ against $\Pi^\perp [h]$. But this would contradict the Neyman -- Pearson Lemma \citep[e.g.][Theorem 3.2.1]{LR05}.
    } Thus $\psi^\star\in \tilde{\mc{C}}$ and
    \begin{equation}\label{thm:most-stringent:eq:best-test}
        \tilde{\uppi}^\star([h]) \define \sup_{\phi \in \tilde{\mc{C}}} G_{[h]}\phi  = G_{[h]} \psi^\star.    
    \end{equation}
    For $i=1, \ldots, d_\theta$, let $u_i\define \Pi^\perp[(e_i, 0)]$ and let $X \define (Zu_1, \ldots, Zu_{d_\theta})\transp$.
    Let $\psi$ be the test which rejects when $(X\transp \effinfo^\dagger X)^2 > c_r$, 
    for $c_r$ the $1-\alpha$ quantile of the $\chi^2_r$. By Theorem 69.10 in \cite{S85} and Theorem 9.2.3 in \cite{RM71}, $ G_{[h]}\psi = 1 - \P(\chi^2_r(a) \le c_r)$, $a = \tau\transp \effinfo \tau = \|\Pi^\perp [h]\|_{K}^2$ where $[\tau] = \pi_1[h]$.
    $G_{[h]}\psi^\star - G_{[h]}\psi$ depends only on $\|\Pi^\perp [h]\|_{K}^2$.
    Fix a $\varepsilon>0$ and suppose that for some  $\phi\in \tilde{\mc{C}}$, $\tilde{R}(\phi) < \tilde{R}(\psi) - 2\varepsilon$. There is an $a > 0$ such that 
    \begin{equation*}
        \sup\left\{G_{[g + h]}\psi^\star - G_{[g+h]}\psi : [h]\in [\ker \pi_1]^\perp, \|[h]\|_{K}^2 = a\right\} \ge \tilde{R}(\psi) - \varepsilon, \text{ for all } \ [g]\in \ker \pi_1.
    \end{equation*}
    In consequence, for all $[g]\in \ker \pi_1$, all  $[h]\in S_a \define \{[h]\in [\ker\pi_1]^\perp : \|[h]\|_{K}^2 = a\}$,
    \begin{equation*}
        G_{[g + h]}\psi^\star -  G_{[g + h]}\phi \le \tilde{R}(\psi) - 2\varepsilon \le G_{[g + h]}\psi^\star -  G_{[g + h]}\psi - \varepsilon,
    \end{equation*}
    which produces a contradiction to Theorem 30.2 in \cite{S85}:
    \begin{equation*}
         \inf_{[h]\in S_a} G_{[h]} \phi \ge \inf_{[h]\in S_a} G_{[h]}\psi + \varepsilon = 1 - \P(\chi^2_r(a) \le c_r) + \varepsilon. 
    \end{equation*}

    To complete the proof, it suffices to show that $\varphi:\Omega\to [0, 1]$ is in $\mc{C}$ if and only if $\varphi \in \tilde{\mc{C}}$ and $R(\varphi) = \tilde{R}(\varphi)$. The first part follows from $h \in H_{0}$ if and only if $[h]\in \ker \pi_1$ and Proposition \ref{prop:exp-equiv-shift}.
    For the second part, \eqref{thm:most-stringent:eq:best-test}, Proposition \ref{prop:exp-equiv-shift} and the first part together imply that $\tilde{\uppi}^\star([h]) = \uppi^\star(h)$ for all $h\in H$. Therefore, 
    \begin{equation*}
        \tilde{\uppi}^\star([h]) - G_{[h]}\varphi =  \tilde{\uppi}^\star(\pi_V(h)) - G_{\pi_V(h)}\varphi = \uppi^\star(h) - P_{h}\varphi,\quad h\in H.
    \end{equation*}
    As $h\in H_{1}$ $\Longleftrightarrow$ $[h]\in \mb{H} \setminus \ker \pi_1$, $\tilde {R}(\varphi)   = \sup\{\uppi^\star(h) - P_{h}\varphi : h\in H_{1}\} = R(\varphi)$.
\end{proof}

The first part of Corollary \ref{cor:psi-maximin} provides conditions under which \eqref{thm:most-stringent:eq:pwr-fcn} is the asymptotic power function of $\psi_{n, \theta_0}$ under $P_{n, h}$. The following Proposition demonstrates that if $\uppi_n:H \to [0, 1]$ is a sequence of power functions corresponding to tests in the experiments $\mathscr{E}_{n}$ of asymptotic size $\alpha$, then each cluster point of $\uppi_n$ corresponds to a test $\phi$ in the limit experiment $\mathscr{E}$ whose regret is bounded below by that of the most stringent test, $\psi$.

\begin{proposition}\label{prop:most-stringent}
    Suppose Assumptions \ref{ass:LAN} and \ref{ass:IP} hold and that $r = \rank(\effinfo)\ge 1$. Let $\phi_n:\mc{W}_n \to [0, 1]$ be a sequence of tests such that for each $h = (0, b)\in H$,
        \begin{equation}\label{prop:most-stringent:eq-level-alpha}
            \limsupn P_{n,h} \phi_n \le \alpha.
        \end{equation}
    For each $h\in H$, let $\uppi_n(h)\define P_{n, h} \phi_n$. If $\uppi$ is a cluster point of $\uppi_n$ (with respect to the topology of pointwise convergence on $[0, 1]^{H}$), then $\uppi$ is the power function of a test $\phi$ in $\mathscr{E}$ and $R(\phi) \ge R(\psi)$.
\end{proposition}

\begin{proof}
    By \eqref{prop:most-stringent:eq-level-alpha}, Proposition \ref{prop:conv-exp} and Theorem 7.1 in \cite{vdV91b} there is a level $\alpha$ test $\phi$ in $\mathscr{E}$ with $P_{h}\phi = \uppi(h)$. Apply Theorem \ref{thm:most-stringent}.
\end{proof}

\section{Technicalities}\label{sm:sec:technicalities}
\begin{lemma}\label{lem:closure-pi1-restriction}
	Suppose Assumption \ref{ass:LAN} holds and $B$ is a linear space. 
	Let $\pi_1'$ denote the restriction of $\pi_1$ to $\mb{H}$. Then, the closure of $\ker \pi_1'$ in $\overline{\mb{H}}$ is $\ker \pi_1$.
\end{lemma}
\begin{proof}
	Since $\pi_1$ is continuous, $\ker \pi_1 = \pi_1^{-1}(\{0\})$ is closed. Hence it suffices to show that $\ker \pi_1 = \{[h]\in \overline{\mb{H}} : [h] = [0, b]\} \subset \cl \ker \pi_1' = \cl \{[h]\in \mb{H} : [h] = [0, b]\}$.
	Let $[h] = [0, b]\in \ker \pi_1$. There is a sequence $\mb{H} \ni[h_n] = [t_n, b_n]\to [h]$. Using \eqref{lem:ker-effinfo-is-fzero:eq:IP-decomp}
	\begin{equation*}
		\|[h_n] - [h]\|_{K} = t_n\transp \effinfo t_n + \left\| t_{n}\transp \bm{e} + [0, b_n] - [0, b]\right\|_{K}, \quad \bm{e} \define  (\Pi[e_1, 0], \ldots, \Pi[e_{d_\theta}, 0])\transp
	\end{equation*}
	For each $n\in \N$, there are $\check{\bm{e}}_n = ([0, \check{b}_{1, n}], \ldots, [0, \check{b}_{d_\theta, n}])\transp$ with each $[0, \check{b}_{j, n}] \in \mb{H}$ such that $\left\| [0, \check{b}_{j, n}] - \Pi[e_j, 0]\right\|_{K} \le 1 / (n|t_{n, j}|)$. Putting $[0, \tilde{b}_{n}] \define t_n\transp \check{\bm{e}}_n + [0, b_n]$,
	\begin{equation*}
		\|[0, \tilde{b}_{n}] - [0, b]\|_{K} \le \|t_n\transp \check{\bm{e}}_n  - t_n\transp\bm{e} \|_{K} + \| t_n\transp\bm{e} + [0, b_n]  - [0, b]\|_{K} \le \frac{d_\theta}{n} + o(1) = o(1).
	\end{equation*}
	Since each $[0, \tilde{b}_{n}]\in \ker \pi_1'$, the limit $[0, b] \in \cl \ker \pi_1'$. 
\end{proof}

\begin{lemma}\label{lem:seq-w-LR-to-0-no-pwr-diff}
	Suppose that for $h_n, g\in H$,	$ P_{n, g} \lcontig P_{n}$ and $L_{n}(h_n) - L_{n}(g) = o_{P_{n}}(1)$. Then $d_{TV}(P_{n, h_n},P_{n, g})\to 0$.
\end{lemma}
\begin{proof}
	By the contiguity $L_{n}(h_n) - L_{n}(g) = o_{P_{n, g}}(1)$. Apply Lemma S3.3 in \cite{HLM22-S}.
\end{proof}

\begin{corollary}\label{cor:h-eqiv-g-no-pwr-diff}
	Suppose that Assumption \ref{ass:LAN} holds and $H$ is a linear space equipped with the semi-norm $\|\cdot\|_{K}$. If $h, g\in H$ satisfy $\|h-g\|_{K} = 0$, $d_{TV}(P_{n, h}, P_{n, g}) \to 0$.	
\end{corollary}
\begin{proof}
	By Assumption \ref{ass:LAN}, the reverse triangle inequality and $\sigma(h-g) = \|h-g\|_{K}$ we have that $L_{n}(h) - L_{n}(g) = o_{P_{n}}(1)$. Apply Lemma \ref{lem:seq-w-LR-to-0-no-pwr-diff} with $h_n = h$.
\end{proof}

\begin{lemma}\label{lema:cond-exp-orthog-set-equality}
	Let $(U, X)$ be a random vector on a probability space $(\Omega, \F, \P)$ with $U\in L_2(\P)$ and $\E[UU\transp|X]$ non -- singular almost surely. Let $B\subset L_2(\Omega, \sigma(U, X), \P)$ be the set of bounded functions $b$ of $(u, x)$ such that $\E[b(U, X)U|X] = 0$. Then 
	\begin{equation*}
		\cl B = \{U Z : Z \text{ is a bounded, } \sigma(X) \text{--measurable random variable}\}^\perp.
	\end{equation*}
\end{lemma}
\begin{proof}
	Suppose $b\in B$. Then $\E[b(U, X)UZ] = \E\left[\E[b(U, X)U|X]Z\right] = 0$. Conversely suppose $b\in L_2(\Omega, \sigma(U, X), \P)$ and $\E[b(U, X)UZ] = 0$ for $Z$ any bounded $\sigma(X)$ -- measurable random variable. By 
	Proposition A.3.1 in \cite{BKRW98}, $\E[b(U, X)U|X] = 0$ a.s. whence $b\in \cl B$ by Lemma C.7 in \cite{N91} 
\end{proof}

\begin{theorem}\label{thm:projection-convergence}
	Let $H$ be a Hilbert space. Let $h_n, h\in H$, and $L_n, L$ closed (proper) linear subspaces of $H$. Let $g_n\define \Pi(h_n | L_n)$ and $g\define \Pi(h|L)$. If (i) $h_n \to h$ and (ii) for each $f\in L$, there is a sequence $(f_n)_{n\in \N}$ and a $N\in \N$ such that $f_n \to f$ and $f_n\in L_n$ for $n\ge N$, then $g_n\to g$.
\end{theorem}
\begin{proof}
    Let $\Pi_n$ be the orthogonal projection onto $L_n$ and $\Pi$ that onto $L$. First suppose $h_n=h$ ($n\in \N$).
	As $(g_n)_{n\in \N}$ is bounded, any subsequence contains a weakly convergent subsequence, say $g_{n_k} \rightharpoonup g^\star$. By self-adjointness and idempotency (SAI)
    \begin{equation}\label{thm:projection-convergence:eq:step1}
        \IP{g_{n_k}}{g_{n_k}} = \IP{\Pi_{n_k} h}{\Pi_{n_k} h} = \IP{h}{\Pi_{n_k} h} \to \IP{h}{g^\star}.
    \end{equation}
    Let $f\in L$. By hypothesis there are $(f_n)_{n\in \N}$ with $f_n\to f$ and $f_n\in L_n$ for $n\ge N_1$. So $f_{n_k}\to f$ and $f_{n_k}\in L_{n_k}$ for  $k\ge K_1$. Since $h - \Pi_{n_k}h \rightharpoonup h - g^\star$, by Proposition 16.7 in \cite{RF10} and the fact that $h - g_{n_k}\in L_{n_k}^\perp$ for each $k$, $\IP{h - g^\star}{f} = \lim_{k\to\infty} \IP{h - g_{n_k}}{f_{n_k}} = 0$.
    Hence $g^\star = \Pi h = g$. By SAI of $\Pi$ and \eqref{thm:projection-convergence:eq:step1}, $ \lim_{k\to\infty}\IP{g_{n_k}}{g_{n_k}} = \IP{h}{\Pi h} = \IP{\Pi h}{\Pi h} = \IP{g}{g}$
    and hence $g_{n_k}\to g$ by the Radon -- Riesz Theorem. 
	As the initial subsequence was arbitrary, $g_n\to g$. To complete the proof, for $h_n\to h$ an arbitrary convergent sequence,
	$\|g_n - g\| \le \|h_n - h\| + \|\Pi_n h - \Pi h\|$.
    The first RHS term is $o(1)$ by assumption; the second by the case with $h_n = h$.    
\end{proof}

\begin{theorem}\label{sm:thm:cond-exp-weak-conv}
    Let 
	$H$ be a linear space and $B\subset H$ a linear subspace of $H$. 
    Suppose that $\mathsf{G}_n$ is a Gaussian process on a probability space $(\Omega, \F, \bbP)$ with index set $H$ and covariance kernel $K_n$ and that $\mathsf{G}$ is a  Gaussian process on $(\Omega, \F, \bbP)$ with index set $H$ and covariance kernel $K$. Suppose that $K_n(h, g)\to K(h, g)$, $h, g\in H$.
	Let $H$ be equipped with the positive semi - definite, symmetric bilinear form defined as $\IP{h}{g}\define K(h, g)$ and suppose that $H$ is separable under the induced pseudometric.
	Fix $h, g\in H$ and let $\ms{G}_{n}\define \sigma(\{\mathsf{G}_nf: f\in B\})$ and $\ms{G}_{n}\define \sigma(\{\mathsf{G}f: f\in B\})$. Then,
    \begin{equation*}
        X_n \define  \left(
            \mathsf{G}_n h,
            \E[\mathsf{G}_n g | \ms{G}_n]
        \right) \weakconv X\define \left(
            \mathsf{G} h,
            \E[\mathsf{G} g | \ms{G}]
        \right).
    \end{equation*}
\end{theorem}
\begin{proof}
	$\cl \{\mathsf{G} b: b\in B\}$ is a separable Hilbert space and so has an orthonormal basis, which may be taken to be 
	formed of $\mathsf{G}b_j$, $(b_j)_{j\in \N}\subset B$.
	 Let $\mc{G}_n\define \sigma(\{\mathsf{G}_n b_i: i\in \N\})$, $\mc{G}\define \sigma(\{\mathsf{G} b_i: i\in \N\})$, $B_m\define (b_1, \ldots, b_m)$, $\mc{G}_{n}^m\define \sigma(\{\mathsf{G}_n b: b\in B_m\})$, $\mc{G}^m\define \sigma(\{\mathsf{G}b: b\in B_m\})$,  $X_n^m \define  \left(
		\mathsf{G}_n h,
		\E[\mathsf{G}_n g | \mc{G}_n^m]
	\right)$ and $ X^m\define \left(
		\mathsf{G} h,
		\E[\mathsf{G} g | \mc{G}^m]
	\right)$.
	Now let $Z_n^m\define\left(\mathsf{G}_nh, \mathsf{G}_n g, \mathsf{G}_n b_1, \ldots, \mathsf{G}_n b_m\right)\transp \sim \mc{N}(0, \Sigma_n^m)$ and $Z^m\define \left(\mathsf{G}_nh, \mathsf{G} g, \mathsf{G}_n b_1, \ldots, \mathsf{G}_n b_{m}\right)\transp\sim \mc{N}(0, \Sigma^m)$. 
   Partition $\Sigma^m$ conformally with $Z^m_1 = \mathsf{G}h$, $Z_2^m = \mathsf{G}g$ and $Z_{3}^m = \left(\mathsf{G}_n b_1, \ldots, \mathsf{G}_n b_m\right)\transp$ 
	and similarly for $\Sigma_{n}^m$, $Z_{n}^m$. Then 
    \begin{align*}
        X_n^m &= \left(\mathsf{G}_n h,\, \E[\mathsf{G}_n g | \mc{G}_n^m]\right) = \left(Z_{n, 1}^m,\, Z_{n, 2}^m - [\Sigma_n^m]_{2, 3}[\Sigma_n^m]_{3, 3}^{-1} Z_{n, 3}^m\right)\\
		X^m &= \left(\mathsf{G} h,\, \E[\mathsf{G} g | \mc{G}^m]\right) = \left(Z_{1}^m,\, Z_2^m - [\Sigma^m]_{2, 3}[\Sigma^m]_{3, 3}^{-1} Z_{3}^m\right).
    \end{align*}
	Since $K_n(h_1, h_2)\to K(h_1, h_2)$ for all $h_1, h_2\in H$, $\Sigma_n^m\to \Sigma^m$ as $n\to\infty$ and the inverses in the preceding displays exist for all sufficiently large $n$ since $\{\mathsf{G}b_i: i\in \N\}$ is orthonormal. By $\Sigma_n^m\to \Sigma^n$, Levy's continuity Theorem and the Cram\'e{r} -- Wold Theorem, $Z_n^m\weakconv Z^m$. Hence, 
    \begin{equation}\label{eq:B32-i}
		X_n^m \weakconv X^m.
    \end{equation}
    Let $\Pi^m$ be the orthogonal projection onto $S_m\define \Span \{\mathsf{G}b: b\in B_m\}$. Then, $X^m = \E[\mathsf{G} g | \mc{G}^m]  = \E[\mathsf{G} g | \ms{G}^m] = \Pi^m \mathsf{G} g$, 
	by Theorem 9.1 in \cite{J97}.
    $S_m \subset S_{m+1}$ and $S \define \cl \{\mathsf{G}b: b\in B\} = \cl \cup_{m\in \N} S_m$, so by Theorem \ref{thm:projection-convergence} and 
    Theorem 9.1 in \cite{J97},  $\|\E[\mathsf{G} g | \ms{G}^m] - \E[\mathsf{G}g | \ms{G}]\|_{L_2} = \|\Pi^m \mathsf{G} g - \Pi \mathsf{G}g\|_{L_2} \to 0$ and so 
    \begin{equation}\label{eq:B32-ii}
        X^m\weakconv X.
    \end{equation}
    Define $Y_{n} \define \E[\mathsf{G}_nh | \mc{G}_{n}]$, $Y_{n}^m\define  \E[\mathsf{G}_nh | \mc{G}_{n}^m] $, $Y^m\define  \E[\mathsf{G}h | \mc{G}^m]$ and $Y\define \E[\mathsf{G}h | \mc{G}]$.
    As $Y_n \in \cl \{\mathsf{G}_nb : b\in B\}$ and $Y_{n}^m\in  \{\mathsf{G}_nb : b\in B_m\}$ \cite[][Theorem 9.1]{J97}, $  Y_n - Y_{n}^m \sim \mc{N}(0, \sigma_{n,m}^2)$ where $\sigma_{n, m}^2 \define \Var(Y_n - Y_n^m)$.
    As $Y_n = \E[\mathsf{G}_nh | \ms{G}_n]$ and $Y_n^m = \E[\mathsf{G}_nh | \ms{G}_n^m]$ \cite[][Theorem 9.1]{J97},
    \begin{equation}\label{eq:exponential-bound-Gaussian}
		\bbP\left(\|X_n - X_n^m\| >\varepsilon\right) = \bbP\left(|Y_n - Y_n^m| >\varepsilon\right)\le C \exp\left(-\frac{\varepsilon^2}{\sigma_{n, m}^2}\right).
    \end{equation}    
    We show next that $\sigma_{n, m}^2\to \sigma_m^2\define \Var(Y - Y^m)$. For this let $f_0 \define h$, $f_i\define b_i$, $i\in \N$. Consider the restricted processes $F_{n}\define (F_{n,i})_{i\in \N}$ and $F \define (F_{i})_{i\in \N}$ where $F_{n, i}\define \mathsf{G}_n f_{i-1}$ and $F_i\define \mathsf{G}f_{i-1}$. $F_n$ and $F$ are random elements in the separable metric space $(\R^\infty, d)$ where $d$ is the metric given in Example 1.2 of \cite{B99}.
    Hence $F_{n}\weakconv F$ in $(\R^\infty, d)$ by Example 2.4 of \cite{B99}. By 
	Skorohod's representation Theorem \cite[e.g.][Theorem 6.7]{B99} there are random elements $\tilde{F}_n$ and $\tilde{F}$ defined on a common probability space such that $\tilde{F}_n \to \tilde{F}$ surely, $\mc{L}(\tilde{F}) = \mc{L}(F)$ and $\mc{L}(\tilde{F}_n) = \mc{L}(F_n)$. 
    Thus 
	$\tilde{F}_n$ and $\tilde{F}$ are 
	Gaussian processes. As $\Cov(\tilde{F}_{n, i},\tilde{F}_{n, j} )= K_n(f_{i-1}, f_{j-1})\to K(f_{i-1}, f_{j-1}) = \Cov(\tilde{F}_i, \tilde{F}_j)$ each $(\tilde{F}_{n,i})_{n\in \N}$ is uniformly square integrable. As $(\R^\infty, d)$ has the topology of pointwise convergence 
	each $\tilde{F}_{n, i}\to \tilde{F}_i$ surely. Hence $\tilde{F}_{n, i}\xrightarrow{L_2}\tilde{F}_i$. By the equality in law
    \begin{align*}
       \tilde{Y}_n^m&\define \E[\tilde{F}_{n, 1} | \{\tilde{F}_{n, i}: 2\le i \le m\} ] \sim Y_{n}^m, \quad \tilde{Y}_n\define \E[\tilde{F}_{n, 1} | \{\tilde{F}_{n, i}: i\in \N, i \neq 1\} ] \sim Y_{n},\\
	   \tilde{Y}^m&\define \E[\tilde{F}_{1} |\{\tilde{F}_{i}: 2\le i \le m\} ] \sim Y^m, \qquad\  \tilde{Y}\define \E[\tilde{F}_{1} | \{\tilde{F}_{i}: i\in \N, i\neq 1\} ] \sim Y.
    \end{align*}
    Define $\tilde{S}_{n}^m\define \Span\{\tilde{F}_{n, i}: 2\le i\le m\}$, $\tilde{S}_{n}\define \cl\Span\{\tilde{F}_{n, i}: i\in \N, i\neq 1\}$, $\tilde{S}^m\define \Span\{\tilde{F}_{i}: 2\le i\le m\}$ and $\tilde{S}\define \cl\Span\{\tilde{F}_i: i\in \N, i\neq 1\}$. Then $\tilde{Y}_n^m = \Pi[\tilde{F}_{n, 1} | S_n^m], \, \tilde{Y}_n = \Pi[\tilde{F}_{n, 1} | S_n], \, \tilde{Y}^m = \Pi[\tilde{F}_{1} | S^m],\, \tilde{Y} = \Pi[\tilde{F}_{1} | S]$
	by Theorem 9.1 in \cite{J97}. 
    We will apply Theorem \ref{thm:projection-convergence} twice (in $L_2$). It is straightforward to check the hypotheses are satisfied with (i) $L_n\define \tilde{S}_n^m$, $L\define \tilde{S}^m$; (ii) $L_n \define \tilde{S}_n$, $L\define \tilde{S}$ and $h_n\define \tilde{F}_{n, 1}$, $h\define \tilde{F}_1$ in both cases.
	By Theorem \ref{thm:projection-convergence},
    \begin{equation*}
        \|\tilde{Y}_n - \tilde{Y}_n^m - (\tilde{Y} - \tilde{Y}^m)\|_{L_2} \le \|\tilde{Y}_n - \tilde{Y}\|_{L_2} + \|\tilde{Y}_n^m -\tilde{Y}^m\|_{L_2} \to 0,
    \end{equation*}
    hence $\sigma_{n, m}^2 = \Var(Y_n - Y_n^m) = \Var(\tilde{Y}_n - \tilde{Y}_n^m) \to \Var(\tilde{Y} - \tilde{Y}^m) = \Var(Y - Y^m) = \sigma_m^2$. 

    To see that $\lim_{m\to\infty}\sigma_m^2 = 0$
	 set $L_m \define\Span\{\mathsf{G}b: b\in B_m\}$ and $L\define \cl \{\mathsf{G}b: b\in B\}$. It is easy to check the hypotheses of Theorem \ref{thm:projection-convergence} (with $m$ in place of $n$) hold.
	Hence $Y^m \xrightarrow{L_2}Y$ and so $\sigma_m^2 = \Var(Y - Y^m)\to 0$. In conjunction with \eqref{eq:exponential-bound-Gaussian},
    \begin{equation}\label{eq:B32-iii}
        \limm \limsupn  \bbP\left(\|X_n - X_n^m\| >\varepsilon\right)\le   \limm \limsupn  C \exp\left(-\frac{\varepsilon^2}{\sigma_{n, m}^2}\right) = 0.
    \end{equation}
    The result now follows from Theorem 3.2 in \cite{B99}.
\end{proof}

\begin{lemma}\label{lem:SIM-markov-conditional}
    Let $(m_n)_{n\in \N}$ be an increasing sequence of natural numbers such that $m_n \le n$, $(Y_{n, i})_{n\in \N, 1\le i\le m_n}$  a triangular array of random vectors and $\mc{C}_n$ a collection of random variables. Suppose that with probability approaching one either
    \begin{enumerate}
        \item $\E\left[\left\|Y_{n, i}\right\|\middle|\mc{C}_n\right]\le \delta_n n^{-1/2}$ for some $\delta_n \to 0$ and all $i\le m_n$; or \label{lem:SIM-markov-conditional.itm.1}
        \item For each component $Y_{n, i, s}$ of $Y_{n, i}$ and any $i\neq j\le m_n$, $\E[Y_{n, i, s}Y_{n, j, s}|\mc{C}_n] = 0$ almost surely and $\E[Y_{n, i, s}^2|\mc{C}_n] \le \delta_n$ for some $\delta_n\to 0$ and all  $i\le m_n$.\label{lem:SIM-markov-conditional.itm.2}
    \end{enumerate}
    Then $\frac{1}{\sqrt{m_n}} \sum_{i=1}^{m_n} Y_{n, i}$ converges to zero in probability.
\end{lemma}
\begin{proof}
   If condition \ref{lem:SIM-markov-conditional.itm.1} holds, $\E \left\|  m_n^{-1/2} \sum_{i=1}^{m_n} Y_{n, i} \right\| 
       \le \delta_n m_n^{1/2} n^{-1/2} \to 0$. If condition  \ref{lem:SIM-markov-conditional.itm.2} holds, $\E\left(m_n^{-1/2}\sum_{i=1}^{m_n} Y_{n, i, s} \right)^2 = m_n^{-1}\sum_{i=1}^{m_n} \E Y_{n, i, s}^2 \le \delta_n \to 0$ for each component $Y_{n, i, s}$ of $Y_{n, i}$.  In either case the claim then follows by Markov's inequality.
\end{proof}

\section{Additional details for the examples}\label{sm:sec:examples-detail}

\subsection{Single index model}\label{ssec:example-details-sim}

\subsubsection{Proofs of results in the main text}

\begin{proof}[Proof of Proposition \ref{prop:sim-joint-conv}]
    As is easy to verify, each component of $g_{n}$ belongs to $L_2^0(P_n)$.
    For any $b\in B$,  $\E[\epsilon b_2(\epsilon, X)|X] =0$ by \eqref{eq:sim-b2-conditions}. Plugging in for $Db$ and using this allows the conclusion that $ \E\left[g(W) [Db](W)\right] = 0$. Apply Lemma \ref{lem:iid-orthocomp-joint-conv}.
\end{proof}

\begin{proof}[Proof of Proposition \ref{prop:sim-estimation-HL}]
    For part \ref{ass:consistent:itm:ghat} of Assumption \ref{ass:consistent} note that for some $a_j \in \{-1, 1\}$, $ \fracrootn\sumin \hat{g}_{n, \theta, i} - g(W_i) = \sum_{l=1}^5 a_l \fracrootn  \left[ \sum_{i=1}^{m_n} R_{l, n, i} + \sum_{i=m_n+1}^n R_{l, n ,i}\right]$
    where
    \begin{align*}
        R_{1, n, i} &\define \upomega(X_i) (\hat{f}_{n, i}(V_{\theta, i}) - f(V_{\theta, i}))f'(V_{\theta, i})(X_{2, i} - Z_{0}(V_{\theta, i}))\\
        R_{2, n, i} &\define \upomega(X_i) (Y_i - f(V_{\theta, i}))\left(f'(V_{\theta, i}) - \widehat{f'}_{n, i}(V_{\theta, i}) \right)(X_{2, i} - Z_{0}(V_{\theta, i}))\\
        R_{3, n, i} &\define \upomega(X_i) (Y_i - f(V_{\theta, i}))\widehat{f'}_{n, i}(V_{\theta, i}) \left(\hat{Z}_{0, n, i}(V_{\theta, i}) - Z_{0}(V_{\theta, i}) \right)\\
        R_{4, n, i} &\define \upomega(X_i) (\hat{f}_{n, i}(V_{\theta, i}) - f(V_{\theta, i}))\left(f'(V_{\theta, i}) - \widehat{f'}_{n, i}(V_{\theta, i}) \right)(X_{2, i} - Z_{0}(V_{\theta, i}))\\
        R_{5,n, i} &\define \upomega(X_i) (\hat{f}_{n, i}(V_{\theta, i}) - f(V_{\theta, i}))\widehat{f'}_{n, i}(V_{\theta, i}) \left(\hat{Z}_{0, n, i}(V_{\theta,i}) - Z_{0}(V_{\theta, i}) \right).
    \end{align*}
    We verify that one of \ref{lem:SIM-markov-conditional.itm.1} or \ref{lem:SIM-markov-conditional.itm.2} of Lemma \ref{lem:SIM-markov-conditional} is satisfied with $Y_{n, i} = R_{l, n, i}$ for $i=1, \ldots, m_n\coloneqq \lfloor n/2\rfloor$ or $i=m_n + 1, \ldots, n$.
    Suppose that $1 \le i \le m_n$ and let $\mc{C}_n = \mc{C}_{n, 2}$ (the case with $m_n+1\le i\le n$ and $\mc{C}_n = \mc{C}_{n, 1}$ is analogous).  Each $\hat{Z}_{k, n, i}(V_{\theta, i})$ is $\sigma(V_{\theta, i}, \mc{C}_n)$ -- measurable for $k=0, 1, 2, 3, 4$. $f(V_{\theta, i})$, $f'(V_{\theta, i})$, $\upomega(X_i)$ and $X_{2, i} - Z_0(V_{\theta, i})$ are bounded uniformly in $i$ and there are events $E_n$ with $P_{n, 0}E_n\to 1$ on which $\mathsf{R}_{l, n ,i} \le r_n$ ,  $\hat{f}_{n, i}(V_{\theta, i})$, $\widehat{f'}_{n, i}(V_{\theta, i})$, $\hat{Z}_{1, n, i}(V_{\theta, i})$ are bounded above uniformly in $i$ and $\hat{Z}_{2, n, i}(V_{\theta, i})$ is bounded above and below uniformly in $i$, for all large enough $n\in \N$.
    On these sets, 
    \begin{equation}\label{prop:sim-estimation-HL:eq-Z0}
        \E\left[\left\|  \hat{Z}_{0, n, i}(V_{\theta, i}) - Z_{0}(V_{\theta, i}) \right\|^2 \middle |\mc{C}_n \right]\lesssim r_n^2.
    \end{equation}

    For $l=1, 2, 3$, the first part of condition \ref{lem:SIM-markov-conditional.itm.2} follows by the law of iterated expectations and independence since $\E[\upomega(X_i)(X_{2, i} - Z_0(V_{\theta, i}))|V_{\theta, i}] = 0$ ($l=1$) and $\E[\epsilon_i|X_i] = 0$ ($l=2, 3$). The second part follows by the uniform boundedness noted above, $R_{l, n ,i} \le r_n$ on $E_n$ along with equations \eqref{eq:sim:eps-phi--1} and \eqref{prop:sim-estimation-HL:eq-Z0}.

    \emph{$l=4$}: By the uniform boundedness and the Cauchy -- Schwarz inequality,
    $ \E[\|R_{4, n, i}\||\mc{C}_n] \lesssim \E\left[\left|\hat{f}_{n, i}(V_{\theta, i}) - f(V_{\theta, i})\right|\left|f'(V_{\theta, i}) - \widehat{f'}_{n, i}(V_{\theta, i}) \right|\middle|\mc{C}_n\right]$
    and the RHS is upper bounded by $\mathsf{R}_{3, n, i}\mathsf{R}_{4, n, i} = o(n^{-1/2})$ on $E_n$.

    \emph{$l=5$}: By the uniform boundedness and the Cauchy -- Schwarz inequality,
    $\E[\|R_{4, n, i}\||\mc{C}_n] \lesssim \E\left[\left|\hat{f}_{n, i}(V_{\theta, i}) - f(V_{\theta, i})\right|\left\|\hat{Z}_{0, n, i}(V_{\theta,i}) - Z_{0}(V_{\theta, i})  \right\|\middle|\mc{C}_n\right]$.
    For a $C>0$, the RHS is upper bounded by $Cr_n \mathsf{R}_{3, n, i} = o(n^{-1/2})$ on $E_n$ by \eqref{prop:sim-estimation-HL:eq-Z0}.
    
    For parts \ref{ass:consistent:itm:lambdahat} and \ref{ass:consistent:itm:rank} of Asssumption \ref{ass:consistent}, we show that $\|\check{V}_{n, \theta} - V\| = o_{P_{n, 0}}(\upnu_n)$, which suffices by Proposition S1 of \cite{LM21-S}.
    For $\breve{V} \define \EP gg\transp$, 
    \begin{equation*}
        \check{V}_{n, \theta} - V=  \check{V}_{n, \theta} - \breve{V}+ \breve{V}- V= \meanin \left[\hat{g}_{n, \theta, i}\hat{g}_{n, \theta, i}\transp - g(W_i)g(W_i)\transp\right] + \fracrootn \G_n [g g\transp].
    \end{equation*}
    $\E (g_{l}g_{k})^2<\infty$ by $\E[\epsilon^4]<\infty$ and Assumption \ref{ass:sim-estimation-HL}.
    Hence $ \fracrootn \G_n [g g\transp] = O_{P_{n, 0}}(n^{-1/2})$ by the CLT. For the other term, 
    \begin{equation*}
        \meanin \left(\hat{g}_{n, \theta, i, k} - g_{k}(W_i)\right)^2 \lesssim \sum_{l=1}^5 \fracn \left[ \sum_{i=1}^{m_n} R_{l, n, i, k}^2 + \sum_{i=m_n+1}^n R_{l, n ,i, k}^2\right].
    \end{equation*}
    For $l=1, 2, 3$ we established that if $1 \le i \le m_n$ and $\mc{C}_n = \mc{C}_{n, 2}$ then $\E[R_{l, n, i, k}^2|\mc{C}_n]\lesssim r_n^2$ on $E_n$. We show this also holds for $l=4, 5$. (The case with $m_n+1\le i\le n$ with $\mc{C}_n = \mc{C}_{n,1}$ is once again analogous). For $l\in \{4, 5\}$, $ \E\left[R_{l, n, i, k}^2|\mc{C}_n\right] \lesssim \E\left[\left(\hat{f}_{n, i}(V_{\theta, i}) - f(V_{\theta, i})\right)^2\middle |\mc{C}_n\right]$,
    by the uniform boundedness (for all large enough $n$) and the RHS is bounded above by $r_n^2$ on $E_n$.
    By Markov's inequality, \linebreak $\fracn \left[ \sum_{i=1}^{m_n} R_{l, n, i, k}^2 + \sum_{i=m_n+1}^n R_{l, n ,i, k}^2\right] = O_{P^n}(r_n^2)$ for $l=1, \ldots, 5$  hence the same is true of $\meanin \left\|\hat{g}_{n, \theta, i,  } - g(W_i)\right\|^2$. Therefore, $\|\check{V}_{n, \theta} - \breve{V}\|_2 = O_{P^n}(r_n)$ as $\left\| \check{V}_{n, \theta} - \breve{V}\right\|_2^2$ is bounded above by a multiple of
    \begin{equation*}
        \meanin \left\|\hat{g}_{n, \theta, i}\right\|^2\meanin \left\|\hat{g}_{n, \theta, i}- g(W_i)\right\|^2+\meanin \left\|\hat{g}_{n, \theta, i} - g(W_i) \right\|^2\meanin \left\|g(W_i)\right\|^2.	\qedhere
	\end{equation*}
\end{proof}

\subsubsection{The LAN condition}\label{app:ssec:sim-LAN}

Here I provide examples of local perturbations $P_{n, h}$ and lower level conditions under which the LAN condition in Assumption \ref{ass:sim-LAN} holds. Let $\varphi_{n}$ be as in equation \eqref{eq:SIM-local-alt}
with $B_1\define C_b^1(\ms{D})$
and $B_2$ taken to be the set of functions $b_2:\R^{1+K}\to \R$ such that $b_2$ is bounded, $e\mapsto b_2(e, x)$ is continuously differentiable with bounded derivative and equation \eqref{eq:sim-b2-conditions} holds.

\begin{proposition}\label{prop:sim-DQM-LAN}
    Suppose Assumption \ref{ass:sim-mdl} holds, $\mc{W}_n = \prod_{i=1}^n \R^{1+K}$, $e\mapsto \sqrt{\zeta(e, x)} \in \mc{C}^1$,    
    and $p_{n, h} = p_{\gamma + \varphi_n(h)}^n$ with $p_{\gamma}$ as in \eqref{eq:SIM-running-example-dens}. Then Assumption \ref{ass:sim-LAN} holds.
\end{proposition}
\begin{proof}
    Define $\gamma_t(h)\define \gamma + t(\tau, b_1, b_2\zeta)$ for $h = (\tau, b_1, b_2)$ and $t\in [0, \infty)$. It is easy to verify that $P_{\gamma_t(h)}\in \{P_{\gamma}: \gamma\in \Gamma\}$ for all small enough $t$. This ensures the required domination in Assumption \ref{ass:iid} given Assumption \ref{ass:sim-mdl}.
    Next note that $t\mapsto\sqrt{p_{\gamma_t(h)}}$ is continuously differentiable everywhere since
    it is a composition of continuously differentiable functions for $t$ small enough that $(1+ tb_2)$ is bounded away from zero. This ensures that $q_t(W) \define \deriv{\log p_{\gamma_s(h)(W)}}{s}\rvert_{s=t}$ is defined for small enough $t$. Writing $v_{t}\define V_{\theta + t\tau}$ and $e_t\define Y - f(v_t) - tb_1(v_t)$ this has the form
    \begin{equation}\label{prop:SIM-ULAN:eq:log-deriv}
        \begin{aligned}
            q_t(W)\define &-\phi(e_t, X)[f'(v_t)X_2\transp\tau  + tb_1'(v_t)X_2\transp\tau  + b_1(v_t)]\\
            &\qquad  + \frac{b_2(e_t, X) - tb_2'(e_t, X)[f'(v_t)X_2\transp\tau  + tb_1'(v_t)X_2\transp\tau  + b_1(v_t)]}{1 + tb_2(e_t, X)},
        \end{aligned}
    \end{equation}
    which is a composition of continuous functions. By boundedness of $f'$, $b_1$, $b_1'$, $b_2$, $b_2'$, $(1 + tb_2)^{-1}$ and equation \eqref{ass:sim-parameters:eq:sim-moments},
    $\int |q_t(W)|^{2+\rho} \darg{P_{\gamma_t(h)}}  \le C\E\left[\left(|\phi(\epsilon, X)|^{2+\rho} + 1\right)\|X\|^{2+\rho}\right]<\infty$ for a positive constant $C<\infty$ and a $\rho>0$.
    This implies that for any $t_n\to t$, $(q_{t_n}(W)^2)_{n\in \N}$ is uniformly $P_{\gamma_{t_n}(h)}$ -- integrable. Combination with $q_{t_n}(W)^2 \to q_t(W)^2$ (everywhere) yields $ \int q_{t_n}(W)^2 p_{\gamma_{t_n}(h)}(W)\dlambda \to  \int q_{t}(W)^2 p_{\gamma_{t}(h)}(W)\dlambda$.
    Applying Lemma 1.8 in \cite{vdV02} demonstrates that equation \eqref{eq:dqm} holds, with $Ah$ as in \eqref{eq:Agamh}. Lemma 1.7 of \cite{vdV02} ensures that $A h\in L_2^0(P)$. The form of $Ah$ reveals that it is a linear map on $H$. It is bounded:
    \begin{equation*}
        \|Ah\|^2\le C_1  \E\left[\phi(\epsilon, X)^2 \|X\|^2\right] \|\tau\|^2 + \E\left[\phi(\epsilon, X)^2 \right] \|b_1\|^2 + \|b_2\|^2 \le C_2 \|h\|^2,
    \end{equation*}
    where $C_1,C_2 \in (0, \infty)$ are positive constants.
    Apply Lemma \ref{lem:iid-LAN-DQM}.
\end{proof}

\subsection{IV model with non-parametric first stage}\label{sm:sec:IV}

\subsubsection{Proofs of results in the main text}
\begin{proof}[Proof of Lemma \ref{lem:iv-effscr}]
    $J(Z)$ is nonsingular by \eqref{eq:iv-EphiU}. By Proposition 2.8.4 in \cite{B09}, $J(Z)_{1, 1}^{-1} = \E[\epsilon^2|Z]^{-1}$ exists and is positive.
    Define $\dot{l}(W) \define  \left(\dot{l}_{1}(W)\transp, \dot{l}_{2}(W)\transp\right)\transp = -\phi_1(\upxi)(X\transp, Z_1\transp)\transp$, $[D_{ 1} b_1](W)\define -\phi_2(\upxi)\transp b_1(Z) $ and $[D_{2} b_2](W) \define b_2(\upxi)$, 
    where $\upxi = (Y - X\transp\theta - Z_1\transp\beta, X - \pi(Z),Z)$.
    By Lemma \ref{lema:cond-exp-orthog-set-equality} and Proposition A.3.5 in \cite{BKRW98}, with $\mc{T}_2\define \{[D_{2} b_2](W): b_2
   \in B_2\}$, 
    \begin{align*}
        \breve{l}_{\gamma}(W) &\define \Pi\left[ \dot{l}(W) |\mc{T}_2^\perp \right]= \E\left[-\phi_1(\xi)[X\transp, Z_1\transp]\transp U\transp |Z\right]\E\left[UU\transp |Z\right]^{-1}U,\\
        [\breve{D}_{1}b_1](W) &\define \Pi\left[ [\dot{D}_1b_1](W) |\mc{T}_2^\perp \right]= \E\left[-b_1(Z)\transp \phi_2(\xi) U\transp |Z\right]\E\left[UU\transp |Z\right]^{-1}U.
    \end{align*}
    Let $K \define d_\beta$, and evaluating the conditional expectations using \eqref{eq:iv-EphiU} we obtain:
    \begin{equation*}
        \begin{aligned}
            \breve{l}_{\gamma}(W) &= \begin{bmatrix}
                \pi(Z)\\
                Z_1
            \end{bmatrix}
                 \begin{bmatrix}
                1 & 0_K\transp
            \end{bmatrix} J(Z)^{-1}U =  \begin{bmatrix}
                \pi(Z)\\
                Z_1
            \end{bmatrix}E_1~,\\
            [\breve{D}_{\gamma, 1}b_1](W) &= b_1(Z)\transp \begin{bmatrix}
                0_K & I_K
            \end{bmatrix} J(Z)^{-1}U = b_1(Z)\transp E_2~.
        \end{aligned}, \quad E \define J(Z)^{-1}U.    
    \end{equation*}
    The projection of $\breve{l}(W)$ onto the orthocomplement of $\{\breve{D}_{1}b_1: b_1\in B_{1}\}$ is equal to 
      $\tilde{l}(W)\define \Pi\left[
        \dot{l}(W) | \{[D_1b_1](W) + D_2b_2(W) : b\in B\}^\perp
    \right]$ by Proposition A.2.4 in \cite{BKRW98}.
    The components of $\left[\pi(Z)\transp, Z_1\transp\right]\transp\E[E_1E_2\transp|Z]\E[E_2E_2\transp | Z]^{-1}E_2$ belong to $\cl \{b_1(Z)\transp E_2 : b_1\in B_{1}\}$ as $B_1$ is dense in $L_2$
    and by iterated expectations
    \begin{equation*}
        \E\left[\tilde{l}_{\gamma}(W) b_1(Z)\transp E_2\right] = \E\left[\left[\pi(Z)\transp, Z_1\transp\right]\transp\left[E_1 - \E[E_1E_2\transp|Z]\E[E_2E_2\transp|Z]^{-1}E_2\right]E_2\transp b_1(Z) \right] = 0.
    \end{equation*}
    Hence $\tilde{l}(W)$, the efficient score for $(\theta, \beta)$, has the form
    \begin{equation}
           \tilde{l}(W) = \begin{bmatrix}
            \tilde{l}_{1}(W)\\
            \tilde{l}_{2}(W)
           \end{bmatrix} = \begin{bmatrix}
            \pi(Z)\\
            Z_1
           \end{bmatrix}\left[E_1 - \E[E_1E_2\transp|Z]\E[E_2E_2\transp|Z]^{-1}E_2\right].
    \end{equation}
    $\tilde{\ell}(W)= \tilde{l}_{1}(W) - \E[\tilde{l}_{1}(W)\tilde{l}_{2}(W)\transp]\E[\tilde{l}_{2}(W)\tilde{l}_{2}(W)\transp]^{-1}\tilde{l}_{2}(W)$ by Example A.2.1 in \cite{BKRW98}. To calculate this note that with $Q(Z)\define J(Z)^{-1} = \E[EE\transp|Z]$
    \begin{equation*}
        E_1 - Q(Z)_{1, 2}Q(Z)_{2, 2}^{-1}E_2 = \left[Q(Z)_{1,1}- Q(Z)_{1, 2}Q(Z)_{2, 2}^{-1}Q(Z)_{2,1}\right]U_1 = \E[\epsilon^2|Z]^{-1}\epsilon
    \end{equation*}
    from which the result follows by direct calculation.
\end{proof}

\begin{proof}[Proof of Lemma \ref{lem:iv-bar-effscr}]
    The second claim follows from the expressions in \eqref{eq:iv-effscr} \& \eqref{eq:iv-bar-effscr}. For the first, by Assumption \ref{ass:iv-mdl} and \eqref{eq:iv-EphiU}, $\E[\|g(W)\|^2]<\infty$. By \eqref{eq:iv-b2-conditions}
    \begin{equation*}
        \E\left[ g(W) b_2(\xi)\right] = \E[\epsilon^2]^{-1}\E\left[ (\pi(Z) - MZ_1) \E\left[\epsilon b_2(U, Z)\middle|Z\right]\right] = 0, 
    \end{equation*}
    where $M\define \E[XZ_1\transp]\E[Z_1Z_1\transp]^{-1}$ and $\upxi = (Y - X\transp\theta - Z_1\transp\beta, X - \pi(Z),Z)$. By \eqref{eq:iv-EphiU}
    \begin{equation*}
        \E\left[g(W) \phi_2(\xi)\transp b_1(Z)\right]
        = \E[\epsilon^2]^{-1}\E\left[(\pi(Z) - MZ_1) \E\left[\epsilon \phi_2(\epsilon, \upsilon, Z)\transp \middle|Z \right]b_1(Z)\right] = 0.
    \end{equation*}
    Lastly also $\E\left[g(W)\phi_1(\xi)b_0\transp Z_1
    \right] =0$ as  by \eqref{eq:iv-EphiU} and $\E[\upsilon|Z] =0$,
    \begin{equation*}
        \E\left[g(W)\phi_1(\xi)b_0\transp Z_1
        \right] 
            = -\E[\epsilon^2]^{-1} \left[\E[
                \pi(Z)Z_1\transp] - \E[\pi(Z)Z_1\transp]\E[Z_1Z_1\transp]^{-1} \E[Z_1Z_1\transp ] \right]b_0. 
    \qedhere
    \end{equation*}
    
\end{proof}

\begin{proof}[Proof of Proposition \ref{prop:iv-joint-conv}]
    Assumptions \ref{ass:iv-mdl}, \ref{ass:iv-LAN},  equation \eqref{eq:iv-EphiU} and Lemma \ref{lem:iv-bar-effscr} verify the conditions required to apply Lemma \ref{lem:iid-orthocomp-joint-conv}.
\end{proof}

\begin{proof}[Proof of Proposition \ref{prop:iv-estimation-bar-effscr2}]
    Let $\beta_n = \beta + b_{n, 0}/\sqrt{n}$ with $b_{n, 0}\to b_0\in \R^{d_\beta}$. Let $\breve{\epsilon}_{n, i}$, $\breve{g}_{n, \theta, i}$, $\breve{V}_{n, \theta}$, $\breve{\Lambda}_{n, \theta}$ and $\breve{r}_{n, \theta}$ be formed analogously to $\hat{\epsilon}_{n, i}$, $\hat{g}_{n, \theta, i}$, $\hat{V}_{n, \theta}$, $\hat{\Lambda}_{n, \theta}$ and $\hat{r}_{n, \theta}$  with $\beta_n$ in place of $\hat\beta_n$. As $\hat\beta_n\in \ms{S}_n$, by Lemma S3.1 in \cite{HLM22-S} it suffices to show that Assumption \ref{ass:consistent} holds for $\breve{g}_{n, \theta}\define \fracrootn\sumin \breve{g}_{n, \theta, i}$, $\breve{\Lambda}_{n, \theta}$ and $\breve{r}_{n, \theta}$.
    For Assumption \ref{ass:consistent} part \ref{ass:consistent:itm:ghat}, by Lemma \ref{lem:iv-estimation-bar-effscr-details2}, $ \fracrootn \sumin \left[g_{n, \theta, i} -  {g}(W_i)\right] = \sum_{l=1}^4 \mathsf{R}_{n, l} = o_{P_{n, 0}}(1)$.
    For Assumption \ref{ass:consistent} parts \ref{ass:consistent:itm:lambdahat} and \ref{ass:consistent:itm:rank}, note that 
    \begin{equation}\label{lem:iv-estimation-bar-effscr2:eq-rate}
        \check{V}_{n, \theta} - V  = \meanin \left[g_{n, \theta, i}g_{n, \theta, i}\transp - g(W_i)g(W_i)\transp\right] + \fracrootn \G_n gg\transp.
    \end{equation}
    For the first right hand side term, by Cauchy --- Schwarz,
    \begin{equation*}
       \left\| \meanin \left[g_{n, \theta, i}g_{n, \theta, i}\transp -  g(W_i)g(W_i)\transp\right]\right\| \lesssim \left[\sum_{l=1}^4 \mathsf{S}_{n, l}\right]\left[\meanin \|g_{n, \theta, i}\|^2 + \meanin \|g(W_i)\|^2\right].
    \end{equation*}
    As $\E\|g(W_i)\|^2 < \infty$ under Assumption \ref{ass:iv-mdl} , $\meanin \|g(W_i)\|^2 = O_{P_{n, 0}}(1)$.
    By Lemma \ref{lem:iv-estimation-bar-effscr-details2} $\meanin \|g_{n, \theta, i} - g(W_i)\|^2\lesssim  \sum_{l=1}^4 \mathsf{S}_{n, l} = O_{P_{n,0}}(\delta_n^2 + n^{-1})$,
    hence by the preceding display, the first RHS term in \eqref{lem:iv-estimation-bar-effscr2:eq-rate} is $O_{P_{n, 0}}(\delta_n^2 + n^{-1})$.    
    By Assumption \ref{ass:iv-est2} 
    $\E\|g(W)\|^4 <\infty$. Hence $\G_n g_{\gamma}g_{\gamma}\transp = O_{P_{n, 0}}(1)$ by the CLT and so the second RHS term in \eqref{lem:iv-estimation-bar-effscr2:eq-rate} is $O_{P_{n, 0}}(n^{-1/2})$. 
    The result now follows by the condition on $\upnu_n$ and Proposition S1 of \cite{LM21-S}.
\end{proof}

\begin{lemma}\label{lem:iv-estimation-bar-effscr-details2}
    In the setting of Proposition \ref{prop:iv-estimation-bar-effscr2}, with $\breve{s}_n\define \meanin \breve{\epsilon}_{n, i}^2$,
    \begin{enumerate}
        \item $\|\hat{M_n} - M\| = O_{P_{n, 0}}(n^{-1/2})$ where $\hat{M}_n\define \left[\meanin X_i Z_{1, i}\transp\right] \left[\meanin Z_{1,i} Z_{1, i}\transp\right]^{-1}$;
        \item $\meanin |\breve{\epsilon}_{n, i} - \epsilon_{n, i}|^2 = O_{P_{n, 0}}(n^{-1})$ and $|\breve{s}_{n}^{-1} - \E[\epsilon^2]^{-1}| = O_{P_{n, 0}}(n^{-1/2})$;
        \item  $\meanin \|\tilde{\pi}_{n, i}(Z_i)\|^2x_i^2 = O_{P_{n, 0}}(\delta_n^2)$ for $x_i\in \{1, \epsilon_i\}$;
        \item $\mathsf{R}_{n, 1} = \fracrootn\sumin \breve{s}_n^{-1}\breve{\epsilon}_{n, i}[M - \hat{M}_n]Z_{1, i} = o_{P_{n, 0}}(1)$;
        \item  $\mathsf{R}_{n, 2} = \fracrootn\sumin \breve{s}_n^{-1}\breve{\epsilon}_{n, i}\left[\hat{\pi}_{n, i}(Z_i) - \pi(Z_i)\right] = o_{P_{n, 0}}(1)$;
        \item $\mathsf{R}_{n, 3} = \fracrootn\sumin \breve{s}_n^{-1}(\breve{\epsilon}_{n, i} -\epsilon_i)f(Z_i) = o_{P_{n, 0}}(1)$, where $f(Z_i)\define \pi(Z_i) - MZ_{1, i}$;
        \item $\mathsf{R}_{n, 4} = \fracrootn\sumin (\breve{s}_n^{-1} - \E[\epsilon^2]^{-1})\epsilon_i f(Z_i) = o_{P_{n, 0}}(1)$;
        \item $\mathsf{S}_{n, 1} = \meanin \|\breve{s}_n^{-1}\breve{\epsilon}_{n, i}\left[M - \hat{M}_n\right]Z_{1, i}\|^2 = O_{P_{n, 0}}(n^{-1})$;
        \item  $\mathsf{S}_{n, 2} = \meanin \|\breve{s}_n^{-1}\breve{\epsilon}_{n, i}\left[\hat{\pi}_{n, i}(Z_i) - \pi(Z_i)\right]\|^2 = O_{P_{n, 0}}(\delta_n^2)$;
        \item $\mathsf{S}_{n, 3} =  \meanin \|\breve{s}_n^{-1}(\breve{\epsilon}_{n, i} - \epsilon_i)f(Z_i)\|^2 = O_{P_{n, 0}}(n^{-1})$;
        \item $\mathsf{S}_{n, 4} =  \meanin \|(\breve{s}_n^{-1} - \E[\epsilon^2]^{-1})\epsilon_if(Z_i)\|^2 = O_{P_{n, 0}}(n^{-1})$.
    \end{enumerate}
\end{lemma}
\begin{proof}
    Let $\tilde{\pi}_{n, i}(Z_i)\define \hat{\pi}_{n, i}(Z_i) - \pi(Z_i)$. By a simplification of the argument in Proposition \ref{prop:iv-LAN}, \eqref{prop:iv-LAN:eq:DQM} holds for $t\mapsto \gamma + t(0, (b_0, 0, 0))$. Then $P_{n, 0} \mcontig P_{n, (0, b_{0, n}, 0, 0)}\coloneqq Q_n$ by Example 6.5, Theorem 7.2 \& Lemma 7.6 in \cite{vdV98}.
    \begin{enumerate}
        \item Follows from the CLT, given the  moment conditions in Assumption \ref{ass:iv-mdl}.
        \item The first holds by standard arguments as $\beta_n - \beta = O(n^{-1/2})$ and $\E \|Z_i\|^2 < \infty$; the second by the CLT and delta method.
        \item  As $P_{n, 0}(\meanin \|\tilde{\pi}_{n, i}(Z_i)\|^2x_i^2> K \delta_n^2)$ is bounded by $P_{n, 0}(\bm{1}_{F_n} \meanin \|\tilde{\pi}_{n, i}(Z_i)\|^2x_i^2 > K\delta_n^2) + P_{n, \gamma}F_n^\complement$. By Markov's inequality $\meanin \|\tilde{\pi}_{n, i}(Z_i)\|^2x_i^2 = O_{P_{n, 0}}(\delta_n^2)$ as by \eqref{eq:iv-EphiU}, \eqref{ass:iv-est2:eq:pi-est-bound} $\E[\bm{1}_{F_{n}}\|\tilde{\pi}_{n, i}(Z_i)\|^2x_i^2]\le \E[\E[\bm{1}_{F_{n, i}}\|\tilde{\pi}_{n, i}(Z_i)\|^2x_i^2 | Z_i, \mc{C}_{n, -i}]] \lesssim \delta_n^2$ , where $F_{n,i}$ is the $\sigma(\mc{C}_{n,-i})$ -- measurable set on which \eqref{ass:iv-est2:eq:pi-est-bound} holds for index $i$.
        \item $\mathsf{R}_{n, 1}\transp = \breve{s}_n^{-1} \left[\meanin Z_{1, i}\transp \epsilon_i + \meanin Z_{1, i}\transp (\breve{\epsilon}_{n, i} - \epsilon_i)\right] \sqrt{n}[M - \hat{M}_n]\transp$.
        By (i) and (ii) it suffices to note
        $\meanin Z_{1, i}\transp \epsilon_i = o_{P_{n, 0}}(1)$ by the WLLN and \linebreak $\meanin Z_{1, i}\transp (\breve{\epsilon}_{n, i} - \epsilon_i) = o_{P_{n, 0}}(1)$ by $\E \|Z_i\|^2<\infty$, (ii) and Cauchy -- Schwarz.
        \item $\mathsf{R}_{n, 2} = \breve{s}_n^{-1}\sqrt{n}(\beta - \beta_n)\transp \meanin Z_{1, i} \tilde{\pi}_{n, i}(Z_i) + \breve{s}_n^{-1}\fracrootn\sumin \epsilon_{i}\tilde{\pi}_{n, i}(Z_i)$. The first RHS term is $o_{P_{n, 0}}(1)$ by (iii) and $\E Z_{1, i}^2 < \infty$. For the second, by $\breve{s}_n^{-1} = O_{P_{n, 0}}(1)$, Assumption \ref{ass:iv-est2} and Markov's inequality it suffices to observe that $\E\left[\meanin \bm{1}_{F_n} \tilde{\pi}_{n, i, k}(Z_i)^2 \epsilon_{i}^2\right]\lesssim \delta_n^2$ by the argument in (iii) and by \eqref{ass:iv-est2:eq:pi-est-bound2}, \\  $\E\left[\meanin \sum_{j=1, j\neq i}^n  \bm{1}_{F_n}\bm{1}_{G_n}\tilde{\pi}_{n, i, k}(Z_i)\tilde{\pi}_{n, j, k}(Z_j)\transp \epsilon_i\epsilon_j            
        \right] \lesssim \delta_n^2 \to 0.$
        
        \item $\mathsf{R}_{n, 3} =\breve{s}_n^{-1}[\meanin f(Z_i)Z_{1, i}\transp]\sqrt{n}(\beta - \beta_n)$, where the bracketed term is $o_{P_{n, 0}}(1)$ by the WLLN as $ \E\left[f(Z)Z_{1}\transp\right] = \E\left[\pi(Z)Z_1\transp - \E[\pi(Z)Z_1\transp]\E[Z_1Z_1\transp]^{-1}Z_1Z_1\transp\right] = 0$ and the remaning factors are $O_{P_{n, 0}}(1)$ by (ii) and $\beta - \beta_n = O(n^{-1/2})$. 
        \item As $\E[\epsilon f(Z)] = \E[\E[\epsilon|Z]f(Z)] = 0$ this follows from (ii) and the WLLN.
        \item $\mathsf{S}_{n, 1} \lesssim s_n^{-2}\|M - M_n\|^2 [\meanin \epsilon_i^2 \|Z_{1, i}\|^2 + 2\|\beta - \beta_n\| |\epsilon_i| \|Z_{1, i}\|^3 + \|\beta - \beta_n\|^2\|Z_{1, i}\|^4]$ hence this follows by (i), (ii) and the moment conditions in Assumption \ref{ass:iv-mdl}.
        \item
        By contiguity $Q_n F_n\to 1$. As the (conditional) distribution of $(\breve{\epsilon}_{n, i}, Z_i) | C_{n, -i}$ under $Q_n$ as that of $(\epsilon_i, Z_i) | C_{n, -i}$, under $Q_n$, 
        $\E[\breve{\epsilon}_{n, i}^2 |Z_i, \mc{C}_{n, -i}] \le C$ a.s. by \eqref{eq:iv-EphiU}. Therefore, under $Q_n$, $  \E\left[ \bm{1}_{F_{n}} \meanin \breve{\epsilon}_{n, i}^2\|\tilde{\pi}_{n, i}(Z_i)\|^2 \right] \lesssim \delta_n^2$, similar to in (iii),
        and hence Markov's inequality implies $\meanin \breve{\epsilon}_{n, i}^2\|\tilde{\pi}_{n, i}(Z_i)\|^2 = O_{Q_n}(\delta_n^2)$. By contiguity this holds also under $P_{n, 0}$.
        \item As $\mathsf{S}_{n, 2} \le \breve{s}_{n}^{-2}\|\beta - \beta_n\|^2 \meanin \|Z_{1, i}\|^2  [\|\pi(Z_i)\|^2 + \|M\|^2\|Z_{1, i}\|^2]$, the result holds by (ii), $\beta - \beta_n = O(n^{-1/2})$ \& the moment conditions in Assumption \ref{ass:iv-mdl}.
        \item Since $\mathsf{S}_{n, 4} \le (\breve{s}_n^{-1} - \E[\epsilon_i^2])^2\meanin \epsilon_i^2(\|\pi(Z_i)\|^2 + \|M\|^2\|Z_{1, i}\|^2)$, this follows from (ii) and the  moment conditions in Assumption \ref{ass:iv-mdl}.\qedhere
    \end{enumerate}
\end{proof}

\paragraph*{Condition \eqref{ass:iv-est2:eq:pi-est-bound2}} The condition in equation \eqref{ass:iv-est2:eq:pi-est-bound2} is natural when $\hat{\pi}_{n, i}$ is a leave-one-out series estimator: $\hat{\pi}_{n, i}(Z_i)\define \hat\upalpha_{n, i}\transp p_{K_n}(Z_i)$ for $p_{K_{n}}(Z_i)$ a $K_n$-vector of functions of $Z_i$ and $\hat\upalpha_{n, i} = Q_{n, i}^{-1}\frac{1}{n-1}\sum_{j=1, j\neq i}^n p_{K_n}(Z_i)X_i\transp$ with $Q_{n, i}\define \left[\frac{1}{n-1}\sum_{j=1, j\neq i}^n p_{K_n}(Z_i)p_{K_n}(Z_i)\transp\right]$. Then, with $\tilde{\pi}_{n, i}(Z_i)\define \hat{\pi}_{n, i}(Z_i) - \pi(Z_i)$ and $G_n \in \sigma(Z_1, \ldots, Z_n)$,
\begin{equation}\label{eq:series-eq-condition1}
    \begin{aligned}
        \E\left[\bm{1}_{G_n}\tilde{\pi}_{n, i, k}(Z_i)\tilde{\pi}_{n, j, k}(Z_j)\epsilon_i\epsilon_j\right] 
        &= -\E\left[\bm{1}_{G_n}\tilde{\pi}_{n, i, k}(Z_i) \epsilon_j \E\left[ p_{K_n}(Z_j)\transp  \hat\upalpha_{n, j}e_k\epsilon_i  \middle|Z_i, \mc{C}_{n, -i}\right]\right].
    \end{aligned}
\end{equation}
as by $\E[\epsilon_i | Z_i] = 0$ and independence $\E[\epsilon_i \pi(Z_i) | Z_i, \mc{C}_{n, -i}]=0$. The RHS of \eqref{eq:series-eq-condition1} is 
\begin{equation*}
    -\E\left[\bm{1}_{G_n} \tilde{\pi}_{n, i, k}(Z_i) \epsilon_j p_{K_n}(Z_j)\transp Q_{n, j}^{-1} \frac{1}{n-1}\sum_{l=1, l\neq j}^n p_{K_n}(Z_l) \E\left[  X_l\transp \epsilon_i  \middle|Z_i, \mc{C}_{n, -i}\right]e_k\right],
\end{equation*}
and as $\E\left[  X_l\transp\epsilon_i  \middle|Z_i, \mc{C}_{n, -i}\right] =0$ if $l\neq i$, therefore with $\mu(Z_i)\define \E[\upsilon_i\transp \epsilon_i|Z_i]e_k$,
\begin{align*}
    &\E\left[\bm{1}_{G_n}\tilde{\pi}_{n, i, k}(Z_i)\tilde{\pi}_{n, j, k}(Z_j)\epsilon_i\epsilon_j\right] = -\E\left[\bm{1}_{G_n}\tilde{\pi}_{n, i, k}(Z_i) \epsilon_j p_{K_n}(Z_j)\transp Q_{n, j}^{-1} \frac{1}{n-1} p_{K_n}(Z_i) \mu(Z_i)\right]\\
    &\qquad = -\frac{1}{n-1}\E\left[\bm{1}_{G_n}e_k\transp p_{K_n}(Z_j)\transp Q_{n, j}^{-1}  p_{K_n}(Z_i) \mu(Z_i) \E\left[\tilde{\pi}_{n, i, k}(Z_i) \epsilon_j  | Z_j, \mc{C}_{n, -j}\right]\right].
\end{align*}
Arguing as before with the roles of $i$ and $j$ interchanged and using \eqref{eq:iv-EphiU} yields
\begin{align*}
    |\E\left[\bm{1}_{G_n}\tilde{\pi}_{n, i}(Z_i)\tilde{\pi}_{n, j}(Z_j)\epsilon_i\epsilon_j\right]| 
    &\le C^2 \frac{|\E\left[\bm{1}_{G_n} p_{K_n}(Z_j)\transp Q_{n, j}^{-1}  p_{K_n}(Z_i) p_{K_n}(Z_i)\transp Q_{n, i}^{-1}  p_{K_n}(Z_j)    \right]|}{(n-1)^2}.
\end{align*}
Therefore \eqref{ass:iv-est2:eq:pi-est-bound2} holds if the RHS is bounded above by a constant multiple of $\delta_n^2 / n$.

\subsubsection{The LAN condition}\label{app:ssec:iv-LAN}

Here I provide examples of $P_{n, h}$ and lower level conditions under which the LAN condition in Assumption \ref{ass:iv-LAN} holds. Let $\varphi_n$ be as in \eqref{eq:iv-varphi} with
\begin{equation}\label{app:eq:iv-varphik}
    \varphi_{n, 1}(b_1)\define b_1 / \sqrt{n},\quad  \varphi_{n, 2}(b_2)\define \zeta b_2 / \sqrt{n}, \qquad (b_1, b_2)\in B_{1} \times B_{2},
\end{equation}
where $B_{1}$ is the space of bounded functions $b_1:\R^{d_Z}\to \R^{d_\theta}$ and $B_{2}$ the space of bounded functions $b_2:\R^{d_w}\to \R$ which are continuously differentiable in their first $1 + d_\theta$ components with bounded derivative and such that \eqref{eq:iv-b2-conditions} hold.

\begin{proposition}\label{prop:iv-LAN}
    If Assumption \ref{ass:iv-mdl} holds, $\mc{W}_n = \prod_{i=1}^n \R^{d_w}$, $u\mapsto \sqrt{\zeta(u, z)} \in \mc{C}^1$
      and $p_{n,h} = p_{\gamma + \varphi_n(h)}^n$ with $p_\gamma$ as in \eqref{eq:iv-dens}. Then Assumption \ref{ass:iv-LAN} holds.
\end{proposition}
\begin{proof}
   For all large enough $n$ each $\gamma + \varphi_n(h)\in \Gamma$. Assumption \ref{ass:iid} is satisfied by construction; to apply Lemma \ref{lem:iid-LAN-DQM} it remains to verify \eqref{eq:dqm} (with $h_n = h$).
    Let $q_{\tau, b, t}\define p_{(\theta, \eta) + t(\tau, (b_0, b_1, b_2\zeta))}$, $t\in[0, \infty)$ and let $q\define q_{0,0,0}$. For all small enough $\tau, b$ and $t$, $\gamma + t(\tau, (b_0, b_1, b_2\zeta))\in \Gamma$. It suffices to show
    \begin{equation}\label{prop:iv-LAN:eq:DQM}
        \int \left[q_{\tau, b, t}^{1/2} - q^{1/2} - \frac{t}{2}\left((\tau\transp, b_0\transp) \dot{l}- \phi\transp b_1 + b_2\right)q^{1/2}\right]^2 \darg{\nu} = o(t^2) \quad \text{ as }t \downarrow 0,
    \end{equation}
    where $ \dot{l}(W)\define -\phi(\epsilon(\theta, \beta), \upsilon(\pi),Z)[X\transp, Z_1\transp]\transp$.
    Note that $t\mapsto \sqrt{q_{\tau, b, t}} \in \mc{C}^1$ follows from $(e, v)\mapsto \sqrt{\zeta(e, v, z)} \in \mc{C}^1$. 
     Under $q_{\tau, b, s}$,
     $\pderiv{\log q_{ \tau, b, t}}{t}|_{t = s}$
      has the same law as 
    \begin{align*}
        E_s&\define -\phi_1(\epsilon, \upsilon, Z)[X\transp, Z_1\transp](\tau\transp, b_0\transp)\transp -\phi_2(\epsilon, \upsilon, Z)\transp b_1(Z)\\
        &\qquad  + \frac{b_2(\epsilon, \upsilon, Z) - sb_{2, 1}(\epsilon, \upsilon, Z)[X\transp, Z_1\transp](\tau\transp, b_0\transp) - sb_{2, 2}(\epsilon, \upsilon, Z)\transp b_1(Z)}{1 + sb_2(\epsilon, \upsilon, Z)}~, 
    \end{align*}
    where $b_{2, i}$ indicates the derivative of $(e, v)\mapsto b_2(e, v, z)$ in the $i$-th argument. Take a neighbourhood of 0, $\mc{U}\define [0, \delta)$ such that $1+sb_2(\epsilon, \upsilon, Z)$ is bounded below. 
     Let $ \overline{E^2} \define C\left[\phi_1(\epsilon, \upsilon, Z)^2[ \|X\|^2 + \|Z\|^2]+ \|\phi_2(\epsilon, \upsilon, Z)\|^2 + \|X\|^2 + \|Z\|^2 + 1 \right]$
    for some positive constant $C$. Provided $C$ is large enough, by Assumption \ref{ass:iv-mdl} 
    $E_s^2 \le \overline{E^2}$ a.s. and $\E \overline{E^2} < \infty$. 
    Therefore, as $E_{s_n}^2\to E_{s}^2$ pointwise, $\E E_{s_n}^2 \to \E E_s^2$, which verifies that Lemma 7.6 in \cite{vdV98} applies, whence \eqref{prop:iv-LAN:eq:DQM} holds.
\end{proof}

\section{Additional simulation details \& results}\label{sm:sec:extra-sim-results}

\subsection{Single index model}\label{sm:ssec:sim-extra-sim-results}

As discussed in Section \ref{sm:shape-constraints}, locally regular C($\alpha$) tests do not exhibit size distortions when nuisance parameters are estimated under shape constraints. Here I explore this in simulation, using Example \ref{ex:SIM-running-example} with $\ms{F}$ restricted to contain only monotonically increasing functions. I set $H_0:\theta = \theta_0=0$ and consider three possible link functions: $f_1$ is a logistic function, whilst $f_2$ and $f_3$ are double logistic functions which include a flat section in between two increasing sections. These functions are formally defined in \eqref{sm:eq:dbllogis} below and plotted in Figure \ref{fig:SIM-dbllogistic-fs}. Each considered link function has flat sections which may cause monotonicity constraints to bind in the estimation of $f$. I explore the effect this has on the rejection frequencies of the $\psi_{n, \theta_0}$ test as described on p. \pageref{para:sim:ghat} and an \cite{I93} -- style Wald test. Both tests are computed with $f, f'$ estimated by 9 monotonic I -- splines \cite[e.g.][]{R88}, whilst $Z_1$ is estimated using 6 cubic B -- splines. As $\effinfo > 0$ in this design, $\upnu =0$. $\epsilon$ is drawn from a standard normal and the covariates are drawn as $X= (Z_1, 0.2 Z_1 + 0.4 Z_2 + 0.8)$, where each $Z_k\sim U(-1.5, 1.5)$ is independent.

The $f_j$ functions used are as follows.  Let $b(x)\define \bm{1}\{x >0 \}\exp(-1/x)$ (a bump function) and form the smooth transition function $a(x)\define b(x) / (b(x) + b(1-x))$. Then with $g(v; a, b)\define 1 / (1 + \exp(- (x - b)/ a))$, a logistic function, let
\begin{equation}\label{sm:eq:dbllogis}
	\begin{aligned}
		f_1(v) &\define 8g(v, 0.25, 0)~;\\
		f_2(v) &\define 4\big[\bm{1}\{4v\le -1\}g(4v, 0.4, -3)  + \bm{1}\{4v> 1\}(1 + g(4v, 0.4, 3))\\
		&\qquad\quad + \bm{1}\{1- < 4v\le 1\}a((4v + 1)/2) (1 + g(1, 0.4, 3) - g(-1, 0.4, -3))\big]\\
		f_3(v)&\define  4\big[\bm{1}\{3v\le -1\}g(3v, 0.2, -3) + \bm{1}\{3v> 1\}(1 + g(2v, 0.2, 3))~;\\
		&\qquad\quad + \bm{1}\{1- < 3v\le 1\}a((3v + 1)/2) (1 + g(1, 0.2, 3) - g(-1, 0.2, -3))\big]~.\\
	\end{aligned}
\end{equation}

Table \ref{tbl:SIM-size2} displays the empirical rejection frequencies attained by $\psi_{n, \theta_0}$ and the Wald test. The former provides rejection rates close to the nominal level of 5\% in each simulation design considered. The Wald test displays substantial overrejection in each simulation design. 
The 3 panels of Figure \ref{fig:SIM-power-2} depict the finite-sample power curves for $f = f_1, f_2, f_3$ respectively. In each panel, the Wald test shows a relatively slow increase in power as $\theta$ moves away from $\theta_0$ with $\psi_{n, \theta_0}$ providing a much higher rate of increase in power as $\theta$ deviates from the null.\footnote{The power of the Wald test exceeds that of $\psi_{n, \theta_0}$ around the null. However, this is not a like-for-like comparison, as the Wald test over-rejects; see Table \ref{tbl:SIM-size2}.}

\section{Tables and Figures}\label{sm:sec:tbls-figs}

\begin{figure}[htbp]
	\caption{Index functions $f_j(v) = 5\exp(-v^2 / 2c_j^2)$}
	\label{fig:SIM-gaussian-fs}
	\begin{minipage}{.99\textwidth}
		\begin{subfigure}[b]{0.49\textwidth}
			\centering
			\includegraphics[width=\textwidth]{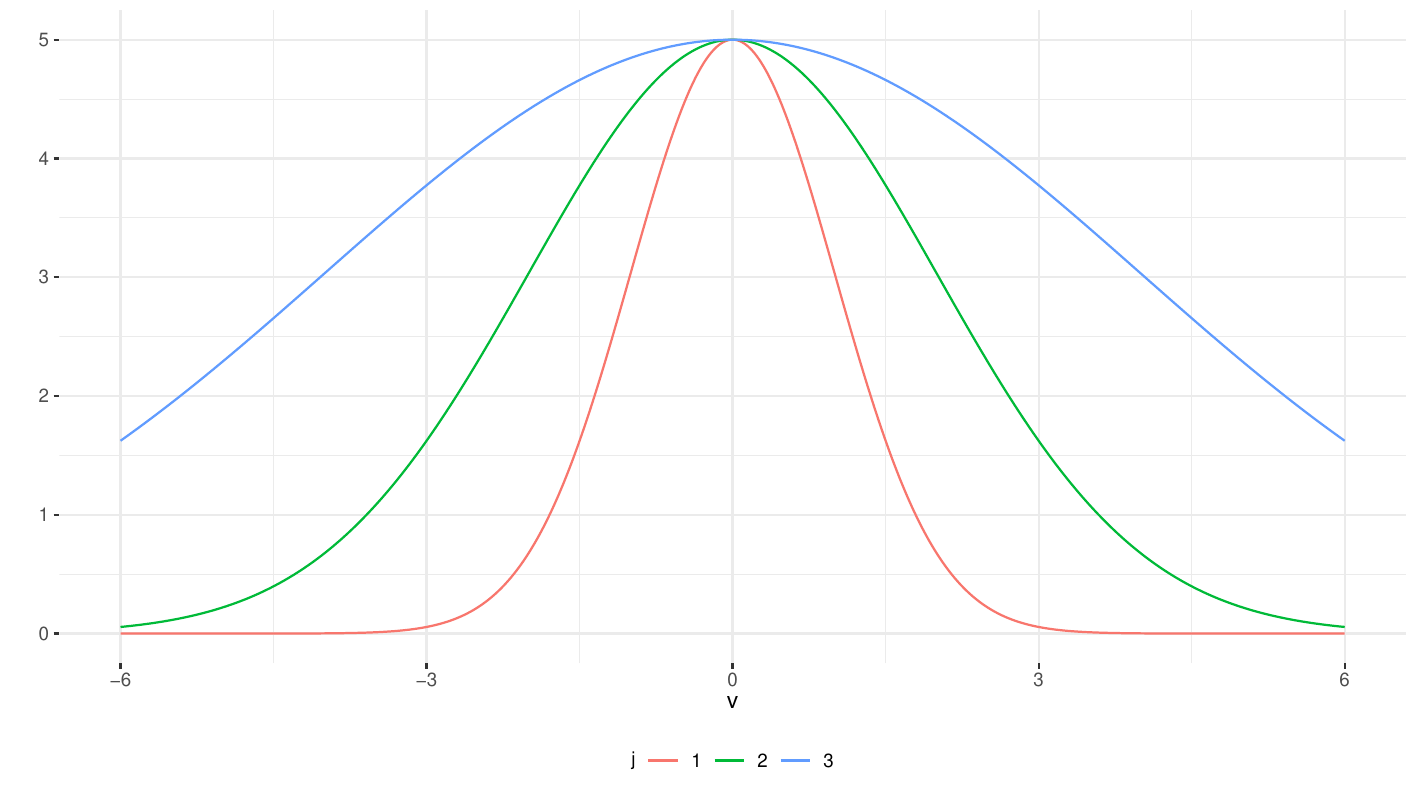}
			\caption[]%
			{{\small $f$}}
		\end{subfigure}
		\hfill
		\begin{subfigure}[b]{0.49\textwidth}
			\centering
			\includegraphics[width=\textwidth]{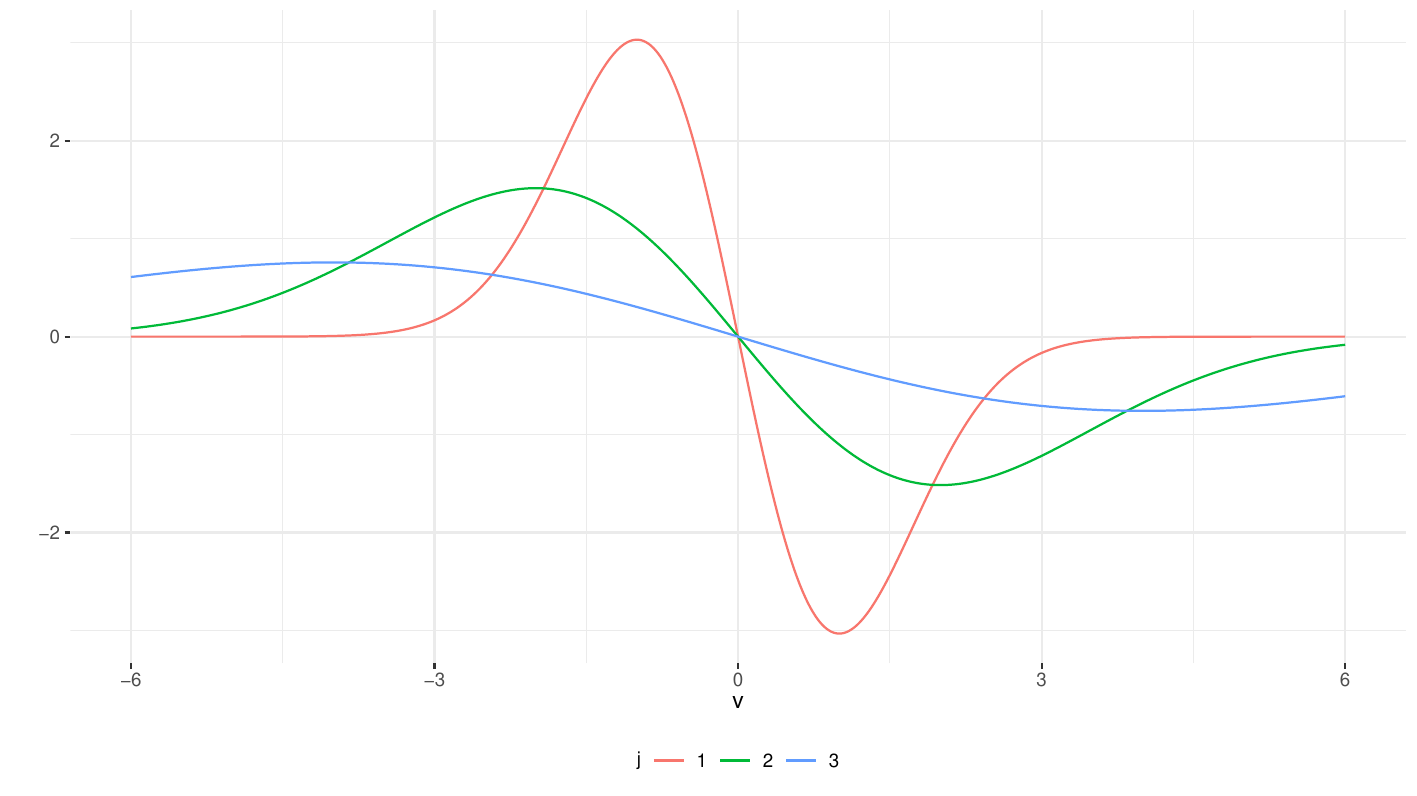}
			\caption[]%
			{{$f'$}}
		\end{subfigure}
		\centering
		\end{minipage}
\end{figure}

\begin{figure}[htbp]
	\caption{Index functions$f_j(v) = 25\left(1 + \exp(-v  /c_j)
	\right)^{-1}$}
	\label{fig:SIM-logistic-fs}
	\begin{minipage}{.99\textwidth}
		\begin{subfigure}[b]{0.49\textwidth}
			\centering
			\includegraphics[width=\textwidth]{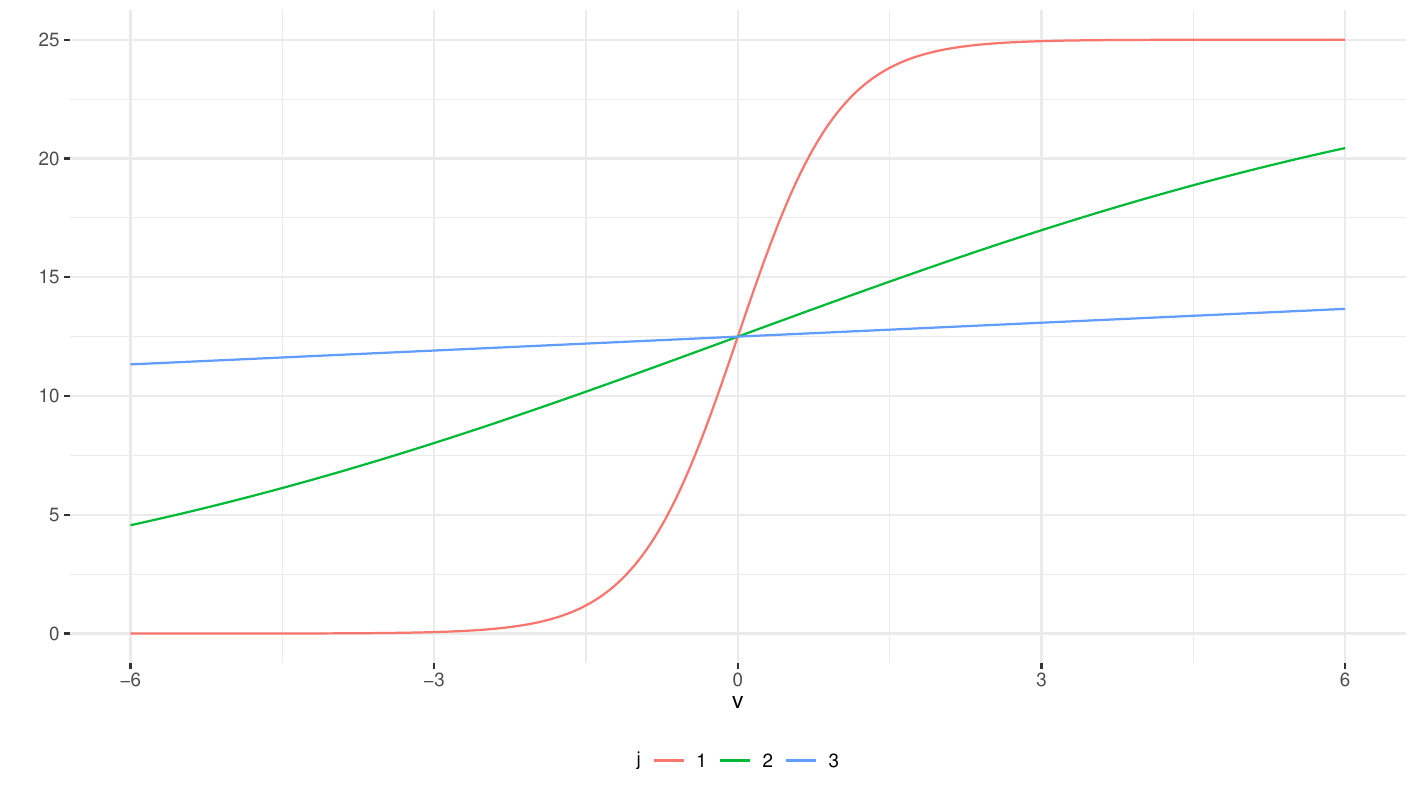}
			\caption[]%
			{{\small $f$}}
		\end{subfigure}
		\hfill
		\begin{subfigure}[b]{0.49\textwidth}
			\centering
			\includegraphics[width=\textwidth]{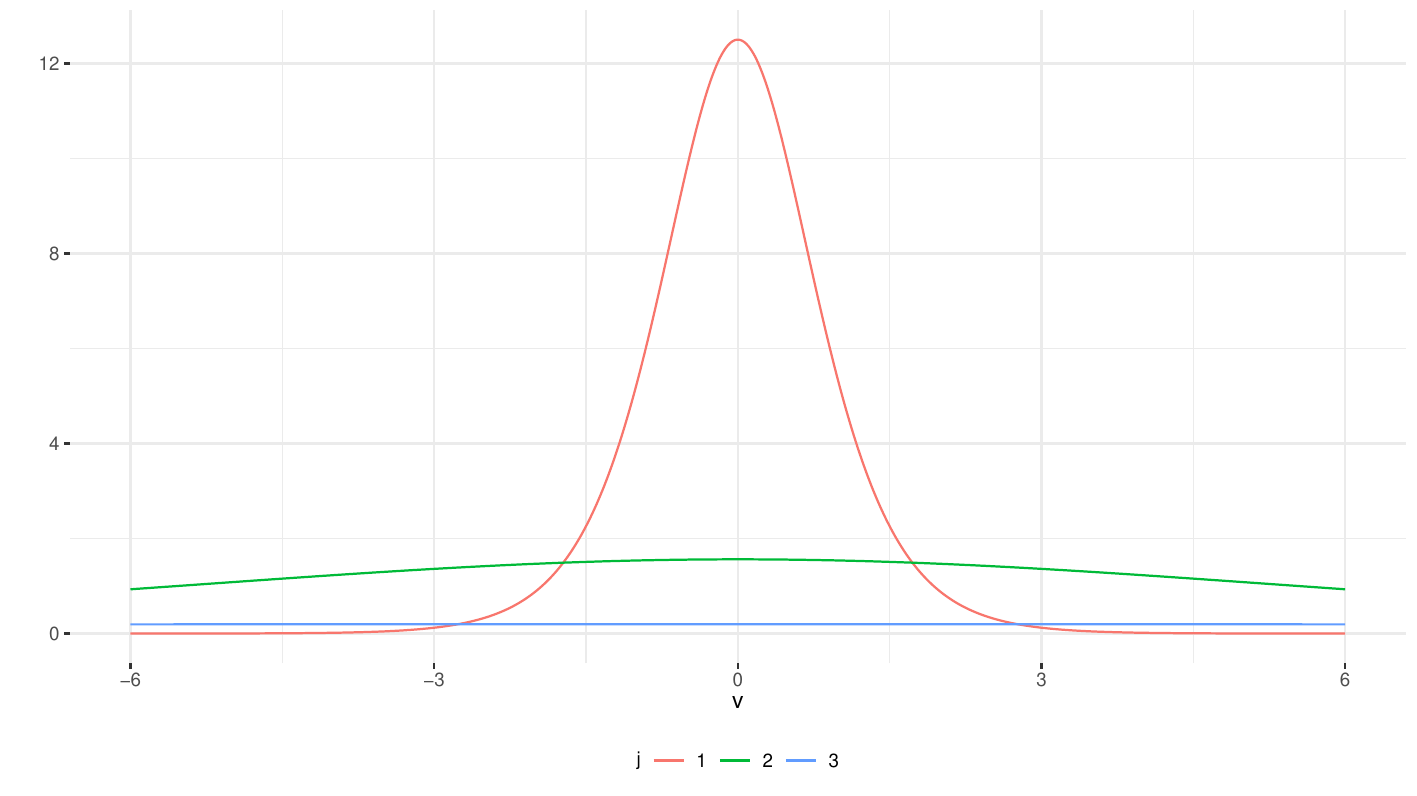}
			\caption[]%
			{{$f'$}}
		\end{subfigure}
		\centering
		\end{minipage}
\end{figure}

\begin{figure}[htbp]
	\caption{ Double logistic index functions as in \eqref{sm:eq:dbllogis}}
	\label{fig:SIM-dbllogistic-fs}
	\begin{minipage}{.99\textwidth}
		\begin{subfigure}[b]{0.49\textwidth}
			\centering
			\includegraphics[width=\textwidth]{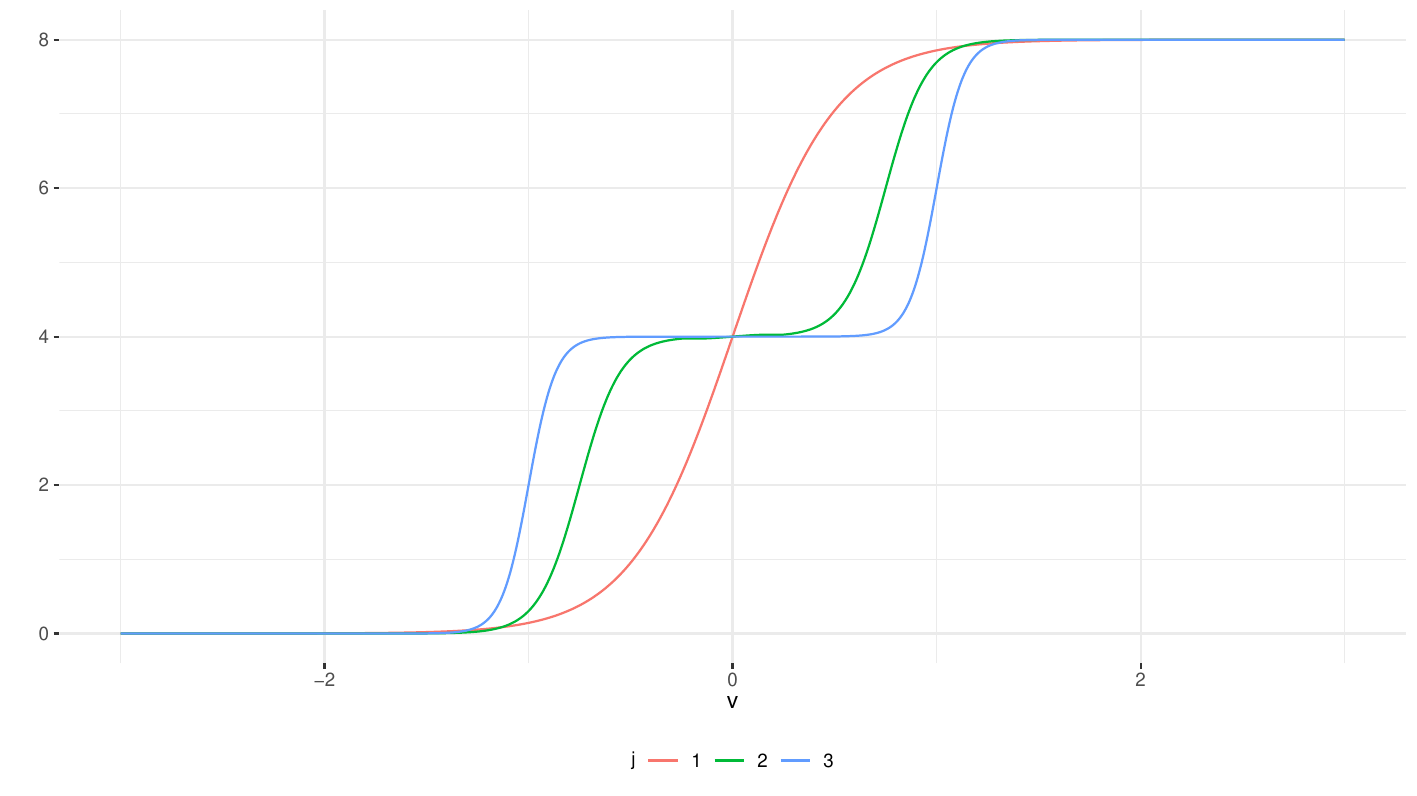}
			\caption[]%
			{{\small $f$}}
		\end{subfigure}
		\hfill
		\begin{subfigure}[b]{0.49\textwidth}
			\centering
			\includegraphics[width=\textwidth]{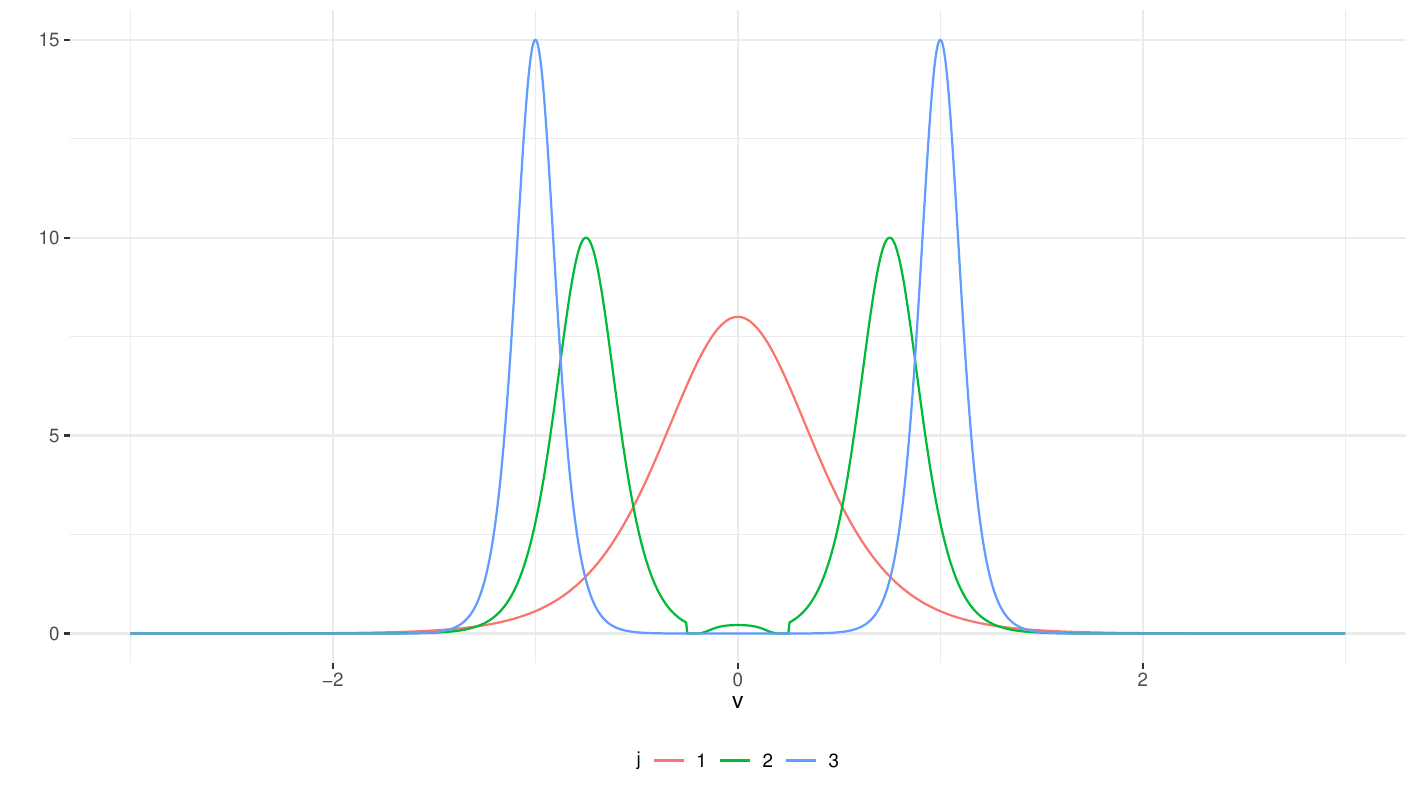}
			\caption[]%
			{{$f'$}}
		\end{subfigure}
		\centering
		\end{minipage}
\end{figure}

\begin{table*}[!htbp]
	\footnotesize
	\begin{center}
		\caption{\label{tbl:SIM-size2} \small ERF (\%), index function as in \eqref{sm:eq:dbllogis}}
		\begin{threeparttable}
			
\begin{tabular}{l*{6}{>{\raggedleft\arraybackslash}p{3em}}l*{6}{>{\raggedleft\arraybackslash}p{3em}}l*{6}{>{\raggedleft\arraybackslash}p{3em}}l*{6}{>{\raggedleft\arraybackslash}p{3em}}l*{6}{>{\raggedleft\arraybackslash}p{3em}}l*{6}{>{\raggedleft\arraybackslash}p{3em}}l*{6}{>{\raggedleft\arraybackslash}p{3em}}}
\toprule
\multicolumn{1}{c}{} & \multicolumn{3}{c}{$\psi_{n, \theta_0}$} & \multicolumn{3}{c}{Wald} \\
\cmidrule(l{3pt}r{3pt}){2-4} \cmidrule(l{3pt}r{3pt}){5-7}
$n$ & $f_1$ & $f_2$ & $f_3$ & $f_1$ & $f_2$ & $f_3$\\
\midrule
400 & 5.72 & 5.38 & 5.78 & 28.58 & 26.06 & 34.22\\
600 & 5.82 & 5.06 & 5.70 & 26.56 & 23.04 & 35.22\\
800 & 5.28 & 5.08 & 5.62 & 23.72 & 19.94 & 32.88\\
\bottomrule
\end{tabular}

		\end{threeparttable}
	\end{center}
\end{table*}

\begin{figure}[htbp]
	\begin{minipage}{.99\textwidth}
		\centering
		\caption{\label{fig:SIM-power-2} \small ERF (\%), index function as in \eqref{sm:eq:dbllogis}}
		\includegraphics[scale = 0.45]{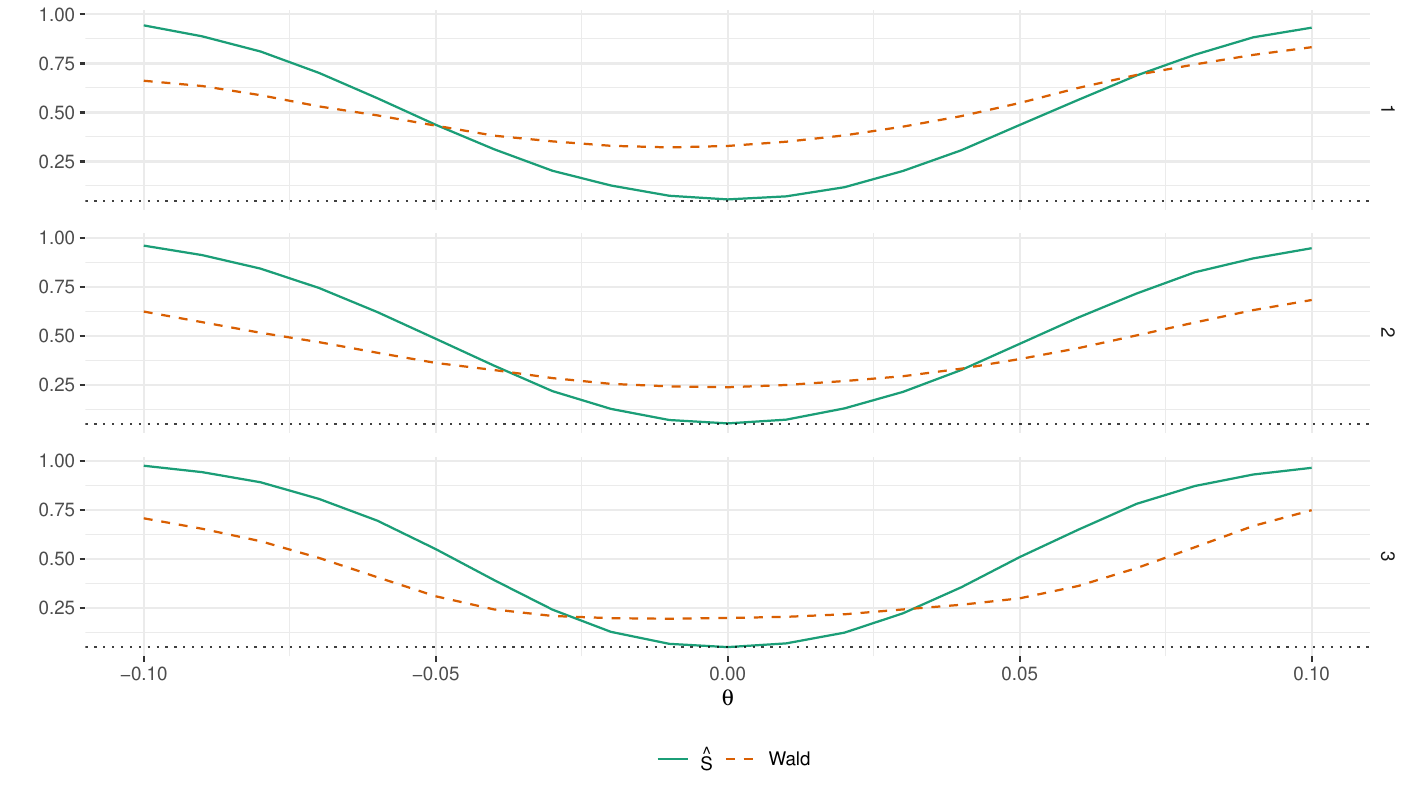}
		\end{minipage}
		\vspace*{-1em}
\end{figure}

\end{document}